\documentclass[10pt, a4paper ]{article}
\usepackage[utf8]{inputenc}
\usepackage[british]{babel} 
\usepackage[big]{layaureo}
\usepackage{lmodern}
\usepackage{amsmath}
\usepackage{amsthm}
\usepackage{amsfonts}
\usepackage{amssymb}
\usepackage{array}
\usepackage{esint}

\usepackage{hyperref}
\hypersetup{colorlinks,breaklinks,urlcolor=blue, linkcolor=black, citecolor=black}

\usepackage[dvipsnames]{color}

\definecolor{mygray}{gray}{0.35}

{\newtheorem{thm}{Theorem}[section]
	\newtheorem{defi}[thm]{Definition}
	\newtheorem{prop}[thm]{Proposition}
	\newtheorem{lemma}[thm]{Lemma}
	\newtheorem{corol}[thm]{Corollary}
	{\theoremstyle{definition}{
			\newtheorem{remark}[thm]{Remark} 
			\newtheorem{example}[thm]{Example} 
			 
}}}

\usepackage{mathrsfs}
\usepackage{amsfonts}
\usepackage{xfrac}

\usepackage{amsthm}
\usepackage{amsmath}
\usepackage{amssymb}

\newsavebox{\fminipagebox}
\NewDocumentEnvironment{fminipage}{m O{\fboxsep}}
{\par\kern#2\noindent\begin{lrbox}{\fminipagebox}
		\begin{minipage}{#1}\ignorespaces}
		{\end{minipage}\end{lrbox}%
	\makebox[#1]{%
		\kern\dimexpr-\fboxsep-\fboxrule\relax
		\fbox{\usebox{\fminipagebox}}%
		\kern\dimexpr-\fboxsep-\fboxrule\relax
	}\par\kern#2
}

\newcommand{\R}{\mathbb{R}}
\newcommand{\Rd}{\R^{d}}

\newcommand{\N}{\mathbb{N}}

\newcommand{\half}{\frac{1}{2}}
\newcommand{\inverse}{^{-1}}

\newcommand{\abs}[1]{\left\lvert {#1} \right\rvert}

\newcommand{\transpose}{^{\mathsf{T}}}
\newcommand{\squared}{^{2}}

\newcommand{\tzero}{t_{0}}

\newcommand{\partition}{\pi}
\newcommand{\timeHorizon}{T}
\newcommand{\timeWindow}{{[0,\timeHorizon]}}

\newcommand{\hurstExponent}{H}
\newcommand{\volterraKernel}{{K}_{\hurstExponent}}
\newcommand{\volterraProcess}{\zeta}
\newcommand{\covarianceKernel}{{R}_{\hurstExponent}}

\newcommand{\zeroSleqTtimeHorizon}{0\leq s\leq t \leq \timeHorizon}
\newcommand{\zerostTimeHorizon}{0\leq s, t \leq \timeHorizon}
\newcommand{\simplex}{\lbrace (s,t) \in \R\squared:\,  \zeroSleqTtimeHorizon\rbrace}
\newcommand{\timesSleqUleqT}{s \leq u \leq t}

\newcommand{\convexHull}{\mathrm{Conv}}

\newcommand{\restrictedto}[1]{\arrowvert_{#1}}

\newcommand{\trace}{\mathtt{trace}}

\newcommand{\derivative}{^{\prime}}

\newcommand{\partialij}{\partial^{2}_{i,j}}
\newcommand{\gradx}{\nabla_{x}}
\newcommand{\gradz}{\nabla_{z}}
\newcommand{\Hessianx}{\nabla^{2}_{xx}}
\newcommand{\Hessianz}{\nabla^{2}_{zz}}

\newcommand{\norm}[1][\cdot]{\left\lVert {#1}\right\rVert}
\newcommand{\tripleNorm}[1]{\vert \vert \vert {#1}\vert \vert \vert }

\newcommand{\pvarNormInterval}[3][\cdot]{\lVert {#1} \rVert_{{#2}\text{-var}, {#3}}}

\newcommand{\Homomorphisms}{\text{Hom}}

\newcommand{\banachSpace}{\mathcal{B}}
\newcommand{\pairing}[2]{\langle{#1},\, {#2} \rangle }
\newcommand{\Ltwo}{L\squared}

\newcommand{\Cpvar}[1][p]{C^{{#1}\text{-var}}}

\newcommand{\approxAdditivepVariation}{AA_{p\text{-var}}}
\newcommand{\semigroupP}{\mathtt{P}}

\newcommand{\generatorL}{{\mathsf{L}}}
\newcommand{\generatorA}{\mathsf{A}}
\newcommand{\CalphaHoelderLoc}{C^{\alpha\text{-H\"ol}}_{\text{loc}}}
\newcommand{\ContFunctionsOfEllipticPDEregularity}{\mathcal{C}^{\alpha}}

\newcommand{\Prob}{{{P}}}
\newcommand{\probabilityQ}{{{Q}}}

\newcommand{\Expectation}{{{E}}}
\newcommand{\Variance}{\mathrm{Var}}

\newcommand{\sigmaAlgebra}{\mathfrak{F}}

\newcommand{\probabilitySpace}{\big(\Omega,\sigmaAlgebra, \Prob \big)}

\newcommand{\brownianMotion}{W}
\newcommand{\geometricBrownianMotion}{X}
\def\one{\mbox{1\hspace{-4.25pt}\fontsize{12}{14.4}\selectfont\textrm{1}}}

\newcommand{\control}{\omega}
\newcommand{\roughX}[1][X]{\mathbf{{#1}}}
\newcommand{\secondOrderX}[1][X]{\mathbb{\MakeUppercase{#1}}}
\newcommand{\secondOrderXst}[1][X]{\mathbb{\MakeUppercase{#1}}_{s,t}}
\newcommand{\roughPair}[1][X]{\mathbf{{#1}}=(#1,\mathbb{{#1}})}
\newcommand{\roughBracket}[1][X]{\left[\roughX[#1]\right]}

\newcommand{\roughBracketijst}[1][X]{\left[\roughX[#1]\right]^{i,j}_{s,t}}

\newcommand{\Xs}{X_{s}}
\newcommand{\Xt}{X_{t}}	
\newcommand{\Xst}{X_{s,t}}

\newcommand{\Xsu}{X_{s,u}}
\newcommand{\Xut}{X_{u,t}}

\newcommand{\Xist}{\Xi_{s,t}}

\newcommand{\pricep}{\mathfrak{p}}
\newcommand{\portfolioValue}{V}
\newcommand{\discountedPortfolioValue}{\tilde{\portfolioValue}}
\newcommand{\DeltaHedge}{\mathtt{Delta}}
\newcommand{\GammaHedge}{\mathtt{Gamma}}
\newcommand{\interestRate}{r}
\newcommand{\risklessAsset}{S^{0}}
\newcommand{\pricePath}{S}
\newcommand{\enhancedPricePath}{\roughX[\pricePath]}

\newcommand{\enhancedPricePair}{\roughPair[\pricePath]}
\newcommand{\discountedPricePath}{\tilde{S}}

\newcommand{\volatilityCoefficient}{\sigma}
\newcommand{\volatilityOperator}{\generatorA}
\newcommand{\pricingMeasure}{\probabilityQ}
\newcommand{\pricingExpectation}{\Expectation_{\pricingMeasure}}
\newcommand{\linearDrag}{\ell}

\newcommand{\deceptiveArithmeticBM}{Z}
\newcommand{\BSsolutionInForwardPrices}{w}
\newcommand{\BSsolutionForPortfolioFunction}{v}
\newcommand{\benchmarkMarkovianModel}{\left( \left(\Omega,\sigmaAlgebra, \pricingMeasure, (\sigmaAlgebra_t)_t \right), \volatilityOperator \right)}

\newcommand{\approximatelyAdditive}{approximately additive}
\newcommand{ \roughS}{\roughX[S]}
\newcommand{\enhancedTruePricePath}{{\mathbf{\pricePath}^{\mathrm{true}}}}
\newcommand{\secondOrderTruePricePath}{\secondOrderX[\pricePath]^{\mathrm{true}}}
\newcommand{\roughPairTruePricePath}{\enhancedTruePricePath=(\pricePath,\secondOrderTruePricePath)}

\newcommand{\enhancedDiscountedPricePath}{\roughX[\discountedPricePath]}
\newcommand{\secondOrderS}{\secondOrderX[S]}
\newcommand{\roughBracketS}{[ \roughS]}
\newcommand{\Hzero}{H^{0}}
\newcommand{\piH}[1][\pi]{\, ^{{{#1}}}\! H}
\newcommand{\piHzero}[1][ ]{\, ^{\pi_{{#1}}}\! H^{0}}
\newcommand{\piHprime}[1][ ]{\, ^{\pi_{{#1}}}\! H\derivative}
\newcommand{\subscriptuuprime}{_{u,u\derivative}}
\newcommand{\subscriptumintuprimemint}{_{u\wedge t,\, u\derivative\wedge t}}
\newcommand{\subscriptst}{_{s,t}}
\newcommand{\intzerot}{\int_{0}^{t}}
\newcommand{\eminusrt}{e^{-rt}}

\newcommand{\htilde}{\tilde{h}}
\newcommand{\sigmasquare}{\mathfrak{a}}
\newcommand{\aij}{a^{i,j}}
\newcommand{\Xuuprime}[1][t]{ X_{u \wedge {#1} , u\derivative \wedge  {#1} } }
\newcommand{\uprime}{u\derivative}
\newcommand{\secondOrderXuuprime}[1][t]{ \secondOrderX_{u \wedge {#1} , u\derivative \wedge  {#1} } }
\newcommand{\piphi}{{^{\pi}}\!\!\phi}
\newcommand{\piphizero}{\piphi^{0}}
\newcommand{\piphione}{\piphi^{1}}
\newcommand{\piphitwo}{\piphi^{2}}
\newcommand{\piDeltaHedge}{{^{\pi}}\! \DeltaHedge}
\newcommand{\piGammaHedge}{{^{\pi}}\! \GammaHedge}
\newcommand{\onepi}{{^{(1)}\!\pi}}
\newcommand{\twopi}{{^{(2)}\!\pi}}
\newcommand{\rebal}{\mathtt{rebal}^{\pi}}

\newcommand{\price}{S}

\newcommand{\tripleNormClosure}{L^1(\timeWindow;\Ltwo(\Prob))}

\title{Option pricing models without probability:\\ a rough paths approach}

\author{
	{John Armstrong\footnote{King's College London.}, Claudio Bellani\footnote{Imperial College London.}, Damiano Brigo\footnotemark[\value{footnote}\footnote{Damiano Brigo is grateful to the participants of the conference in \cite{Bri19pro} for helpful feedback}], Thomas Cass\footnotemark[\value{footnote}
	\footnote{The work of Thomas Cass is supported by EPSRC Programme Grant EP/S026347/1}]
}}

\date{Wednesday 8 July 2020}

\begin{document}

\maketitle

\numberwithin{equation}{section}

\vspace{0.5cm}

\textbf{Note from the authors}. 
The authors would like to dedicate this paper to their late colleague Mark Davis. Each of the authors benefited greatly from discussions with Mark over the years and did so, in particular, during the preparation of an early version of this manuscript.  One aspect of the presentation below is a perspective on the so-called Fundamental Theorem of Derivative Trading. Mark often stressed the importance of this result to the understanding of real-world trading; indeed he included a version of it in his entry ``Black-Scholes Formula'' in the Encyclopedia of Quantitative Finance \cite{EncyclopediaQF2010}.  His decency and his good-natured common sense will be missed.

\vspace{0.5cm}

\begin{quotation}
{\small \noindent\textbf{Abstract.} We describe the pricing and hedging of
financial options without the use of probability using rough paths. By
encoding the volatility of assets in an enhancement of the price trajectory,
we give a pathwise presentation of the replication of European options. The
continuity properties of rough-paths allow us to generalise the so-called
fundamental theorem of derivative trading, showing that a small
misspecification of the model will yield only a small excess profit or loss of
the replication strategy. Our hedging strategy is an enhanced version of
classical delta hedging where we use volatility swaps to hedge the second
order terms arising in rough-path integrals, resulting in improved robustness. }
\end{quotation}

\vspace{0.5cm}

\section{Introduction}

\label{sec.introduction}

The theory of rough-paths provides a framework for understanding differential
systems driven by irregular input signals. An asset-price process arising from
a diffusion model may  be associated a rough-path. Conversely we will find a
necessary condition for a rough-path to arise from a given diffusion model,
and we will call a rough-path satisfying this condition a \emph{diffusive
rough-path}. An investment strategy gives rise to a rough differential
equation (RDE) describing the evolution of the profit and loss (P\&L) of the strategy under a
given asset price signal. Given an option with a smooth payoff function, we
will show that the P\&L of a modified version of the classical delta hedging
strategy replicates the option payoff for any diffusive rough-path. The
modification we make to achieve this replicaton is to augment the
delta-hedging strategy with additional trades determined by a particular type
of volatility swap. By assuming that the price of these swaps is well
controlled we see that, in the continuous time limit, purchasing these swaps
will not influence the P\&L.

A core property of  RDE solutions is their continuity with respect to the
input rough-path. A first consequence therefore of our rough-path approach is
robustness of our proposed hedging strategy: if the true asset price signal is
close to a diffusive signal, our hedging strategy will still approximately
replicate the option payoff.
This relates to the classical Fundamental Theorem of Derivative Trading. Classical references for this formula are  the entry ``Black-Scholes Formula'' in the Encyclopedia of Quantitative Finance \cite{EncyclopediaQF2010}, and  \cite{KJPS98rob}; a recent paper about the same topic is \cite{EJP17fun}.
This classical theorem shows that if one hedges according to
a given diffusion model but the actual asset price process is determined by a
nearby diffusion model, the error of the classical delta hedging strategy will
be small. Our approach goes beyond this in that it allows for asset price
signals that do not arise from any diffusion model at all. Due to phenomena
such as market-impact and front-running, any differential equation describing
the dynamics of the P\&L of an investment strategy in terms of the asset price
dynamics is likely to contain some error. A perturbation of the second order
term of the asset price dynamics allows us to model such an error, and hence
explain the robustness of hedging strategies in more realistic markets than
those given by diffusion models.

A second consequence of our rough-path approach is that it demonstrates that a
theory of hedging is possible without the need for probability theory, despite
the central role of probability in the classical treatment of hedging
established in \cite{HK79mar,HP81mar}. Our work clarifies the use of
probability theory in justifying prices by identifying two steps: (i) showing
that the asset price paths of a diffusion model satisfy our diffusivity
condition; and (ii) deducing the uniqueness of the price of an option from the
existence of a replicating strategy via a no-arbitrage argument. In a market
with an arbitrage any price is possible, so there is no hope of obtaining
uniqueness without invoking a no-arbitrage condition, and hence involving
probability theory. We see, therefore, that the correct probability-free
analogue of classical pricing is demonstrating the existence of replicating
strategy for a given initial endowment. In this way we may interpret our
theory as giving a probability-free approach to pricing.

In diffusion models, the quadratic variation is a well-defined pathwise notion
which determines the price. Our definition of a diffusive rough-path
identifies the exact property needed for the delta hedging strategy to work in
a rough-path context. A continuous pricing signal is enhanced with a
specification for its \emph{rough bracket} to obtain a \emph{reduced rough
path} (see \cite[Chapter 5]{FH14cou}) which we will term an \emph{enhanced
price path}. Our financial model will take the form of a specification for the
properties of the rough bracket. Thus our model specification is tantamount to
a choice of enhancer, and it is this rough bracket which provides the
appropriate analogue of quadratic variation for our asset pricing model. In
our version of the Fundamental Theorem of Derivative Pricing, we will study
the effect of a misspecification of the financial model by examining the
sensitivity of our strategy to the choice of enhancer.

{{The purely pathwise nature of the enhancer, the price and hence the implied
volatility is in marked contrast to the statistical (and therefore
probabilistic) notion of historical volatility}}. This dichotomy between pathwise
and probabilistic properties has been noted before. For example, it is
exploited in \cite{BM00opt}, which partly inspired the present work (see also \cite{Bri19pro}), to give examples of diffusion models which are statistically indistinguishable using samples on a fixed time grid yet which
have arbitrarily different option prices.

Ours is not the first work to give a non-probabilistic formulation of option
pricing. \cite{BW94dyn} obtained a pathwise formulation of option pricing
using the non-probabilistic approach to It\^{o} calculus given in
\cite{Foe81cal}. A similar approach to pricing can also be found in
\cite{Rig15pat}. One caveat of this approach is that the continuous-time
integrals used in \cite{BW94dyn} depend upon the discrete-time approximating
sequence, which more or less precludes obtaining robustness in their approach.
To cirucmvent this dependence, our proposed strategy is an augmented version
of delta-hedging where one also invests in volatility swaps in order to hedge
the second order part of the pricing signal. This yields a robust trading
strategy, however at the expense of introducing assumptions on the price of
volatility swaps to ensure our strategy is self-financing.

Another approach is given in \cite{BSV08pri}, which uses quadratic variation
to give a theory of pricing which is able to accommodate non semi-martingale
signals. However, this theory is restricted to signals of regularity at least
as regular as those of semi-martingales whereas our theory accommodates paths
of finite $p$-variation for $2<p<3$.

One additional assumption that we must make in our approach is that the option
payoff is differentiable. We will show that for a European call option with
strike $K$, one can find diffusive rough-paths for which our strategy fails to
replicate the option payoff. However, these rough-paths must have a stock
price exactly equal to the strike at maturity. In a probabilistic theory, such
paths occur with probability zero, so may be neglected. However, our
interpretation is that the existence of such paths demonstrates a genuine lack
of robustness of the classical delta hedging strategy. The need for a robust
strategy becomes more important towards maturity as classical diffusion models
break down and new phenomena occur such as the ``pinning'' of stock prices
around exchange traded strikes (see e.g. \cite{AL03mar}, \cite{AKL12mat}, \cite{JIS08mod}, \cite{GJ12pin}). The failure of our
strategy for certain stock paths indicates that one should switch strategy
near maturity to a genuinely probabilistic strategy, such as a buy-and-hold
strategy. This reflects actual trading practice, where delta hedging
strategies are abandoned and quite different strategies adopted near maturity.

The article is organised as follows. In Section
\ref{sec.preliminariesAndHeuristics} we recall the classical theory of hedging
and establish our notation. In Section \ref{sec.pathwiseIntegrals} we explain
how we rely on reduced rough path integration to tackle the technical
difficulty of integration with respect to unbounded variation signals and
explain in detail the difference between our approach and that of
\cite{BW94dyn}. In Section \ref{sec.enhancedPathOfDiffusionType} we define
what is meant by a diffusive rough-path. In Section
\ref{sec.pathwiseFormulationOfFundamentalEquationsOfHedging} we demonstrate
formally how to obtain a pathwise formulation of the classical formulas of
Mathematical Finance in continuous time. Section \ref{sec.EnlargedStrategies}
shows how our continuous time trading strategy can be interpreted as a limit
of discrete time trading strategies in volatility swaps. Section
\ref{sec.BSfromOurPathwisePerspective} demonstrates that our proposed strategy
fails in the Black-Scholes model for call options when the stock price
terminates at the strike. Section \ref{sec.conclusions} presents our conclusions.

An appendix \ref{sec.disentanglingVolatilityFromMarginalVariances}
expands on \cite{BM00opt} by explaining their construction from our pathwise
point of view.

\section{Notation and preliminaries}

\label{sec.preliminariesAndHeuristics}

We will develop a rough-path version of a classical  diffusion model, and will
begin by describing the  classical model. We suppose that each component of
the price vector $\pricePath_{t} \in\Rd$ of $d$ non-dividend paying stocks
displays the following dynamics in the pricing measure
\[
d\pricePath^{i}_{t} = \sigma^{i}_{j} dB^{j}, \quad i=1,\dots, d,
\quad\pricePath_{0} = s_{0} \in\Rd,
\]
on a stochastic base $(\Omega, \sigmaAlgebra , \Prob , (\sigmaAlgebra_{t})_{t}
, (B_{t})_{t})$ carrying a standard $n$-dimensional Brownian motion
$(B_{t})_{t}$.  We assume  $\sigma$ in $\CalphaHoelderLoc (\Rd, \R^{d\times
n})$, $0<\alpha<1$.  Einstein's summation convention on double indices is
employed and will be throughout all the paper.

We assume there is also a riskless asset, denoted by $\risklessAsset_{t}$
which follows the one-dimensional deterministic dynamics
\[
d\risklessAsset_{t} = r\risklessAsset_{t} dt, \quad\risklessAsset_{0}=1.
\]
We let $\discountedPricePath_{t}$ be the discounted price of the risky asset
at time $t$, namely $\discountedPricePath_{t} = \pricePath_{t} /
\risklessAsset_{t} = \eminusrt \pricePath_{t}$.

 Let $f(\pricePath_\timeHorizon)$ be the payoff of a Vanilla option on the underlying $\pricePath$. We assume that $f$ is a continuous and bounded function on $\Rd$. Let 
 \begin{equation} \label{eq.payoffInForwardPrices}
 h(x):= f(e^{\interestRate \timeHorizon} x)
 \end{equation}
 and $\tilde{h} := e^{-\interestRate\timeHorizon}h$. The payoff is therefore equivalently written as $h(\discountedPricePath_\timeHorizon)$, and its discounted value is $\tilde{h} (\discountedPricePath_\timeHorizon)$.
 
 The classical theory of \cite{HK79mar,HP81mar}
 tells us that the option payoff can be replicated at time $t$ for a price, $\discountedPortfolioValue_t$
 satisfying
 \[
 \discountedPortfolioValue_t = \semigroupP_{T-t} \htilde (\discountedPricePath_t),
 \]
 where $(\semigroupP_t)_t$ is the semigroup on $C_b(\Rd)$ generated\footnote{
 	By this we mean the semigroup of linear operators on $C_b(\Rd)$ such that for any continuous and bounded $f$, the solution of the Cauchy problem
 	\[
 	\begin{cases}
 	\Big(\partial_t - \volatilityOperator  \Big) u = 0 \quad \text{ in } \R_+ \times \Rd \\
 	u(0,x)= f(x) \quad\text{ on }\lbrace 0 \rbrace \times \Rd 
 	\end{cases}
 	\]	
 	is represented as $u(t,x)=\semigroupP_t f(x)$.
 } by the infinitesimal operator
 \[
 \volatilityOperator\varphi (x) = \half \left( \sum_k \sigma^i_k \sigma^j_k   \right) \partialij \varphi(x) \, , \quad \varphi \in C^2(\Rd) \cap C_b(\Rd),
 \]
 of the dynamics of $\discountedPricePath$.
 We call $\volatilityOperator$ the {\em volatility operator}. To ensure the absence of arbitrage, we must make
 some assumptions to ensure that the solutions
 to the Black--Scholes PDE are unique. In this paper will typically
 assume that the volatility operator is uniformly
 elliptic.
 
 The stochastic process $\discountedPortfolioValue_t$ is then a deterministic function $\BSsolutionInForwardPrices=\BSsolutionInForwardPrices(t,x)$ of time and space applied after $(t,\discountedPricePath_{t})$ which solves
 \begin{equation}\label{Eq.discountedBSpde}
 \begin{cases}
 \Big(\partial_t + \volatilityOperator  \Big) \BSsolutionInForwardPrices = 0 \quad \text{ in } [0,T) \times \Rd \\
 \BSsolutionInForwardPrices(\timeHorizon,x)= \htilde(x) \quad\text{ on }\lbrace \timeHorizon \rbrace \times \Rd .
 \end{cases}
 \end{equation}
 Equation \eqref{Eq.discountedBSpde} is the discounted version of the  Black-Scholes partial differential equation.
 
 We will write $v(t,z)=e^{rt}w(t,e^{-rt}z)$
 for the undiscounted value function and will use following notation for the Greeks:
 \[
 \DeltaHedge_t := \nabla_z \BSsolutionForPortfolioFunction (t,S_t) = \gradx \BSsolutionInForwardPrices(t,\discountedPricePath_{t}),
 \]
 taking values in $\Rd \cong \Homomorphisms(\Rd,\R)$, and
 \[
 \GammaHedge_t := \Hessianz \BSsolutionForPortfolioFunction (t,S_t) = e^{rt} \Hessianx \BSsolutionInForwardPrices(t,\discountedPricePath_t),
 \]
 taking values in $\R^{d\times d} \cong \Homomorphisms(\Rd \otimes \Rd , \R) \cong \Homomorphisms(\Rd,\Homomorphisms(\Rd,\R))$.

 In our setup, the pricing PDE is justified by the existence
 of a replicating strategy for the payoff.
 An investment strategy
 may be viewed as a pair $(\Hzero_t , H_t)$ indicating the quantities
 to purchase at each time of the riskless and risky asset.
 By It\^o's formula,
 \begin{align} \label{Eq.integralAccruing}
 \BSsolutionInForwardPrices(t,X_t) - \BSsolutionInForwardPrices(0,X_0)  = & 	\intzerot \gradx \BSsolutionInForwardPrices(u,X_u)\sigma(X_u)dW_u + \intzerot \big(\partial_t +\volatilityOperator \big) \BSsolutionInForwardPrices(u,X_u) du \nonumber \\
 = & \intzerot \gradx \BSsolutionInForwardPrices(u,X_u)dX_u.
 \end{align}
 It follows that the
 delta hedging  strategy $\phi_t = (\Hzero_t , H_t)$ given by 
 \begin{equation}\label{Eq.replicatingStrategy}
 \begin{split}
 H_t &:= \gradx \BSsolutionInForwardPrices (t,\discountedPricePath_t) \\
 \Hzero_t &:= \BSsolutionInForwardPrices(t,\discountedPricePath_t) - H_t \discountedPricePath_t
 \end{split}
 \end{equation}
 is such that the undiscounted portfolio process
 \[
 \portfolioValue_t(\phi) = \Hzero_t e^{\interestRate t} + H_t \pricePath_t = e^{rt}\BSsolutionInForwardPrices(t,\discountedPricePath_{t}),
 \]
 is self-financing, i.e.\ it satisfies
 \begin{equation}\label{Eq.selfFinancingContTime}
 \portfolioValue_t (\phi) = \portfolioValue_0 (\phi) + \intzerot \Hzero_u d\risklessAsset_u + \intzerot H_u dS_u,
 \end{equation}
 and replicates the option payoff, i.e. $\portfolioValue_\timeHorizon(\phi) = f(\pricePath_T)$.

 \section{Pathwise integrals}\label{sec.pathwiseIntegrals}
 
 In this section we will first motivate the
 use of compensated Riemann sums for our problem
 and then describe the definitions and results
 on pathwise integration that we will require.
 Proofs are given in \ref{sec.proofsPathwiseIntegrals}.
 
 \subsection{Motivation}
 Consider the setting introduced in Section \ref{sec.preliminariesAndHeuristics}. Let $H$ denote a continuous hedging strategy and let $X$ be a semimartingale. In this setting, the It\^o integral $\int H dX$ is employed to represent portfolio values associated with the strategy $H$ if the price of the traded asset is modelled on the semimartingale $X$. This is justified by the following limit argument, which links real trading to the continuous-time abstraction. 
 
 In real trading,  hedging happens in discrete time and, given a partition $\partition$ of $\timeWindow$ with mesh-size $\abs{\partition} := \sup\lbrace \abs{u\derivative-u}: \, u \in \pi\rbrace$, the strategy $H$  is commonly replaced by 
 \[
 \piH_t := \sum_{u \in \partition} H_u \one \big\lbrace t \in (u,u\derivative] \big\rbrace , 
 \]
 where for $ u\in \pi$ we denoted by $u\derivative := \inf \lbrace v \in \pi : \, u<v\rbrace$ its next partition point. 
 The (It\^o-) integral of the elementary caglad process $\piH$ is
 \[
 (\piH . X)_t = \sum_{u \in \pi} H_u X \subscriptuuprime.
 \]
 $H$ being continuous, the discrete-time integral converges in probability to the It\^o integral $\int HdX$ along any sequence $(\pi^n)_n$ of partitions with mesh-size $\abs{\pi^n}$ shrinking to zero (see  \cite[Chapter IV, (2.13) Proposition]{RY99con}). Two issues are to be stressed about this convergence. First, the convergence does not happen pathwise for arbitrary sequences of partitions; if it did, the integrator $X$ would be of bounded variation (see Proposition \ref{Prop.[RY99,ChapterIV,(2.21)Exercise]} below), which is not the case for semimartingales. Second, the convergence is not uniform with respect to equivalent martingale measures on the same filtered probability space (see Example \ref{Ex.nonCommutativityOfLimits} below).

 Therefore, the justification for the employment of It\^o integrals in hedging crucially relies on the probabilistic framework. Disengaging from this justification, A. Bick and W. Willinger in \cite{BW94dyn} formulated a probability-free analysis of hedging strategies' portfolio values expressed as pathwise integrals \`a la F\"ollmer. Their approach bore the weight of the necessary specification of which approximating sequence every integral depends on, but this seemed to represent the trade-off needed for the probability-free description. 
 
 With the development of Rough Path Theory however, it was understood that, by appropriately compensating Riemann sums, the pathwise convergence could be achieved consistently across all possible sequences of approximations. Moreover, by employing Rough Path Theory, the assumption on the regularity of the stock price can be weakened: instead of assuming that the paths are of finite quadratic variation, we will be able to assume only that they are of finite $p$-variation for some $p$ smaller than 3. This will let our formulas make one further step away from the probabilistic framework, because these formulas will not assume the typical semimartingale regularity of a finite non-null quadratic variation. 
 
 Rough Path Theory achieves the stated convergence by deploying higher-order expansions of integrators and integrands. In the canonical semimartingale setting, these higher order expansions will entail that the Riemann approximation of the It\^o integral $\int H dX$ will be compensated with a term of the following form 
 \begin{equation}\label{eq.compensationWithQuadraticVariation}
  \sum_{u \in \pi} H\derivative _u \secondOrderXuuprime =
 \sum_{u \in \pi} \big[ H\derivative_u \Xuuprime \otimes \Xuuprime - H\derivative _u \pairing{X}{X}_{u\wedge t , u\derivative \wedge t}\big],
 \end{equation}
 where $\langle X, X \rangle$ stands for the quadratic variation of the semimartingale $X$, and $H\derivative$ is a suitable derivative of $H$ that will be introduced below.
 
 The compensation term is the technical ingredient that distinguishes our integrals from the classical It\^o integrals and from the pathwise integration employed in \cite{BW94dyn}. On the one hand, the compensation term vanishes in probability; indeed, with respect to any $\Prob$-equivalent measure the quantity in equation \eqref{eq.compensationWithQuadraticVariation} goes to zero in probability as the mesh-size $\lvert \partition \rvert$ shrinks to zero (see for example \cite[ChapterIV, (1.33) Exercise]{RY99con}). Hence, the probabilistic It\^o integration is insensitive to whether or not the compensation is incorporated. On the other hand, if the compensation is not incorporated, then the Riemann sums cannot converge pathwise along arbitrary sequences of partitions with vanishing mesh-zise (See Proposition \ref{Prop.[RY99,ChapterIV,(2.21)Exercise]}). This was the reason  why A. Bick and W. Willinger \cite{BW94dyn} were forced to attach to their continuous-time integrals the specification of the sequence of partitions along which the integrals were approximated. Employing compensated Riemann sums, instead, allows to circumvent this specification.
 
 It is the purpose of this section to describe the compensations to the Riemann sums which we will later apply to the integrals in the classical formulas of Mathematical Finance.

\subsubsection*{Technical motivation for compensated  Riemann  sums} \label{Sec.motivationCompensatedSums}
There are two technical reasons which motivate the use of compensated Riemann sums. First, if the integrator has semimartingale-type regularity (or worse), then uncompensated Riemann sums cannot converge pathwise (see Proposition \ref{Prop.[RY99,ChapterIV,(2.21)Exercise]}). Second, in the semimartingale setting the rate of convergence of uncompensated Riemann sums to the corresponding It\^o integral depends on the underlying probability measure (see Example \ref{Ex.nonCommutativityOfLimits}).

 
 \begin{prop}[{\cite[Chapter IV, (2.21) Exercise]{RY99con}}]\label{Prop.[RY99,ChapterIV,(2.21)Exercise]}
 	Let $X$ be an $\Rd$-valued function on $[0,T]$ such that for every $H$ in $C([0,T],\Rd)$ the limit of (uncompensated) Riemann sums
 	\[
 	\lim_n (\piH[\pi_n].X)_T = \lim_n \sum_{u \in \partition_n} H_u X\subscriptuuprime
 	\]		
 	exists in $\R$ along any sequence $(\pi_n)_n$ of partitions of $[0,T]$ with vanishing mesh-size. Then, the limit does not in fact depend on the specific sequence of partitions and $X$ is of bounded variation. 
 \end{prop}

The second technical motivation follows. It discusses the convergence of Riemann sums in relation to the probability measure with respect to which such a convergence happens. 

 \begin{example}\label{Ex.nonCommutativityOfLimits}
 	Let $\Omega$ be the space $C([0,1],\R)$ of continuous real-valued functions on $[0,1]$,  and let $\sigmaAlgebra$ be its Borel $\sigma$-algebra. Let $P$ be the Wiener measure on $(\Omega, \sigmaAlgebra)$, so that the coordinate map $X_t(\omega)=\omega(t)$, $\omega \in \Omega$, is a $((\sigmaAlgebra_t)_t, P)$-Brownian motion, where $\sigmaAlgebra_t$ is the $P$-completion of $\sigma(X_s:\,  0\leq s\leq t)$. Consider the sequence $P^k$, $k\in \N$, of probability measures on $(\Omega,\sigmaAlgebra, (\sigmaAlgebra_t)_t)$ given by
 	\[
 	\frac{dP^k}{dP}\arrowvert_{\sigmaAlgebra_t} = \exp \Big(kX_t - \frac{k\squared}{2}t\Big), \qquad 0\leq t \leq 1, \quad k=1,2,\dots.
 	\]
 	We observe that for each $k$, the process $X_t - kt$ is a $((\sigmaAlgebra_t)_t, P^k)$-Brownian motion. For each $k$, for every continuous $(\sigmaAlgebra_t)_t$-adapted integrand $H$ and any sequence $(\pi_n)_n$ of partitions with vanishing mesh-size we have that
 	\[
 	\sup \Big\lbrace \lim_{n,m \uparrow \infty} P^k \Big(\abs{(\piH[\pi_n].X)_1 - (\piH[\pi_m].X)_1}> \epsilon \Big): \, \epsilon>0 \Big\rbrace = 0.
 	\]
 	Consider the integrand $H_s \equiv s$ and the sequence $\pi_n = \lbrace l 2^{-n}: \, l=0,\dots, 2^n\rbrace$ of dyadic partitions of $[0,1]$. We claim that for every $\epsilon>0$ it simultaneously holds
 	\[
 	\lim_k \lim_{n } P^k \Big(\abs{(\piH[\pi_n].X)_1 - (\piH[\pi_{n+1}].X)_1}> \epsilon \Big) = 0
 	\]
 	and 
 	\[
 	\lim_{n }\lim_k P^k \Big(\abs{(\piH[\pi_n].X)_1 - (\piH[\pi_{n+1}].X)_1}> \epsilon \Big) = 1.
 	\]
 	Indeed, since $\pi_{n+1} = \pi_n \cup \lbrace (2l+1)2^{-(n+1)}: \, l=0, \dots , 2^n -1\rbrace$, we have
 	\begin{align*}
 	\sum_{\substack{v \in \pi_{n+1} \\v<1}} H_v X_{v,v\derivative} - 
 	\sum_{\substack{u \in \pi_n \\ u < 1}} H_u X \subscriptuuprime =& 
 	\sum_{\substack{u \in \pi_n \\ u < 1}} \big(H_{u+2^{-(n+1)}} - H_u\big) X_{u+2^{-(n+1)}, \, u +2^{-n}} \\
 	=& 2^{-(n+1)} X_{2^{-(n+1)},1}.
 	\end{align*}
 	Hence, under $P^k$ the difference $(\piH[\pi_{n+1}].X)_1 - (\piH[\pi_{n}].X)_1$ is distributed as 
 	\[
 	2^{-(n+1)}\sigma_n \Big( N + k \sigma_n\Big),
 	\]
 	where $\sigma_n = (1-2^{-(n+1)})^{\half}$ and $N$ is a standard normal random variable. We conclude 
 	\begin{align*}
 	P^k \Big(\big\lvert(\piH[\pi_{n+1}].X)_1 -& (\piH[\pi_{n}].X)_1\big\rvert>\epsilon \Big) \\
 	=&  
 	\Phi \Big(k\sigma_n - \frac{\epsilon}{\sigma_n}2^{n+1}\Big) + \Phi \Big( -\frac{\epsilon}{\sigma_n}2^{n+1} - k \sigma_n\Big) \\
 	& \stackrel{k \uparrow \infty}{\longrightarrow} 1,
 	\end{align*}
 	where $\Phi$ is the cumulative distribution function of a standard normal random variable.
 \end{example}

\subsection{Additivity, approximate additivity and the Sewing lemma} \label{sec.additivity}

Let $\banachSpace$ be a Banach space and let $X:\timeWindow \rightarrow \banachSpace$ be a continuous path with trajectory in $\banachSpace$. The increments
\begin{equation}\label{eq.notationForIncrements}
\Xst : = \Xt - \Xs, \qquad \zerostTimeHorizon,
\end{equation}
of such path define a two-parameter function $X=\Xst$ on the square $\timeWindow \times \timeWindow$. We employ the notation in \eqref{eq.notationForIncrements} throughout all our work. Moreover, rather than considering general $s$ and $t$ in $\timeWindow$, we often restrict to the simplex $\simplex \subset \timeWindow\squared$. A clear property of $X$ is additivity, in that for all $0\leq s,u,t \leq \timeHorizon$ it holds
\begin{equation}\label{eq.additivityOfIncrements}
\Xst = X_{s,u} + X_{u,t}.
\end{equation} 
Notice that if $X$ is a priori only defined on the simplex but additive, then it can straightforwardly be extended to an additive function on $\timeWindow\times\timeWindow$ by setting $X_{t,s}:=-\Xst$.

Additivity characterises those functions $X$ on  $\timeWindow\times\timeWindow$ that descend from increments of paths, in the following sense.

\begin{prop}\label{prop.additivityMeansPath}
	Let $X:\timeWindow \times \timeWindow \rightarrow \banachSpace$ be additive. Then, there exists a path $x$ on $\banachSpace$ such that 
	\[
	\Xst = x_t - x_s, \qquad \forall \zerostTimeHorizon.
	\]
	Moreover, if $y$ is another path whose increments coincide with $X$, then $y-x$ is constant. 
\end{prop}

 We regard a partition $\partition$ simultaneously as the finite collection of points and as the finite collection of adjacent subintervals of a given time interval $[s,t]$ that $\partition$ subdivides.  Given a partition $\partition$ of $\timeWindow$ and a time instant $t$ in $\timeWindow$,  we adopt the following notational convention: 
\begin{equation}\label{eq.notationAboutPartition}
\begin{split}
t\derivative:= \inf \lbrace u \in \pi : \, u> t \rbrace , & 
\qquad 
\lfloor t \rfloor := \sup\lbrace u \in \pi : \, u\leq t\rbrace, \\  
t-  := \sup\lbrace u \in \pi : \, u\derivative \leq t\rbrace ,&  
\qquad
t\star:= \begin{cases}
t- & \text{ if } t\in \pi \\
\lfloor t \rfloor & \text{ if } t\notin \pi,
\end{cases} \\
\abs{\pi} := \sup \lbrace \abs{u\derivative - u}: \, u \in \pi \rbrace,&  
\qquad
\pi_t := \big(\pi \cup \lbrace t \rbrace \big) \cap [0,t].
\end{split}
\end{equation}

Let $\partition$ be a partition of the time interval $[s,t]$ under consideration. From the additivity \eqref{eq.additivityOfIncrements} it follows that 
\begin{equation}\label{eq.telescopicSumOfIncrements}
\sum_{u \in \partition} X\subscriptuuprime = \Xst.
\end{equation}
This holds irrespectively of the choice of the partition $\partition$, so that if $\partition_n$, $n\geq 1$, is a sequence of partitions with mesh-size shrinking to zero, we can carry formula \eqref{eq.telescopicSumOfIncrements} to the limit in $n$ and write 
\begin{equation}\label{eq.limitOfTelescopicSumOfIncrementsAlongSequence}
\lim_{n\uparrow \infty} \sum_{u \in \partition_n} X\subscriptuuprime = \Xst.
\end{equation}
Actually, this  does not use the fact that the mesh-size $\abs{\partition_n}:=\sup\lbrace \abs{\uprime - u}: \, u \in \partition_n\rbrace$ goes to zero as $n\uparrow \infty$. However, restricting to such class of sequences of partitions will become meaningful soon, because we wish to interpret the limit in \eqref{eq.limitOfTelescopicSumOfIncrementsAlongSequence} as the integral $\int_{s}^{t}dX$. In order to emphasise that the limit in \eqref{eq.limitOfTelescopicSumOfIncrementsAlongSequence} does not depend on the particular sequence of partitions we write
\begin{equation}\label{eq.limitOfTelescopicSumOfIncrements}
\lim_{\abs{\partition}\downarrow 0} \sum_{u \in \partition} X\subscriptuuprime = \Xst.
\end{equation}
In integral notation, this is the trivial, yet fundamental, relation $\int_{s}^{t}dX = \Xst$.
Let $X:\simplex \rightarrow \banachSpace$. We say that $X$ is of finite $p$-variation for some $p\geq 1$ if 
\[
\pvarNormInterval[X]{p}{\timeWindow}^p :=
 \sup \left\lbrace
\sum_{u \in \partition} \abs{X\subscriptuuprime}^p: \, \partition \text{ partition of } \timeWindow 
\right\rbrace
<\infty.
\]
If $X$ is additive, this notation is the usual $p$-variation norm of the underlying path.

For $\timesSleqUleqT$ we introduce the symbol 
\[
\delta X_{s,u,t} := \Xst - \Xsu - \Xut.
\]
If $X$ is additive, then $\delta X \equiv 0$. 

Equation \eqref{eq.limitOfTelescopicSumOfIncrements} is the combination of two statements: a. the limit on the left hand side exists and is the same along every sequence of partitions with vanishing mesh-size; b. such limit defines an additive functional on the simplex, hence a path. We have seen that these properties are immediate if we start from an additive $X$.
We  will now relax the additivity of $X$ to obtain the non-trivial statement in Proposition \ref{prop.SewingLemma}. 

\begin{defi}[``Control function'']
	A \emph{control function} $\control$ is a non-negative continuous function on $\lbrace (s,t)\in \R^2: \,  0\leq s\leq t \leq T\rbrace$, null on the diagonal and such that
	\begin{enumerate}
		\item $\control(s_1,t_1)\leq \control(s_2,t_2)$, if the interval $[s_1,t_1]$ is contained in the interval $[s_2,t_2]$;
		\item $\control(s,u) +\control (u,t)\leq \control(s,t) $, for all $s\leq u\leq t$.
	\end{enumerate}
\end{defi}

 A control function generalizes the concept of the length of an interval. Common controls are $\control(s,t):=\abs{t-s}$ and, for a continuous path  $x$ of finite $p$-variation, $\control(s,t):=\norm[x]_{p\text{-var}, \, [s,t]}^{p}$. From these, new controls can be defined by linear combinations $c_1 \control_1 + c_2 \control_2$ with non-negative coefficients $c_1, c_2 \in \R_{\geq 0}$, and by products $\control_{1}^{\gamma_1}\control_{2}^{\gamma_2}$ with exponents $\gamma_1$ and $\gamma_2$ satisfying $\gamma_1 + \gamma_2\geq 1$,  see \cite[Exercise 1.9]{FV10mul}. 

Given a partition $\partition$ of $[s,t]\subset \timeWindow$ we may use a control function to measure the mesh-size.
	\begin{defi}
		The modulus of continuity of $\control$ on a scale smaller or equal than the mesh-size $\abs{\pi}$ is
		given by
		\[\mathrm{osc}(\control,\abs{\pi}):= \sup\lbrace \control(s,t): \, \abs{t-s}\leq \abs{\pi}\rbrace.\]
	\end{defi}

\begin{defi}[``Approximate additivity'']\label{def.approximateAdditivity}
		A function $\Xi:\simplex \rightarrow\banachSpace$ is said \emph{\approximatelyAdditive} if 
		\begin{enumerate}
			\item it is null and right-continuous on the diagonal, i.e. $\Xi_{s,s}$ $ = \lim_{t\downarrow s} \Xi_{s,t}=0$ for all $s$ in $\timeWindow$;
			\item \label{item.relaxedAdditivity}
						 there exist $\gamma > 1$ and a control function
			$\control$ such that 
			\begin{equation}\label{eq.relaxedAdditivity}
			\abs{\Xist - \Xi_{s,u} - \Xi_{u,t}} \leq \control^{\gamma} (s,t),
			\end{equation}
			for all $\timesSleqUleqT$. 
		\end{enumerate}
\end{defi}
Notice that equation \eqref{eq.relaxedAdditivity} implies that for all $1<\gamma\derivative<\gamma$
\[
\norm[\delta \Xi]_{\control, \gamma\derivative} :=
\sup_{\timesSleqUleqT} 
\frac{\abs{\delta \Xi_{s,u,t}}}{\control^{\gamma\derivative}(s,t)} \leq \control^{\gamma-\gamma\derivative}(0,\timeHorizon).
\]
Therefore, condition \ref{item.relaxedAdditivity} above is equivalent to the existence of a control $\control$ and some $\gamma>1$ such that $\norm[\delta\Xi]_{\control, \gamma} <\infty$.

\begin{prop}[``Sewing Lemma'']\label{prop.SewingLemma}
	Let $\Xi:\simplex \rightarrow \banachSpace$ be \approximatelyAdditive \, and let the control $\control$ and the exponent $\gamma>1$ be such that $\norm[\delta\Xi]_{\control,\gamma}<\infty$. Then, there exists a unique continuous path 
	\[
	\int \Xi : \timeWindow \rightarrow\banachSpace,
	\]
	whose increments we denote by $\int_{s}^{t} \Xi$, such that for all $\zeroSleqTtimeHorizon$ 
	\begin{enumerate}
		\item $\int_{s}^{t}\Xi = \lim_{\abs{\partition}\downarrow 0}\sum_{u\in\partition}\Xi\subscriptuuprime$ with limit in $\banachSpace$;
		\item \begin{equation}\label{Eq.errorOfIncrements}
		\abs{\int_{s}^{t}\Xi - \Xist } \leq \frac{\norm[\delta\Xi]_{\control,\gamma}}{1-2^{1-\gamma}}\control^{\gamma}(s,t).
		\end{equation}
	\end{enumerate}
\end{prop}
\begin{remark}
	With respect to the formulation in \cite[Lemma 4.2]{FH14cou}, Proposition \ref{prop.SewingLemma} extends the so-called Sewing Lemma to the case of general control $\control$. Hence, it allows to handle the  case of $p$-variation regularity, which is more general than the case of $1/p$-H\"older regularity. 
\end{remark}

Owing to Proposition \ref{prop.SewingLemma}, we can regard the integral as the map
\[
\int : \left\lbrace \substack{\text{approximatively} \\ \text{additive} \\ \text{functionals}} \right\rbrace 
\rightarrow 
\Big\lbrace \substack{ \text{additive} \\ \text{functionals}} \Big\rbrace .
\]
Owing to Proposition \ref{prop.additivityMeansPath}, we can unambiguously replace the range $\Big\lbrace \substack{ \text{additive} \\ \text{functionals}} \Big\rbrace $ with the space of continuous paths on $\banachSpace$ starting at $0 \in \banachSpace$. Let $\approxAdditivepVariation (\timeWindow; \banachSpace)$ be the family of \approximatelyAdditive \, functions $\Xi:\simplex \rightarrow \banachSpace$ that are of finite $p$-variation, $p\geq 1$. Then, we can state the following
\begin{corol}
	The restriction of the integral map $\int$ to $\approxAdditivepVariation(\timeWindow; \banachSpace)$ takes value in the space $\Cpvar_{0}(\timeWindow;\banachSpace)$ of continuous paths on $\banachSpace$ that start at the origin $0\in\banachSpace$ and are of finite $p$-variation. Moreover, 
	\begin{eqnarray*}
	\int:   \approxAdditivepVariation(\timeWindow; \banachSpace) &\longrightarrow &\Cpvar_{0} (\timeWindow; \banachSpace) \\
	\Xi &   \longmapsto & \lim_{\abs{\partition}\downarrow 0} \sum_{u\in\pi}\Xi\subscriptuuprime
	\end{eqnarray*}
	is continuous in $p$-variation norm. 
\end{corol}
\begin{proof}
	Immediate from equation \eqref{Eq.errorOfIncrements}.
\end{proof}

\subsection{Young integrals}

Let $X:\timeWindow \rightarrow\banachSpace$ be continuous and of finite $p$-variation. Let $H: \timeWindow \rightarrow W$ be continuous and of finite $q$-variation, where $W=\Homomorphisms(\banachSpace;V)$ and $V$ is a Banach space. We say that $p$ and $q$ are \emph{Young complementary} if $1/p + 1/q >1$.
\begin{prop}[``Young integral\footnote{The original article by L. C. Young is \cite{You36ine}, where the extension of Stieltjes integral was introduced. Our reference is the rough path-oriented presentation of the Young integral contained in \cite[Chapter 6]{FV10mul}.}'']\label{prop.YoungIntegral}
	Let $p$ and $q$ be Young complementary and set 
	\[
	\Xist := H_s \Xst
	\]
	for all $\zeroSleqTtimeHorizon$, or $\Xist := H_t\Xst$ for all $\zeroSleqTtimeHorizon$. Then, $\Xi$ is \approximatelyAdditive \, and of finite $p$-variation. As a consequence, the integral 
	\begin{equation}\label{eq.definitionYoungIntegral}
	H.X := \lim_{\abs{\partition}\downarrow 0 } \sum_{u\in\partition}\Xi\subscriptuuprime
	\end{equation}
	defines a continuous path in $V$ of finite $p$-variation. The integral in \eqref{eq.definitionYoungIntegral} does not depend on whether $\Xi$ is defined according to $\Xist = H_s\Xst$ or to $\Xist = H_t\Xst$.
\end{prop}

The continuity of $H$ was only used to show that the choice to evaluate $H$ at the beginning or at the end of the partition subintervals does not affect the integral. The two choices are respectively referred to as adapted evaluation and terminal evaluation. If $H$ is not continuous but of bounded variation, the Young integral is defined (because $q=1$), but depends on the evaluation choice. If $\partition$ is a partition of $\timeWindow$, we set 
\[
\piH_t:= \sum_{u \in \partition} H_u \one \lbrace t \in (u,\uprime]\rbrace,
\]
which denotes the piecewise constant caglad approximation of $H$ on the grid $\partition$. We let $\piH.X$ be the Young integral of $H$ against $X$ with terminal evaluation, namely
\[
(\piH.X)_{0,t} := \sum_{u\in\partition_t} H_u X\subscriptuuprime. 
\]
In this way, for $H$ continuous and of finite $q$-variation, $1/p+1/q>1$, we can write
\begin{equation}\label{eq.cagladApproximationYoungIntegral}
H.X = \lim_{\abs{\partition}\downarrow 0} \piH.X . 
\end{equation}

\subsection{Compensated integrals \`a la Gubinelli}

When the complementary regularities of integrand $H$ and integrator $X$ are not sufficient for Young integration, we resort to compensated Riemann sums. In particular this is the case if $H$ and $X$ have the same $p$-variation regularity for some $p$ greater than $2$. 

As above, let $X$ be a continuous path of finite $p$-variation with trajectory in the Banach space $\banachSpace$. Recall  that $W$ denotes $\Homomorphisms(\banachSpace;V)$. We use the identification $\Homomorphisms(\banachSpace, $ $ W) $ $ \cong\Homomorphisms(\banachSpace\otimes \banachSpace; V)$, and we write $\Homomorphisms_{\text{sym}}(\banachSpace\otimes \banachSpace; V)$ for the subset of those $\ell$ in $\Homomorphisms(\banachSpace\otimes \banachSpace; V)$ such that $\ell \, (a\otimes b) = \ell \,  (b\otimes a)$ for all $a,b \in \banachSpace$. Also, the symbol $\banachSpace \odot \banachSpace$ will denote the symmetric tensor product of the Banach space $\banachSpace$, so that we can identify $\Homomorphisms_{\text{sym}}(\banachSpace\otimes \banachSpace; V)\cong\Homomorphisms(\banachSpace\odot\banachSpace; V)$. We say that a continuous path $H: \timeWindow \rightarrow W$ admits a  symmetric \emph{Gubinelli derivative} $H\derivative$ with respect to $X$ if there exists a continuous path $H\derivative :\timeWindow \rightarrow \Homomorphisms_{\text{sym}} (\banachSpace\otimes\banachSpace;V)$ of finite $q$-variation such that
\begin{enumerate}
	\item $q$ and $p/2$ are Young complementary;
	\item $R^H_{s,t} : =H_{s,t} - H\derivative_{s}\Xst$ is of finite $pq/(p+q)$-variation.
\end{enumerate}
In this case we say that the pair $(H,H\derivative)$ is $X$-controlled of $(p,q)$-variation regularity. Notice that the regularities of $R^H$ and of $X$ imply that $H$ is of finite $p$-variation. 

\begin{defi}[``Enhancement of a path'']\label{defi.enhancementOfAPath}
	Let $X$ be in $\Cpvar(\timeWindow; \banachSpace)$ and let $A$ be in $\Cpvar[p/2](\timeWindow; \banachSpace\odot\banachSpace)$. The $A$-\emph{enhancement} of $X$ is the pair $\roughPair$, where
	\[
	2\secondOrderXst = \Xst \otimes \Xst - A_{s,t}.
	\]
\end{defi}
Similarly, we speak of an enhanced path\footnote{
	An enhanced path is what in \cite[Chapter 5]{FH14cou} is called \emph{reduced rough path}. An enhanced path satisfies the following two properties, which are taken as the defining properties of reduced rough paths:
	\begin{enumerate}
		\item the symmetric second order process $\secondOrderX=\text{sym}\secondOrderX$ is of finite $p/2$-variation;
		\item the reduced Chen's identity holds, i.e. for all $\timesSleqUleqT$
		\[
		\secondOrderXst - \secondOrderX_{s,u} -  \secondOrderX_{u,t} =  \Xsu \odot \Xut .
		\]
	\end{enumerate}
	 In \cite[Lemma 5.4]{FH14cou} these two properties are shown to necessarily imply the more explicit formulation that we adopted.
 }
 $\roughPair$ of $p$-variation regularity if $X$ is in $\Cpvar(\timeWindow; \banachSpace)$ and $(s,t)\mapsto   \Xst \otimes \Xst -2\secondOrderXst $ defines an additive $\banachSpace\odot \banachSpace$-valued function of finite $p/2$-variation. The path $A_{s,t}:=\Xst \otimes \Xst -2\secondOrderXst$ is called the \emph{enhancer} of $\roughX$ and we often denote such enhancer with the symbol 
\[
\roughBracket_{s,t}:=\Xst \otimes \Xst -2\secondOrderXst.
\]
The symbol  $\roughBracket$ will be referred to as  \emph{volatility} enhancer when the financial meaning of it is to be stressed.  
We say that $\roughPair$ is a \emph{bounded-variation enhancement} of $X$ if 
\[
\sup\left\lbrace
\sum_{u \in \partition}\abs{\roughBracket\subscriptuuprime}: \, \partition \text{ partition of } \timeWindow
\right\rbrace
<\infty.
\]
Notice that $\delta\secondOrderX$ does not depend on the enhancer because $\roughBracket$ is additive; moreover, for all $\timesSleqUleqT$ the following reduced Chen identity holds
\[
\delta\secondOrderX_{s,u,t}=\Xsu\odot\Xut.
\]

\begin{lemma}\label{lemma.approxAdditivityOfCompensatedSummand}
	Let $\roughPair$ be an enhanced path and let $(H,H\derivative)$ be $X$-controlled of $(p,q)$-variation regularity, with $H\derivative$ being symmetric. Then,
	\[
	\Xist := H_s \Xst +H\derivative_{s}\secondOrderXst
	\]
	is \approximatelyAdditive.
\end{lemma}

As a consequence of Lemma \ref{lemma.approxAdditivityOfCompensatedSummand}, the integral given by the compensated Riemann sum
\[
(H,H\derivative).(X,\secondOrderX) = 
\lim_{\abs{\partition}\downarrow 0}
\sum_{u \in \partition} 
\Big[
H_uX\subscriptuuprime + H\derivative_u \secondOrderXuuprime
\Big]
\]
is well-defined. Analogously to \eqref{eq.cagladApproximationYoungIntegral}, we write
\[
(\piH,\piH\derivative).(X,\secondOrderX) = 
\sum_{u \in \partition} 
\Big[
H_uX\subscriptuuprime + H\derivative_u \secondOrderXuuprime
\Big],
\]
so that
\[
(H,H\derivative).(X,\secondOrderX) = 
\lim_{\abs{\partition}\downarrow 0}
(\piH,\piH\derivative).(X,\secondOrderX).
\]

\subsubsection*{Space-gradient integrands associated with $q$-\emph{moderate} pairs}
If $J$ is a time interval,  $n$ and $m$ are non-negative integers and $\alpha$, $\beta$ are in $[0,1)$, consider the space
\[
C_{\text{loc}}^{m+\beta, \, \, n+\alpha}(J\times \Rd; \R^e)
\]
of $\R^e$-valued functions that are $m$ times continuously differentiable in time with the $m$-th time derivative of local $\beta$-H\"older regularity, and $n$ times continuously differentiable in space with all the $n$-th order space derivatives of local $\alpha$-H\"older regularity. Notice that nothing is assumed about the cross derivatives in time and space of functions  in $C_{\text{loc}}^{m+\beta, \, \, n+\alpha}$. Let $ C_{\text{cross}}^{m+\beta, \, \, n+\alpha}(\timeWindow\times \Rd; \R^e)$ be the subspace of $C_{\text{loc}}^{m+\beta, \, \, n+\alpha}(\timeWindow\times \Rd; \R^e)$ consisting of functions $f$ such that
\begin{enumerate}
	\item for every multiindex $I$ with $\abs{I}=n$ and every compact $K\Subset\Rd$, 
	\[
	\sup \Big\lbrace \norm[\partial_{x}^{I}f(t,\cdot)]_{\alpha\text{-H\"ol},\,  K}: \, 0\leq t\leq T \Big\rbrace < \infty;
	\]
	\item for every compact $K\Subset\Rd$,
	\[
	\sup \Big\lbrace \norm[\partial_t^m f (\cdot, x)]_{\beta\text{-H\"ol},\, [0,T]}: \, x \in K\Big\rbrace < \infty.
	\]
\end{enumerate}
Let $\ContFunctionsOfEllipticPDEregularity$ be the space 
\begin{equation}\label{eq.definitionContFunctionsOfEllipticPDEregularity}
\ContFunctionsOfEllipticPDEregularity :=
 C_{\text{loc}}^{1+\alpha/2, \, \, 2+\alpha}([0,\timeHorizon)\times \Rd)
 \cap
 C(\timeWindow\times\Rd).
\end{equation}

\begin{defi}[``$q$-Moderation'']\label{def.qModeration}
Let $w$ be in $\ContFunctionsOfEllipticPDEregularity$ and let $X$ be a continuous path on $\Rd$ of finite $p$-variation, with $p-2<\alpha<1$. We say that the pair $(w,X)$ is $q$-\emph{moderate} if
\begin{enumerate}
	\item the paths 
	\begin{equation*}
	\begin{split}
	H:& \quad t\longmapsto \gradx w(t,X_t)\\
	H\derivative:& \quad  t\longmapsto \Hessianx w (t,X_t),\\
	& \qquad \qquad \quad \quad \quad  0\leq t <\timeHorizon,
	\end{split}
	\end{equation*}
	can be continuously extended up to $\timeWindow$, and $H\derivative$ is of finite $q$-variation for some $1-2/p<1/q<\alpha/p$;
	\item there exists a control function $\control$ such that for all $x$ in the trace $X\timeWindow$ and all $\zeroSleqTtimeHorizon$
	\[
	\abs{\gradx w (t,x) - \gradx w (s,x)}^{p*} \leq \control (s,t),
	\]
	where $p*= pq/(p+q)$;
	\item 
	\[
	\sup_{0\leq s \leq \timeHorizon} \norm[\Hessianx w (s,\cdot)]_{\alpha\text{-H\"ol}, \convexHull X\timeWindow} <\infty,
	\]
	where $\convexHull X\timeWindow$ is the convex hull of the trace of $X$. 
\end{enumerate}
\end{defi}
\begin{remark}
	Let $0<\alpha<1$ and $p,q \geq 1$ be such that $1-2/p<1/q<\alpha/p$.
	Assume that $w \in  C_{\text{loc}}^{1 +\alpha/2, \, \, 2 +\alpha}(\timeWindow\times \Rd; \R)$ is such that $\gradx w$ is in $ C_{\text{cross}}^{1/p + 1/q, \, \, 1 +\alpha}(\timeWindow\times \Rd; \Rd)$ and $\Hessianx w $ is in $ C_{\text{cross}}^{ 1/q, \, \, \alpha}(\timeWindow\times \Rd; \R^{d\times d})$. Then for all $X$ in $\Cpvar(\timeWindow;\Rd)$ the pair $(w,X)$ is $q$-moderate. In particular this holds if $w$ is twice continuously differentiable in the combined time-space variable $(t,x)$ with second derivatives of $\alpha$-H\"older regularity. 
\end{remark}

\begin{lemma}\label{lemma.qModerationAndControlledPaths}
	Let $w$ be in $\ContFunctionsOfEllipticPDEregularity$ and let $X$ be a continuous $\Rd$-valued path of finite $p$-variation, with $p-2<\alpha<1$. Assume that the pair $(w,X)$ is $q$-moderate, $1-2/p<1/q<\alpha/p$. Then, 
	\[
	(H,H\derivative):=\left(\gradx w (t,\Xt),\Hessianx w(t,\Xt)\right)
	\]
	is a Gubinelli $X$-controlled path of $(p,q)$-variation regularity.
\end{lemma}

\section[Enhanced  paths of diffusion type]{Enhanced  paths of diffusion type}\label{sec.enhancedPathOfDiffusionType}

Relying on the pathwise integrals introduced in Section \ref{sec.pathwiseIntegrals}, we now describe a framework to assess the pathwise feature of price trajectories that actually affects the hedging practice, disregarding the  probabilistic specifications of the stochastic models. We will consider enhanced price paths (as defined in Definition \ref{defi.enhancementOfAPath}) that embed the essential feature of the stochastic models. Among these we isolate those that descend from classical diffusion models (Markovian SDEs) wich we adopt as benchmark.

A benchmark Markovian model consists of the pair $\benchmarkMarkovianModel$, where $(\Omega,\sigmaAlgebra,$ $\pricingMeasure,(\sigmaAlgebra_t)_t)$ is a filtered probability space and $\volatilityOperator$ is a diffusion generator. Until further notice, we adopt the perspective of discounted prices, so that only the second order part of $\volatilityOperator$ is considered, with coefficients thought of as functions of the discounted stock price. 
\begin{defi}[``$\alpha$-H\"older volatility operator'']\label{defi.volatilityOperator}
Let $\alpha$ be in the open interval $(0,1)$. An $\alpha$-\emph{H\"older volatility operator} is a second order elliptic differential operator of the form 
\[
\volatilityOperator = \trace \left(\sigmasquare \nabla\squared \right)/2 
=\aij\partialij/2,
\]
where $\sigmasquare = (\aij)_{1\leq i,j\leq d}$ is symmetric and such that  all  coefficients $\aij:\Rd\rightarrow\R$, $1\leq i,j\leq d$, are $\alpha$-H\"older regular. 
\end{defi}
Notice that, for the definition of the classical delta hedging of equation \eqref{Eq.replicatingStrategy}, only the diffusion generator of the market model is relevant, whereas the stochastic base is not. With this respect, we sometimes write ``$\volatilityOperator$-delta hedging'', in order to emphasise that it is defined in terms of the semigroup $e^{t\volatilityOperator}$ of $\volatilityOperator$ as explained in Section \ref{sec.preliminariesAndHeuristics}.

Given an $\alpha$-H\"older volatility operator  $\volatilityOperator = \trace \left(\sigmasquare \nabla\squared \right)/2 $ and a continuous path $X:\timeWindow\rightarrow \Rd$ of finite $p$-variation, we can consider the $\volatilityOperator$-enhancement\footnote{
	Here we are abusing notation: Definition \ref{defi.enhancementOfAPath} prescribed to put the rough bracket $A=A_{s,t}$ in front of the word ``enhancement'', so that actually we should have written $\int_{s}^{t}\sigmasquare(X_u)du$-enhancement. However, the employed notational distortion does not cause confusion and rather stresses the nature of an enhancement of  diffusion-type.
} $\roughPair$ of $X$ given be 
\[
\secondOrderX_{s,t} = \half \big(\Xst \otimes \Xst \big) - \int_{s}^{t} \frac{\sigmasquare(X_u)}{2}du.
\]
Notice that such construction yields a bounded variation enhancement. 
The converse construction, which starts from a bounded variation enhancement  and defines a differential operator, is formalised in the following 

\begin{defi}[``Enhanced path of $\alpha$-diffusion type'']
	Let $\roughPair$ be an enhanced path of $p$-variation regularity. We say that $\roughX$ is of $\alpha$-\emph{diffusion type}, $p-2<\alpha<1$,  if by setting 
	\begin{equation}\label{Eq.lineMeasureInducedByRoughBracket}
	m^{i,j}\Big((s,t]\Big) := \roughBracketijst, \qquad 0\leq s < t \leq T, \quad 1\leq i,j \leq d, 
	\end{equation}
	absolute continuous measures are defined on the interval $\timeWindow$, and if their densities with respect to the Lebesgue measure are given by
	\[
	\frac{dm^{i,j}}{dt} = \aij (X_t), 
	\]
	for some $\sigmasquare=(\aij)_{1\leq i,j \leq d}$ in $C^{\alpha\text{-H\"ol}}_{\text{loc}}(\Rd,\R^{d\times d})$  satisfying the ellipticity condition
	\begin{equation}\label{eq.ellipticityOfVolatilityOperator}
	\aij (x) \xi_i \xi_j \geq c(x) \abs{\xi}^2, \qquad \forall x,\xi \in \Rd,
	\end{equation}
	with some continuous strictly positive $c:\Rd \rightarrow \R_+$. 		
	The operator
	$		\volatilityOperator ^{\roughBracket} :=  \aij  (x) \partialij	\, /2	$
	is called $\roughBracket$-volatility operator, and we say that a diffusive price with Markov generator $\generatorL$ is $\roughBracket$-compatible if the second order part of $\generatorL$ is equal to $\volatilityOperator^{\roughBracket}$.
\end{defi}	

\begin{remark}\label{remark.ellipticityAssumption}
	The ellipticity condition in equation \eqref{eq.ellipticityOfVolatilityOperator} is in place in order to apply the theory from \cite[Chapter 2]{LB07ana} to the existence and uniqueness of semigroups on $C_b(\Rd)$ associated with the volatility operator $\volatilityOperator$. If the solution to the PDE associated with  $\volatilityOperator$ is known to posses a unique solution, the assumed ellipticity can be removed. This is the case for example of the classical Black-Schoels partial differential equation with volatility operator $\volatilityCoefficient\squared  x \squared \partial\squared_{xx} /2$. 
\end{remark}

An enhanced path of $\alpha$-diffusion type is the minimal information that the PDE pricing technology requires from a probabilistic model. Indeed, assume that we wish to use the PDE pricing technology to price a contingent claim $h(X_\timeHorizon)$, where $h$ is in $C_b(\Rd)$ and $X_\timeHorizon$ is the terminal value of a continuous price path $X$ of finite $p$-variation. Let $\roughPair$ be an enhancement of $X$ of $\alpha$-diffusion type and consider the equation 
\begin{equation}\label{eq.minimalPDEtechnology}
\begin{cases}
\Big(\partial_t + \volatilityOperator^{\roughBracket} \Big) \BSsolutionInForwardPrices = 0  & \text{ in } [0,T) \times \Rd \\
\BSsolutionInForwardPrices(T,\cdot)=\htilde(\cdot) & \text{ on } \lbrace T \rbrace \times \Rd . 
\end{cases}
\end{equation}
Then, the Cauchy problem \eqref{eq.minimalPDEtechnology} admits\footnote{
	The existence and regularity of a solution to \eqref{eq.minimalPDEtechnology} is proved for example in \cite[Theorem 2.2.1]{LB07ana}. Recall that the function space $\ContFunctionsOfEllipticPDEregularity$ was defined in equation \eqref{eq.definitionContFunctionsOfEllipticPDEregularity}.
} a solution $\BSsolutionInForwardPrices$ in $\ContFunctionsOfEllipticPDEregularity$ and, on any $\roughBracket$-compatible market model, the value $\BSsolutionInForwardPrices(t,X_t)$ is the discounted price at time $t<\timeHorizon$ of the option maturing at $\timeHorizon$  and yielding $h(X_\timeHorizon)$. 

We are in the position to give the pathwise counterpart to equation \eqref{Eq.integralAccruing}, which is the linchpin of delta hedging. 

\begin{prop}\label{prop.pathwiseIntegralAccruing}
	Let $\roughPair$ be an enhanced path of $\alpha$-diffusion type. Let $\BSsolutionInForwardPrices$ be the solution to \eqref{eq.minimalPDEtechnology} and assume that the pair $(\BSsolutionInForwardPrices,X)$ is $q$-moderate, for some $1-2/p< 1/q<\alpha/p$. Then,
	\[
	(H_t,H\derivative_t):=
	\left(
	\gradx \BSsolutionInForwardPrices (t,\Xt),\Hessianx \BSsolutionInForwardPrices(t,\Xt)
	\right)
	\]
	is a Gubinelli $X$-controlled path of $(p,q)$-variation regularity, and it is such that 
	\begin{equation}\label{eq.pathwiseIntegralAccruing}
	\left((H,H\derivative).(X,\secondOrderX)\right)_{s,t}
	=
	\BSsolutionInForwardPrices(t,\Xt) - \BSsolutionInForwardPrices(s,\Xs)
	\end{equation}
	for all $\zeroSleqTtimeHorizon$.
\end{prop}
\begin{proof}
	The fact that $(H,H\derivative)$ is $X$-controlled follows from Lemma \ref{lemma.qModerationAndControlledPaths}. 
	We can expand the increments of $\BSsolutionInForwardPrices_t := \BSsolutionInForwardPrices(t,\Xt)$ as 
	\begin{align*}
	\BSsolutionInForwardPrices(t,\Xt) - \BSsolutionInForwardPrices(&s,\Xs) \\
	= & 
	(t-s) \int_{0}^{1} \Big[\partial_t \BSsolutionInForwardPrices (s+y(t-s), \Xt) - \partial_t \BSsolutionInForwardPrices (s,\Xt)\Big] dy \\
	& - \frac{t-s}{2} \Big[ \aij (\Xt)\partialij \BSsolutionInForwardPrices (s,\Xt) - \aij (\Xs)\partialij \BSsolutionInForwardPrices (s,\Xs)\Big] \\
	& + \partial_t \BSsolutionInForwardPrices (s,\Xs)(t-s) \\
	& + \Bigg(
	\int_{0}^{1} \int_{0}^{1} \, \Big[ \, \Hessianx \BSsolutionInForwardPrices (s, \Xs + y_1 y_2 \Xst ) \\
	&\qquad \qquad \qquad \qquad - \Hessianx \BSsolutionInForwardPrices (s,\Xs) y_1 \, \Big] \,  dy_2 dy_1
	\Bigg) \big(\Xst \otimes \Xst\big) \\
	&+ \gradx \BSsolutionInForwardPrices (s,\Xs)\Xst + \half \Hessianx \BSsolutionInForwardPrices (s,\Xs) \big(\Xst \otimes \Xst\big).
	\end{align*}
	We have used \eqref{eq.minimalPDEtechnology} on the second line to re-express time derivatives as spatial ones. The assumed $q$-moderation  allows to control the three increment-type summands in the expansion. Let $\convexHull X\timeWindow$ be the convex hull of the trace of $X$ and let $K:=$ $\sup_{0\leq s \leq \timeHorizon}$ $\norm[\Hessianx \BSsolutionInForwardPrices(s,\cdot)]_{\alpha\text{-H\"ol}, \convexHull X\timeWindow}$. Then,
	\begin{align*}
	\Big\lvert
	\int_{0}^{1} \Big[\partial_t \BSsolutionInForwardPrices (s+y(t-s), \Xt) -  \partial_t \BSsolutionInForwardPrices (& s,\Xt)\Big] dy 
	\Big\rvert \\
	\leq &  \norm[\sigmasquare]_{\infty, \, X[0,T]} 
	\Big[
	K \control_{X}^{\alpha/p} + \control_{H\derivative}^{1/q}
	\Big](s,t);
	\end{align*}
	and 
	\begin{align*}
	\Big\lvert
	\aij (\Xt)\partialij \BSsolutionInForwardPrices (s,\Xt) - \aij (\Xs)&\partialij \BSsolutionInForwardPrices ( s,\Xs)
	\Big\rvert \\
	\leq & \norm[\sigmasquare]_{\infty, \, X[0,T]} K\control_{X}^{\alpha/p} (s,t) \\
	& + \norm[H\derivative]_{\infty, \, [0,T]} \norm[\sigmasquare]_{\alpha\text{-H\"ol}, \, \convexHull X[0,T]} \control_{X}^{\alpha/p}(s,t);
	\end{align*}
	 and 
	\begin{align*}
	\Big\lvert
	\Big(\int_{0}^{1} \int_{0}^{1} \Big[\Hessianx \BSsolutionInForwardPrices (s, \Xs + y_1 y_2 \Xst ) - \Hessianx \BSsolutionInForwardPrices & (s,\Xs) y_1 \Big] dy_2 dy_1\Big)
	\big(\Xst \otimes \Xst\big)
	\Big\rvert \\
	\leq & \frac{K}{(1+\alpha)(2+\alpha)} \control_{X}^{(2+\alpha)/p}(s,t).
	\end{align*}
	Recall that, in particular, $\frac{2+\alpha}{p} >1$ by the choice of $\alpha$ in the definition of enhanced path of $\alpha$-diffusion type. Then,  the three estimations above say that, for the expansion  of the increments $\BSsolutionInForwardPrices_{s,t}$, the following holds: there exists a control $\control$ and an exponent $\gamma>1$ such that
	\begin{align*}
	\Big\lvert
	\BSsolutionInForwardPrices_{s,t} - \gradx \BSsolutionInForwardPrices(s,\Xs)&\Xst -  \Hessianx \BSsolutionInForwardPrices(s,\Xs)\secondOrderX_{s,t} \\
	&\quad- \partial_t \BSsolutionInForwardPrices(s,\Xs) - \half \Hessianx \BSsolutionInForwardPrices(s,\Xs) \roughBracket_{s,t}
	\Big\rvert \\
	= & 
	\Big\lvert
	\BSsolutionInForwardPrices_{s,t} - \partial_t \BSsolutionInForwardPrices(s,\Xs)  \\
	&\quad - \gradx \BSsolutionInForwardPrices(s,\Xs)\Xst -  \Hessianx \BSsolutionInForwardPrices(s,\Xs) \big(\Xst \otimes \Xst\big)
	\Big\rvert \\
	\leq& \control^{\gamma}(s,t).
	\end{align*} 
		Hence, 
	\begin{align*}
	\BSsolutionInForwardPrices\subscriptst = 
	\lim _{\abs{\pi}\rightarrow 0} 
	\sum_{u\in \pi\cap [s,t]} \Big[ \partial_t &\BSsolutionInForwardPrices(u, X_u)(\uprime - u) +\half  \partialij \BSsolutionInForwardPrices (u, X_u) \roughBracket^{i,j}\subscriptuuprime \Big] \\
	&+ \underbrace{
		\lim _{\abs{\pi}\rightarrow 0} \underbrace{
			\sum_{u\in \pi\cap [s,t]} \Big[H_u \Xuuprime + H\derivative_u \secondOrderXuuprime \Big].}_{=:((\piH,\piHprime).(X, \secondOrderX))_{s,t}}}_{
		=\left((H,H\derivative).(X,\secondOrderX)\right)\subscriptst
	}
	\end{align*}
	The possibility to split the limit descends from the already-known convergence of $((\piH,\piHprime).(X, \secondOrderX))$ as $\abs{\pi}\rightarrow 0 $. 
	For any $i,j$ the discrete sum $\sum_{u \in \pi}\partialij \BSsolutionInForwardPrices(u,X_u)\roughBracket\subscriptuuprime$ approximates the Stieltjes integral of the continuous function $u\mapsto \partialij \BSsolutionInForwardPrices(u,X_u)$ against the measure $m^{i,j}$ of \eqref{Eq.lineMeasureInducedByRoughBracket}. Hence, in the limit as $\abs{\partition} \rightarrow 0$ it converges to 
	$\intzerot \partialij \BSsolutionInForwardPrices(u,X_u)$ $ \aij(X_u)du$. The cancellation guaranteed by \eqref{eq.minimalPDEtechnology} then implies \eqref{eq.pathwiseIntegralAccruing}.
\end{proof}

The deployment of higher order sensitivities and pathwise integration allows to estimate errors arising from time discretisation of integral quantities. Instances of time discretisation of integral quantities  appear in the costs associated with hedging. Indeed, consider the \emph{cost of financing} of a hedging strategy, defined as
\begin{equation}\label{eq.DefCostOfFinancing}
C_t(\phi) := \phi^0_t \risklessAsset_t + \phi^1_t \pricePath_t - (\phi^0 . \risklessAsset)_t - (\phi^1 . \pricePath)_t ,
\end{equation}
where   $(\phi^0,\phi^1) \in \R\times \Rd$ is the strategy and $\risklessAsset$, $\pricePath$ are respectively the riskless asset and the risky asset.  The symbols $(\phi^0.\risklessAsset)_t$ and $(\phi^1.\pricePath)_t$ denote the time-$t$ marginals of the integral processes of $\phi^0$ and $\phi^1$ respectively against $\risklessAsset$ and $\pricePath$.  Thus, the cost of financing in equation \eqref{eq.DefCostOfFinancing} is the difference between the value of the portfolio at time $t$ and the cost of rebalancing the portfolio during the time window $[0,t]$ in order to follow the hedging strategy. 
If continuous hedging were possible and one were able to take $(\phi^0,\phi^1)=(H^0,H)$ as defined in \eqref{Eq.replicatingStrategy}, then this cost\footnote{
	In the continuous-time abstraction, the term $(\phi^1 . \pricePath)_t$ is to be read as the It\^o integral of the continuous adapted process $\phi^1$ against the continuous semimartingale $S$; the term $(\phi^0 . S^0)_t$ would instead refer to the Lebesgue integral $r\intzerot \phi^0_u e^{ru} du$.
} would  match $\portfolioValue_0 = \BSsolutionInForwardPrices(0,X_0)$, the price at time $t=0$ of the option, on a $\Prob$-full set. We remark that the probability $\Prob$ is the measure of the stochastic base on which in the continuous-time case the It\^o integral $(\phi^1 . S)_t$ would be defined. In practice, the cost of financing has two components: the theoretical price $\portfolioValue_0$ and the cost arising from time discretisation, which is $C_\timeHorizon(\phi) - \portfolioValue_0$. For the latter,  with $\phi$ replaced by the discretisation $(\piH^0, \piH)$  of \eqref{Eq.replicatingStrategy}, we now provide a pathwise estimate that relies on integration bounds. 
Recall that $X$ in Proposition \ref{prop.pathwiseIntegralAccruing} plays the role of the discounted trajectory $\discountedPricePath_{t}=\eminusrt\pricePath_t$.

\begin{corol}\label{Corol.CostOfFinancingDiscountedPerspective}
	Assume the setting of Proposition \ref{prop.pathwiseIntegralAccruing}.
	Let $\control$ be the control function whose $(2/p+1/q)$-th power asserts the approximate additivity of $H_s\Xst + H\derivative_s \secondOrderX_{s,t}$. Along any partition $\partition$ of  $\timeWindow$, the discretised strategy $(\piH^0, \piH)$ stemming from  \eqref{Eq.replicatingStrategy} with $\discountedPricePath=X$  has a cost of financing $C(\piH^0,\piH)$ that is bounded as follows:
	\begin{equation}\label{Eq.costOfFinancingBound}
	\begin{split}
	C_\timeHorizon(\piH^0,&\piH)  \\ \leq &
	\abs{\portfolioValue_0} 
	+e^{rT} \Big( K \control(0,T) \mathrm{osc}(\control, \abs{\pi})^{2/p + 1/q - 1} + \abs{\BSsolutionInForwardPrices_{T-,T}}\Big) \\
	&+ \big\lvert \sum_{\substack{u \in \pi \\ u\derivative < T}} e^{ru\derivative} H\derivative _u \secondOrderX_{u,u\derivative}\big\rvert ,
	\end{split}
	\end{equation}
	where $\mathrm{osc}(\control, \abs{\pi})$ is the modulus of continuity of $\control$ on a scale smaller or equal than the mesh-size of the partition, and  $\BSsolutionInForwardPrices_{T-,T}$ is the difference between $\BSsolutionInForwardPrices(T,X_T)= \htilde (X_T)$ and the discounted value $\BSsolutionInForwardPrices(T-,X_{T-})$ of the option at the second last node of the partition. 
	The path-dependent constant $K$ appearing in the bound is not greater than
	\[
	\frac{1}{1-2^{1-(2/p + 1/q)}} \Big(
	\control_{R^H}^{1/p+1/q}(0,T) \pvarNormInterval[X]{p}{[0,T]} + \pvarNormInterval[H\derivative]{q}{[0,T]}\pvarNormInterval[\secondOrderX]{p/2}{[0,T]}
	\Big),
	\]
	where $\control_{R^{H}}$ is the $pq/(p+q)$-variation control of $H\subscriptst - H\derivative_s\Xst$. 
\end{corol}
\begin{proof}
	Let $\BSsolutionInForwardPrices_t$ be the path $t\mapsto \BSsolutionInForwardPrices(t,\Xt)$. Fix a partition  $\partition$ of $\timeWindow$ and recall the notation in \eqref{eq.notationAboutPartition}. 
	We preliminarily observe that
	\begin{align*}
	(\piHzero.\risklessAsset)_t +(\piH.\pricePath)_t =&
	\sum_{u \in \partition} \big[ \BSsolutionInForwardPrices_u \risklessAsset\subscriptumintuprimemint + H_u\risklessAsset_{\uprime\wedge t}\discountedPricePath\subscriptumintuprimemint  \big] \\
	=& \BSsolutionInForwardPrices_{t\star}\risklessAsset_t - \BSsolutionInForwardPrices_0 + \sum_{u \in \partition}\risklessAsset_{\uprime\wedge t} \big[- \BSsolutionInForwardPrices\subscriptumintuprimemint +H_u\discountedPricePath\subscriptumintuprimemint \big],
	\end{align*}
	where in the second line we have used summation by parts. Then, 
	\begin{align}\label{Eq.rearrangedDiscretisedCostOfFinancing}
	C_t(\piH^0, \piH) = &
	^{\pi}\!\BSsolutionInForwardPrices_t S_t^{0} - \piH_t ^{\pi}\!\discountedPricePath_{t}\risklessAsset_t + \piH_t \pricePath_t - \BSsolutionInForwardPrices_{t\star}\risklessAsset_t\nonumber  \\
	& + \BSsolutionInForwardPrices_0 + \sum_{u \in \pi} \risklessAsset_{\uprime\wedge t}\big[\BSsolutionInForwardPrices\subscriptumintuprimemint - H_u \discountedPricePath\subscriptumintuprimemint\big] \nonumber \\
	= & \risklessAsset_t H_{ t\star }\big( \discountedPricePath_t - \discountedPricePath_{t\star}\big) + \portfolioValue_0 + \sum_{u \in \pi} \risklessAsset_{\uprime\wedge t}\big[\BSsolutionInForwardPrices\subscriptumintuprimemint - H_u \discountedPricePath\subscriptumintuprimemint\big] \nonumber\\
	= & \portfolioValue_0 + \risklessAsset _t \BSsolutionInForwardPrices_{t\star,t} + \sum_{\substack{u \in \pi \\ u\derivative < t}}\risklessAsset_{u\derivative} \big[\BSsolutionInForwardPrices_{u,u\derivative} - H_u \discountedPricePath_{u,u\derivative}\big].
	\end{align}
	By adding and subtracting the compensation, we can apply the Sewing Lemma (Proposition \ref{prop.SewingLemma}) and conclude.
\end{proof}

Until now, we have worked with the identification $\roughX= \enhancedDiscountedPricePath$, i.e. the enhanced  path at hand has represented the actual enhanced path of the discounted stock price. In other words, the market models have been $[\enhancedDiscountedPricePath]$-compatible.  This amounts to considering the square $\sigmasquare = \sigma\sigma\transpose$ of co-volatilities a \emph{true} parameter. In Corollary \ref{Corol.pathwiseRobustness} below, we no longer do so and we distinguish the \emph{modelled} enhancer of $\roughX$ from the \emph{actual} enhancer of $\enhancedDiscountedPricePath$. The only assumption on $\enhancedDiscountedPricePath$ is that it is an enhanced path, i.e. its trace $\discountedPricePath$ is a continuous path of finite $p$-variation, $2<p<3$, and its second order process $\secondOrderX[\discountedPricePath]=(\discountedPricePath \otimes \discountedPricePath - [\enhancedDiscountedPricePath])/2$ is a continuous two-parameter function of finite $p/2$-variation with values in $\Rd\odot\Rd$; the enhancer $[\enhancedDiscountedPricePath]$ is not required to be of bounded variation and the integrals against it will be interpreted as Young integrals.

	\begin{corol}\label{Corol.pathwiseRobustness}
	Let $\roughPair[\discountedPricePath]$ be an enhanced path above the $\Rd$-valued discounted price trajectory $\discountedPricePath$ of $p$-variation regularity. 
	Let $\volatilityOperator$ be an $\alpha$-H\"older volatility operator, with $\alpha>p-2$. Consider the $\volatilityOperator$-enhancement $\roughX=(\discountedPricePath,\secondOrderX)$ of $\discountedPricePath$. If $h$ and $\BSsolutionInForwardPrices$ are as in Proposition \ref{prop.pathwiseIntegralAccruing}, then
	$(H_t,H\derivative_t):=$ $(\gradx \BSsolutionInForwardPrices (t,\discountedPricePath_t),$ $ \Hessianx \BSsolutionInForwardPrices (t,\discountedPricePath_t))$ is a Gubinelli $\discountedPricePath$-controlled path of $(p,q)$-variation regularity and
	\begin{equation} \label{Eq.discountedDavisGammaRobustness}
	\begin{split}
	\htilde(\discountedPricePath_T) - \portfolioValue_0 
	=&
	 \big( (H,H\derivative).(\discountedPricePath,\secondOrderX[\discountedPricePath]) \big)_{0,\timeHorizon}\\
	&+
	 \half (H\derivative .\Big([\enhancedDiscountedPricePath] - \roughBracket \Big) )_{0,\timeHorizon} , 
	\end{split}
	\end{equation}
	where the second summand on the right hand side is a well-defined Young integral. As a consequence, if $\piH$ denotes the strategy obtained by discretising along $\partition$ the $\volatilityOperator$-delta hedging,
	then its cost of financing $C_T(\piHzero, \piH) $ is bounded by 
	\begin{multline}\label{Eq.RobustCostOfFinancing}
	\abs{\portfolioValue_0} + \Big\lvert \sum_{\substack{u \in \pi \\ u\derivative < T}}e^{r\uprime}H\derivative _u \secondOrderX[\discountedPricePath]\subscriptuuprime \Big\rvert \\
	+ e^{rT} \Bigg( K \control(0,T) \mathrm{osc}(\control, \abs{\pi})^{2/p+ 1/q -1} \\
	+  \abs{\BSsolutionInForwardPrices_{T-,T}} + K_{H\derivative}\lVert [\enhancedDiscountedPricePath] - \roughBracket\rVert_{p/2\text{-var}, [0,T]}\Bigg),
	\end{multline}
	where  $\control$, $K$ and $\abs{\BSsolutionInForwardPrices_{T-,T}}$ are as in Corollary \ref{Corol.CostOfFinancingDiscountedPerspective} and 
	\[
	K_{H\derivative} = \frac{2^{-(1-2/p)\squared}}{1-2^{1-(4/p+1/q)}}\norm[H\derivative]_{q\text{-var}, [0,T]} + 2^{-(1-2/p)\squared}\norm[H\derivative]_{\infty, [0,T]}.
	\]	
\end{corol}
\begin{proof}
	The fact that $(\gradx \BSsolutionInForwardPrices (t,\discountedPricePath_t), \Hessianx \BSsolutionInForwardPrices (t,\discountedPricePath_t))$ is  $\discountedPricePath$-controlled of $(p,q)$-variation regularity is already contained in Proposition \ref{prop.pathwiseIntegralAccruing}, because it does not involve the second-order component of $\enhancedDiscountedPricePath$. Also, the Taylor expansion of Proposition \ref{prop.pathwiseIntegralAccruing} yields a control function $\control$ and an exponent $\gamma>1$ such that however chosen a subinterval $[s,t]$ of $[0,T]$, it holds
	\begin{align*}
	\BSsolutionInForwardPrices(t,\discountedPricePath_t) - \BSsolutionInForwardPrices (s,\discountedPricePath_s) = &
	\gradx \BSsolutionInForwardPrices (s,\discountedPricePath_s) \discountedPricePath_{s,t} + \partial_t \BSsolutionInForwardPrices (s,\discountedPricePath_s)(t-s) \\
	& + \half \Hessianx \BSsolutionInForwardPrices (s,\discountedPricePath_s) \discountedPricePath_{s,t} \otimes \discountedPricePath_{s,t} 
	+O (\control^{\gamma}(s,t)) \\
	=& \gradx \BSsolutionInForwardPrices (s,\discountedPricePath_s) \discountedPricePath_{s,t} + \Hessianx \BSsolutionInForwardPrices (s,\discountedPricePath_s) \secondOrderX[\discountedPricePath]_{s,t} 
	\\
	&  + \partial_t \BSsolutionInForwardPrices (s,\discountedPricePath_s)(t-s) + \half \Hessianx \BSsolutionInForwardPrices (s,\discountedPricePath_s) \roughBracket_{s,t}\\
	& + \half \Hessianx \BSsolutionInForwardPrices (s,\discountedPricePath_s)\Big([\enhancedDiscountedPricePath]_{s,t} - \roughBracket_{s,t}\Big)+O (\control^{\gamma}(s,t)).
	\end{align*}
	Therefore, by considering the subintervals $[u,u\derivative]$ of a partition $\partition$ of $[s,t]$, summing over these,  and letting $\abs{\pi} \rightarrow 0$, we obtain
	\begin{equation} \label{Eq.incrementsUnderWrongTrajectory}
	\begin{split}
	\BSsolutionInForwardPrices(t,\discountedPricePath_t) -  \BSsolutionInForwardPrices(s,\discountedPricePath_s)
	=&
	\big( (H,H\derivative).(\discountedPricePath,\secondOrderX[\discountedPricePath]) \big)_{s,t}
	+
	\half (H\derivative .\Big([\enhancedDiscountedPricePath] - \roughBracket \Big) )_{s,t} , 
	\end{split}
	\end{equation}
	and in particular \eqref{Eq.discountedDavisGammaRobustness}. The second summand on the right hand side is a well-defined Young integral because $t\mapsto \Hessianx \BSsolutionInForwardPrices(t,\discountedPricePath_{t})$ is of bounded $q$-variation, $q<p/\alpha$, and $\alpha>p-2$ by assumption.
	
	Write $\BSsolutionInForwardPrices\subscriptst$ for the increments $\BSsolutionInForwardPrices(t,\discountedPricePath_t) - \BSsolutionInForwardPrices(s,\discountedPricePath_s)$, $0\leq s \leq t\leq T$. Owing to \eqref{Eq.incrementsUnderWrongTrajectory}, for every subinterval $[u,\uprime]$ of a partition $\partition$ we can write 
	\begin{align*}
	\BSsolutionInForwardPrices\subscriptuuprime - H_u \discountedPricePath\subscriptuuprime = & 
	\big( (H,H\derivative).(\discountedPricePath,\secondOrderX[\discountedPricePath]) \big)\subscriptuuprime
	 - H_u\discountedPricePath\subscriptuuprime - H\derivative_u \secondOrderX[\discountedPricePath]\subscriptuuprime \\
	& + H\derivative_u \secondOrderX[\discountedPricePath]\subscriptuuprime +
	\half (H\derivative .\Big([\enhancedDiscountedPricePath] - \roughBracket \Big) )\subscriptuuprime.
	\end{align*}
	Therefore, 
	\begin{align*}
	\Big\lvert\sum_{\substack{u \in \pi \\ u\derivative < t}}S^0_{\uprime} \big[\BSsolutionInForwardPrices\subscriptuuprime 
	-&
	 H_u \discountedPricePath \subscriptuuprime\big]\Big\rvert  \\
	\leq &
	e^{rT} K \control(0,T) \mathrm{osc}(\control, \abs{\pi})^{2/p+ 1/q -1} + \Big\lvert \sum_{\substack{u \in \pi \\ u\derivative < t}} e^{r\uprime} H\derivative_u \secondOrderX[\discountedPricePath]\subscriptuuprime \Big\rvert \\
	& + \half	e^{rT} \Big\lVert H\derivative .\Big([\enhancedDiscountedPricePath] - \roughBracket \Big) \Big\rVert_{p/2\text{-var}, [0,T]},
	\end{align*}
	where, by applying the bounds in \cite[Theorem 6.8]{FV10mul} we see
	\begin{align*}
		\Big\lVert H\derivative .\Big([\enhancedDiscountedPricePath] -& \roughBracket \Big)\Big\rVert_{p/2\text{-var}, [0,T]} \\
		\leq &2^{(1-\frac{2}{p})\frac{2}{p}} \Big\lVert[\enhancedDiscountedPricePath] - \roughBracket\Big\rVert_{p/2\text{-var}, [0,T]}  \\
		 &
		\left(
		\frac{1}{1-2^{1-(4/p+1/q)}}\norm[H\derivative]_{q\text{-var}, [0,T]} + \norm[H\derivative]_{\infty, [0,T]}
		\right).
	\end{align*}
	Therefore, by plugging in \eqref{Eq.rearrangedDiscretisedCostOfFinancing}, we conclude.
\end{proof}

\section[Pathwise formulation of fundamental equations of hedging]{Pathwise formulation of fundamental \\ equations of hedging} \label{sec.pathwiseFormulationOfFundamentalEquationsOfHedging}
By adopting the perspective of undiscounted price paths, we recover the classical formulas of Mathematical Finance within our pathwise setting. Given a price path $\pricePath$, we say that a model for $\pricePath$ has been specified when a choice for the enhancement $\roughPair[\pricePath]$ is made. This means choosing the enhancer $[\enhancedPricePath]$, see  Section \ref{sec.pathwiseIntegrals}. We speak of an $\alpha$-\emph{diffusive model specification} if the enhancer is given by
\[
[\enhancedPricePath]_{u,v}^{i,j} = \int_{u}^{v} e^{2rt}\aij (\eminusrt \pricePath_t) dt, \qquad 0\leq u \leq v \leq T, \quad 1\leq i,j \leq d,
\]
where $\aij$, $1\leq i,j\leq d$ are the coefficients of an $\alpha$-H\"older volatility operator and $\interestRate$ is the constant interest rate. In other words, an $\alpha$-diffusive model specification is the undiscounted counterpart to an $\volatilityOperator$-enhancement of some discounted price path, where $\volatilityOperator$ is an $\alpha$-H\"older volatility operator as defined in Definition \ref{defi.volatilityOperator}.

\begin{thm}\label{thm.BSpdeForDiffusiveModelSpecification}
	Let $f(\pricePath_\timeHorizon)$ be a contingent claim, where $f$ is in $C_b(\Rd)$ and $\pricePath_\timeHorizon$ is the terminal value of a continuous $d$-dimensional price path $\pricePath$ of finite $p$-variation. Let $\enhancedPricePair$ be an $\alpha$-diffusive model specification, with $\alpha>p-2$, and let $\volatilityOperator = \aij \partialij/2$ be the corresponding volatility operator. Then, the Black-Scholes partial differential equation
	\begin{equation}\label{eq.fullBSpde}
	\begin{cases}
	{e^{2\interestRate t} \aij (\eminusrt z) \partial_{z^{i},z^{j}}\squared v 
		+\interestRate z^{i} \partial_{z^{i}} v + \partial_t v = \interestRate v 
	} & \text{ in }[0,\timeHorizon)\times \Rd \\
	v(\timeHorizon,z)=f(z) & \text{ on } \lbrace\timeHorizon\rbrace \times \Rd
	\end{cases}
	\end{equation}
	admits a solution $\BSsolutionForPortfolioFunction$ in $\ContFunctionsOfEllipticPDEregularity$ and this solution is unique. Moreover, for every $0\leq t \leq \timeHorizon$, the quantity $\portfolioValue_t:=\BSsolutionForPortfolioFunction(t,\pricePath_t)$ is the fair value at time $t$ of the contingent claim $f(\pricePath_\timeHorizon)$ in the benchmark Markovian model $\benchmarkMarkovianModel$.
\end{thm}
\begin{remark}
	In line with what we argued in Section \ref{sec.preliminariesAndHeuristics}, the statement of Theorem \ref{thm.BSpdeForDiffusiveModelSpecification} shows that  probability only plays a role in the justification of the fairness of the option price from the option buyer's perspective. The justification of the price from the option writer's perspective is instead only based on hedging; see the formulas presented in Proposition \ref{prop.pathwisePortfolioIncrements} below within our pathwise framework. Therefore, the hedging strategy can be formulated in a model that is specified without referring to probabilistic evolutions of the underlying price path $\pricePath$, if the required pathwise information about the price is encoded in  the enhancer $\roughBracketS$. 
\end{remark}
\begin{proof}[Proof of Theorem \ref{thm.BSpdeForDiffusiveModelSpecification}]
	The change of variable $x:=\eminusrt z$ allows to rewrite equation \eqref{eq.fullBSpde} as 
	\begin{equation*}
	\begin{cases}
	\Big(\partial_t  + \volatilityOperator\Big)\BSsolutionInForwardPrices = 0 & \text{ in } [0,\timeHorizon)\times\Rd\\
	\BSsolutionInForwardPrices(\timeHorizon,x)= e^{-\interestRate \timeHorizon}f(e^{+r \timeHorizon}x) & \text{ on } \lbrace \timeHorizon \rbrace \times \Rd,
	\end{cases}
	\end{equation*}
	where $\BSsolutionInForwardPrices(t,x) = \eminusrt \BSsolutionForPortfolioFunction (t,z)$. Therefore, existence, uniqueness and regularity of the solution follow from those of equation \eqref{eq.minimalPDEtechnology}.
	
	Let $\pricep (t,\timeHorizon)$ be the fair value of $f(\pricePath_\timeHorizon)$ in the benchmark Markovian model $((\Omega,\sigmaAlgebra,\probabilityQ$ $ (\sigmaAlgebra_t)_t ), \generatorA )$. This means that the discounted price path $\discountedPricePath$ is thought of as a realisation of a Markov diffusion process on $(\Omega, \sigmaAlgebra, \pricingMeasure)$ with generator $\volatilityOperator$, and such diffusion process is a $\pricingMeasure$-martingale. On the one hand, by the pricing paradigm 
	\begin{align}\label{eq.pricingParadigmSemigroup}
	\pricep(t,\timeHorizon) =&
	 \pricingExpectation
	 \big[
	e^{-\interestRate (\timeHorizon-t)}h(\discountedPricePath_\timeHorizon) \vert \sigmaAlgebra_t
	\big] \nonumber \\
	=& e^{(\timeHorizon-t)\volatilityOperator} h(\discountedPricePath_t),
	\end{align}
	where $h(x):=f(e^{\interestRate \timeHorizon}x)$ and $e^{t\volatilityOperator}$ is the semigroup associated with $\volatilityOperator$.
	On the other hand, the It\^o integral $\discountedPortfolioValue_t:=\int_{0}^{t}\gradz \BSsolutionForPortfolioFunction (u,\pricePath_u)d\discountedPricePath_u$ is such that $\discountedPortfolioValue_\timeHorizon = e^{-\interestRate \timeHorizon}h(\discountedPricePath_\timeHorizon)$, and thus 
	\begin{equation}\label{eq.pricingParadigmReplication}
	\pricep(t,\timeHorizon) =
	e^{\interestRate t }\pricingExpectation
	\big[
	\discountedPortfolioValue_\timeHorizon \vert \sigmaAlgebra_t
	\big] 
	= e^{\interestRate t }\discountedPortfolioValue_t,
	\end{equation}
	because $\discountedPortfolioValue$ is a martingale. Combining \eqref{eq.pricingParadigmSemigroup} and \eqref{eq.pricingParadigmReplication} we obtain the second claim. 
\end{proof}
\begin{prop}\label{prop.pathwisePortfolioIncrements}
	Let $f$ and $\pricePath$ be as in Theorem \ref{thm.BSpdeForDiffusiveModelSpecification}. Let $\enhancedPricePair$ be an $\alpha$-diffusive model specification, with $\alpha>p-2$, and let $\BSsolutionForPortfolioFunction=\BSsolutionForPortfolioFunction(t,z)$ solve equation \eqref{eq.fullBSpde}. If $(\BSsolutionForPortfolioFunction,\pricePath)$ is $q$-moderate, for some $1-2/p<1/q<\alpha/p$, then 
	\[
	\big(\DeltaHedge_t,\GammaHedge_t\big)
	:=\left(
	\gradz \BSsolutionForPortfolioFunction (t,\pricePath_t),
	\Hessianz\BSsolutionForPortfolioFunction(t,\pricePath_t)
	\right)
	\]
	is a Gubinelli $\pricePath$-controlled path of $(p,q)$-variation regularity, and 
	\begin{equation}\label{Eq.pathwisePortfolioIncrements}
	\begin{split}
	\portfolioValue_t-\portfolioValue_0 
	=
	\big(
	(\DeltaHedge,\GammaHedge).&(\pricePath,\secondOrderX[\pricePath])
	\big)_{0,t} \\
	+& \intzerot \big(\portfolioValue_u - \DeltaHedge_u \pricePath_u\big)\, d\risklessAsset_u,
	\end{split}
	\end{equation}
	where $\portfolioValue_t=\BSsolutionForPortfolioFunction(t,\pricePath_t)$ and $\risklessAsset_t= \exp(\interestRate t)$. 
\end{prop}
\begin{proof}
	The proof is analogous to the one of Proposition \ref{prop.pathwiseIntegralAccruing}. Indeed, the same Taylor expansion shows that  for some $\gamma > 1$ and some control function $\control$, on the subintervals $[u,u\derivative]$ of any partition $\partition$, it holds
	\begin{align*}
	\BSsolutionForPortfolioFunction(\uprime, \pricePath_{\uprime} ) - \BSsolutionForPortfolioFunction(u,\pricePath_u) = & 
	\gradz \BSsolutionForPortfolioFunction (u, \pricePath_u) \pricePath\subscriptuuprime + \Hessianz \BSsolutionForPortfolioFunction(u, \pricePath_u) \secondOrderX[\pricePath]\subscriptuuprime \\
	& + \partial_t \BSsolutionForPortfolioFunction(u, \pricePath_u) (\uprime - u) + \half \Hessianz \BSsolutionForPortfolioFunction(u, \pricePath_u)[\enhancedPricePath]\subscriptuuprime \\ & \qquad\qquad \qquad\qquad\qquad \qquad\qquad\qquad + O \big(\control^{\gamma}(u,\uprime)\big).
	\end{align*}
	By applying the operator $\lim_{\abs{\partition}\rightarrow 0}\sum_{u\in \pi}$ to both sides of this expansion, we obtain \eqref{Eq.pathwisePortfolioIncrements} since $\BSsolutionForPortfolioFunction$ solves the Black-Scholes partial differential equation \eqref{eq.fullBSpde}.
\end{proof}

The pathwise differential equation in \eqref{Eq.pathwisePortfolioIncrements} syntactically coincides with the classical Stochastic Differential Equation for the portfolio process in the delta hedging. In addition, the definition of the pathwise integral $(\DeltaHedge,$ $\GammaHedge)$ $.(\pricePath,\secondOrderX[\pricePath])$ explicitly expresses the dependence on the gamma sensitivity, which is not captured by the classical stochastic integral. This provides a theoretical underpinning to the usage of Greeks beyond the leading order delta.

\subsubsection*{Fundamental theorem of derivative trading}

The formulas for pricing and hedging heavily depend on the diffusive model specification. In classical terms of Mathematical Finance, such specification amounts to specifying the diffusion coefficient (volatility) in It\^o's price dynamics. Volatility is not directly observable and consequently a trader is liable to misspecify volatility and to use coefficients that do not faithfully represent the true price dynamics. The \emph{Fundamental Theorem of Derivative Trading} addresses such misspecification. It provides a formula that computes the profit\&loss that a trader incurs into when hedging with the wrong volatility -- a reference for this classical formula is \cite{EJP17fun}. Proposition \ref{prop.fundamentalTheoremDerivativeTrading} contributes to the assessment of model misspecification in two ways: on the one hand, it shows the pathwise nature of the P\&L formula (this aligns with the unifying theme of the section); on the other hand, it provides a generalisation of the classical  P\&L formula. The generalisation consists in removing the assumption that the ``true'' price evolution is governed by an It\^o SDE: it captures the misspecification that arises not just between two diffusive enhancements but between a diffusive enhancement (used by the trader) and a general enhanced path (the ``true'' dynamics).

\begin{prop}[``Fundamental Theorem Of Derivative Trading'']\label{prop.fundamentalTheoremDerivativeTrading}
	Let $f(\pricePath_\timeHorizon)$ be a contingent claim, where $f$ is in $C_b(\Rd)$ and $\pricePath_\timeHorizon$ is the terminal value of a continuous $d$-dimensional price path $\pricePath$ of finite $p$-variation. Let $\roughPairTruePricePath$ be the true enhanced path above the trace $\pricePath$. Let $\enhancedPricePair$ be an $\alpha$-diffusive model specification, $\alpha>p-2$, and let $\volatilityOperator$, $\BSsolutionForPortfolioFunction$, $\DeltaHedge$ and $\GammaHedge$ be as in Proposition \ref{prop.pathwisePortfolioIncrements}. Then, 
	\begin{equation}\label{eq.fundamentalPandLformula}
	P\&L = \tilde{\portfolioValue}_\timeHorizon - f(\pricePath_\timeHorizon) 
	= 
	\half (\GammaHedge. \Big([\enhancedPricePath] - [\enhancedTruePricePath]\Big))_{0,\timeHorizon},
	\end{equation}
	where the integral on the right hand side is a well-defined Young integral, and $\hat{\portfolioValue}_t$ is the value at time $0\leq t \leq\timeHorizon$ of the $\volatilityOperator$-hedging portfolio applied to the true enhancement $\enhancedTruePricePath$, 
	{defined by}
	\begin{equation*}
	\begin{split}
	\hat{\portfolioValue}_t := 
	\BSsolutionForPortfolioFunction(0,\pricePath_0) 
	+ \big( (\DeltaHedge,& \GammaHedge). (\pricePath,\secondOrderTruePricePath)\big)_{0,t} \\
	&+
	\intzerot \Big(\BSsolutionForPortfolioFunction(u,\pricePath_u) - \DeltaHedge_u \pricePath_u \Big) \, d\risklessAsset_u.
	\end{split}
	\end{equation*}
\end{prop}
\begin{remark}
{
    If $\enhancedTruePricePath$ arises from a diffusion model then
    as compensation terms in our integrals
    vanish in probability, our definition of the
    value of the portfolio, $\hat{\portfolioValue}_t$, can
    be justified as a self-financing condition. We will justify this definition
    for general pricing signals in Section \ref{sec.EnlargedStrategies}
    below. 
}

	In order to recognise the extension of the classical Fundamental Theorem of Derivative Trading, we rewrite the Young integral in equation \eqref{eq.fundamentalPandLformula} as 
	\[
	\half \int_{0}^{\timeHorizon} \Hessianz \BSsolutionForPortfolioFunction (t,\pricePath_t) \, 
	d \Big( [\enhancedPricePath]_t - [\enhancedTruePricePath]_t\Big).
	\]
	In the case where $\enhancedTruePricePath$ is a diffusive enhancement, we have that $[\enhancedTruePricePath]_t  = \intzerot e^{2\interestRate u }  $ $ \aij _{\mathrm{true}} (e^{-\interestRate u }\pricePath_u) du$, so that the integral is turned in the familiar form
	\[
	\half \int_{0}^{\timeHorizon} e^{2\interestRate t} \partial\squared_{z^i, z^j} \BSsolutionForPortfolioFunction (t,\pricePath_t) 
	\Big(
	\aij (\eminusrt \pricePath_t) -  \aij _{\mathrm{true}} (e^{-\interestRate t }\pricePath_t)
	\Big)dt.
	\] 
\end{remark}
\begin{remark}
	We remark that our generalisation of the fundamental theorem of derivative trading allows to compare the model performance with the actual rough bracket $[\enhancedTruePricePath]$ of the price trajectory. This quantity is model-independent and directly computable from data using for example the Python package \texttt{iisignature}, available at the link \texttt{https://}\texttt{pypi.org/} \texttt{project/}\texttt{iisignature/}. See also the documentation by J. Reizenstein and B. Graham \cite{RG18iis}. An alternative way to extract the actual rough bracket $[\enhancedTruePricePath]$ from the discretely sampled stream of price data can rely on the convergence result of G. Flint, B. Hambly and T. Lyons  \cite{FHL16dis}, to which Remark \ref{remark.rough_hoff_and_model_specification} below is devoted. 
\end{remark}

\begin{proof}[Proof of Proposition \ref{prop.fundamentalTheoremDerivativeTrading}.]
	We manipulate the Taylor expansion in the proof of Proposition  \ref{prop.pathwisePortfolioIncrements} and, for $0\leq u \leq t \leq T$, we write 
	\begin{align*}
	\BSsolutionForPortfolioFunction(t, S_{t} ) - \BSsolutionForPortfolioFunction(u,S_u) = & 
	\gradz \BSsolutionForPortfolioFunction (u,S_u) S_{u,t} + \Hessianz \BSsolutionForPortfolioFunction(u,S_u) \secondOrderTruePricePath_{u,t} \\
	& + \partial_t \BSsolutionForPortfolioFunction(u,S_u) (t - u) + \half \Hessianz \BSsolutionForPortfolioFunction(u,S_u)[\enhancedPricePath]_{u,t} \\
	& + \half \Hessianz \BSsolutionForPortfolioFunction(u,S_u)\Big([\enhancedTruePricePath]_{u,t} - [\enhancedPricePath]_{u,t}\Big)
	+ O \big({\control}^{{\gamma}}(u,t)\big),
	\end{align*}
	where $\BSsolutionForPortfolioFunction$ is the solution to the $d$-dimensional Black-Scholes partial differential equation \eqref{eq.fullBSpde}, ${\control}$ is a control function and ${\gamma} >1$. We sum over the nodes of a partition and then we let the mesh-size shrink to zero, obtaining \eqref{eq.fundamentalPandLformula}. The good definition of the Young integral of $\GammaHedge$ against $[\enhancedTruePricePath]$ and $[\enhancedPricePath]$ holds as in Corollary \ref{Corol.pathwiseRobustness}.
\end{proof}

\begin{remark}\label{remark.rough_hoff_and_model_specification}
Assume that the price vector  $\price \in \Rd$ is sampled at the points $t_k$, $k = 0,\dots, N$ of a partition $\partition = \lbrace t_k \rbrace$ of the time window $\timeWindow$. Following \cite[Definition 2.1]{FHL16dis}, define the Hoff process associated to the data stream $\lbrace \price_{t_k}: \, t_k \in \partition\rbrace$ as the $\R^{2d}$-valued path $\price^H = (\price^{H,b},\price^{H,f})$ given by
\begin{equation*}
 \price^H_t = 
 \begin{cases}
 (\price_{t_k}, \price_{t_{k+1}} ) & \frac{k}{N}\timeHorizon \leq t < \frac{k+1/2}{N}\timeHorizon \\
 (\price_{t_k}, (1-\alpha(t))\price_{t_{k+1}} + \alpha(t) \price_{t_{k+2}} ) & \frac{k+1/2}{N}\timeHorizon \leq t < \frac{k+3/4}{N}\timeHorizon \\
 ( (1-\beta(t))\price_{t_k} + \beta(t)\price_{t_{k+1}} , \price_{t_{k+2}} ) & \frac{k+3/4}{N}\timeHorizon \leq t < \frac{k+1}{N}\timeHorizon
 \end{cases}
\end{equation*}
where $\alpha$ and $\beta$ are affine functions of $t$ such that $\alpha((k+1/2)\timeHorizon/N)$ $ = \beta((k+3/4)\timeHorizon/N) $ $ =0$ and $\alpha((k+3/4)\timeHorizon/N)$ $ = \beta((k+1)\timeHorizon/N) $ $ =1$. This is a particular choice of linear interpolation of the discrete lead-lag process of $\price$ in which the second (lead) component is updated before the first (lag) component.  Let $(\price^H, \secondOrderS ^ H)$ be the geometric $2$-rough path obtained by enhancing $\price^H$ via standard Stieltjes integration (possible because $\price^H$ is of bounded variation).   G. Flint, B. Hambly and T. Lyons  prove in \cite[Theorem 4.1]{FHL16dis} that, if $\price$ is a semimartingale, then $(\price^H, \secondOrderS ^ H)$ converges to the $2$-rough path
\begin{equation}
(\price^{H,\infty},\secondOrderS^{H,\infty})_{s,t} := 
\exp_2 \left(
\left(\begin{array}{c}\price_{s,t} \\ \price_{s,t}   \end{array}\right)
+
\left(
\begin{array}{c c}
A_{s,t} & A_{s,t} -\half \langle \price \rangle_{s,t} \\
A_{s,t} + \half \langle \price \rangle_{s,t} & A_{s,t}
\end{array}
\right)
\right),
\end{equation}
where $\exp_2$ is the level-2 truncation of the exponential map in the tensor algebra  $T(\R^{2d})$, and $A_{s,t}$ and $\langle \price \rangle_{s,t}$ are respectively the Levy area and the quadratic variation  of the semimartingale $\price$ from time $s$ to time $t$. The convergence happens in the limit as the meshsize of the partition $\partition$ shrinks to zero and with respect to suitable $p$-variation norms (we refer to the original article for the exact way of convergence).

Motivated by G. Flint, B. Hambly and T. Lyons's construction, the actual rough bracket $[\enhancedTruePricePath]$ could be extracted from the available discrete sample of data points along the following lines.

Let $V$ be a vector space, and for $v_1$ and $v_2$ in $V$ let $v=v_1\oplus v_2$ denote their direct sum, which lives in the space $V\oplus V$.  Define the map $q$ given by
\begin{equation}
 \begin{split}
 q:& \left( V \oplus V \right) ^{\otimes 2} \longrightarrow V^{\otimes 2} \\
  & (v_1 \oplus v_2) \otimes (w_1 \oplus w_2) \mapsto v_2 \otimes w_1 - v_1\otimes w_2 ,
 \end{split}
\end{equation}
where $\otimes$ denotes tensor product. Notice that for $v=v_1\oplus v_2$ and $w = w_1 \oplus w_2$ in $V\oplus V$, we have that $q(v \wedge w)= v_2 \odot w_1 - v_1 \odot w_2$, where $v\wedge w$ is the antisymmetric product $(v\otimes w - w\otimes v)/2$ in $(V\oplus V)^{\otimes 2}$, and  $v_i\odot w_j$ is the symmetric product $(v_i\otimes w_j + w_j\otimes v_i)/2$ in $V^{\otimes 2}$. 
Therefore, $q((V\oplus V)^{\wedge 2} ) = V^{\odot 2}$. A proxy for the rough bracket $[\enhancedTruePricePath]$ could then be derived from the data stream  $\lbrace \price_{t_k}: \, t_k \in \partition\rbrace$ by applying the map $q$, defined with $V=\Rd$, to the 2-rough path $(\price^H, \secondOrderS ^ H)$ of the Hoff process associated to such a data stream. 
The aforementioned result of \cite[Theorem 4.1]{FHL16dis} guarantees consistency with the classical semimartingale case, in the limit as the time grid of the sample gets finer and finer. To the best of our knowledge, the limit of the rough path lift of the Hoff process is not understood  in the non-semimartingale case. However, for the practical purpose of estimation from data, the procedure could be applied and the the diffusive  model specification $\roughBracketS$ could be calibrated to match the so-derived $[\enhancedTruePricePath]$, hence minimising the error in equation \eqref{eq.fundamentalPandLformula}.
\end{remark}

\section{Enlarged hedging strategies}\label{sec.EnlargedStrategies}
Given an enhanced price path $\enhancedPricePair$, we interpreted the pathwise integral $(H,H\derivative)$ $.(\pricePath,\secondOrderX[\pricePath])$ as the portfolio trajectory arising from the position $H$ on the risky asset $\pricePath$. 

In this section, we explore the possibility to modify the interpretation of $(H,H\derivative).(\pricePath,$ $\secondOrderX[\pricePath])$. We will not only consider it as representing the values of the position $H$ on $\pricePath$, but we will give a financial interpretation to the compensation $H\derivative\secondOrderX[\pricePath]$ as well. This requires to analyse the mechanics of rebalancing portfolios during hedging periods. Classically, given the partition $\partition$ and the discretised strategy $(\piHzero,\piH)$,\footnote{\label{footnote.LeftSuperscriptPi}Recall that given a (continuous) path $\varphi$ in $\R^m$ and a partition $\pi$ we denote by ${^{\pi}}\!\varphi$  the following piecewise constant caglad approximation:
	\[
	{^{\pi}}\!\varphi_t = \sum_{u \in \pi} \varphi_u \one\left\lbrace t \in (u,\uprime]\right\rbrace.
	\]
} the cost of rebalancing the portfolio from $(u-,u]$ to $(u,\uprime]$ is 
\[
\rebal(u) = \piHzero_{\uprime} \risklessAsset_u + \piH_{\uprime} S_u - \piHzero_{u} \risklessAsset_u - \piH_{u} S_u.
\]
Such discretised strategy is self-financing on the grid $\pi$ if and only if for all $u>0$ in $\pi$ it holds $\rebal(u)=0$, or equivalently if and only if
\[
\Hzero_{\uprime} \risklessAsset_{\uprime} + H_{\uprime} S_{\uprime} - \Hzero_{u} \risklessAsset_u - H_{u} S_u = 
\Hzero_u \risklessAsset\subscriptuuprime + H_u S\subscriptuuprime \qquad \forall u \in \pi \cap [0,T).
\]
Given $t$ in $(0,T]$, set $\pi_t := (\pi \cup \lbrace t \rbrace) \cap [0,t]$. By summing over $u \in \pi_t$, $u<t$, we have 
\[
\Hzero_{t} \risklessAsset_{t} + H_{t} S_{t} - \Hzero_{0} \risklessAsset_0 - H_{0} S_0 =
\sum_{\substack{u \in \pi_t \\ u < t}} \Hzero_u \risklessAsset\subscriptuuprime +
\underbrace{\sum_{\substack{u \in \pi_t \\ u < t}} H_u S\subscriptuuprime}_{= (\piH.S)_t}. 
\]
If $S$ is a semimartingale on $(\Omega, \sigmaAlgebra, \Prob, (\sigmaAlgebra_t)_t)$, then taking the $\Prob$-limit as $\abs{\pi}\rightarrow 0 $ justifies the axiomatic condition \eqref{Eq.selfFinancingContTime}, owing in particular to
\[
\sup \Big\lbrace \limsup_{\abs{\pi}\rightarrow 0} \Prob \left(\Big\lvert (\piH.S)_t - \intzerot HdS \Big\rvert >\epsilon \right): \, \epsilon>0\Big\rbrace = 0.
\]
Here the probabilistic model comes into play to guarantee the convergence of the Riemann sums to the It\^o integral $\intzerot HdS$ of $H$ against the semimartingale $S=S_t(\omega)$, of which the actual price trajectory is thought of as a realisation.

Considering an enhancement $\enhancedPricePath$ of $\pricePath$ and incorporating the appropriate compensation within the rebalancing mechanics, we can refrain from resorting to probability when assessing continuously rebalanced hedging strategies. 

Given a symmetric  $G_t$ in $\R^{d\times d} \cong \Homomorphisms(\Rd \otimes \Rd , \R)$ and a subinterval $[s,t] \subset [0,T]$ we interpret  the real quantity $G_s \secondOrderX[\pricePath]_{s,t}$  as the sum of the payoffs at time $t$ of the $d(d-1)/2$ positions $2G^{i,j}_s = 2G^{j,i}_s$, $1\leq i<j\leq d$, on the swap contracts
\[
\pricePath\subscriptst^i \pricePath\subscriptst^j - {[\enhancedPricePath]}^{i,j}_{s,t}, \qquad 1\leq i<j\leq d,
\]
and of the $d$ positions $G^{i,i}_s$, $1\leq i \leq d$, on the swap contracts
\[
(\pricePath\subscriptst^i)\squared - {[\enhancedPricePath]}^{i,i}_{s,t}, \qquad 1\leq i\leq d.
\]
Hence, for every continuous $\phi_t = (\phi^{0}_t,\phi^{1}_t , \phi^{2}_t ) \in \R \times \Rd \times \R^{d\times d}_{\text{sym}}$ we can interpret 
\[
\piphizero_u \risklessAsset_u + \piphione_u S_u + \piphitwo_u\secondOrderX[\pricePath]_{u-,u}
\]
as the value of our portfolio at time $u$ if on the subinterval $(u-,u]$ we have held $\piphizero_u = \phi^0_{u-}$ positions in cash, $\piphione_u = \phi^1_{u-}$ positions in stocks and $\piphitwo_u = \phi^2_{u-}$ positions in swaps. Strategies that adopt positions in cash, stocks and swaps shall be referred to as \emph{enlarged strategies}. For an enlarged strategy, the rebalancing cost from $(u-,u]$ to $(u,\uprime]$ is 
\begin{equation*}\begin{split}
\rebal(u) = \phi^0_u \risklessAsset_u + \phi^1_u S_u +& \phi^2_u \pricep(u,\uprime) \\
&- \big\lbrace  \phi^0_{u-} \risklessAsset_u + \phi^1_{u-} S_u + \phi^2_{u-} \secondOrderX[\pricePath]_{u-,u} \big\rbrace,
\end{split}
\end{equation*}
where, for $0\leq s < t \leq T$ and $1\leq i\leq j\leq d$, the amount $\pricep^{i,j}(s,t) = \pricep^{j,i}(s,t)$ denotes the (exogenously-given) price at time $s$ of the swap $\secondOrderX[\pricePath]_{s,t}^{i,j}$ with maturity $t$. Notice that, since swap contracts are not primitive financial instruments, in the equation above the payoff $\secondOrderX[\pricePath]_{u-,u}$ at time $u$ is disentangled from the price $\pricep (u,\uprime)$ required at time $u$ to take a unit position on the next swap $\secondOrderX[\pricePath]\subscriptuuprime$.

We assume that the price $\pricep(s,t)$ of the swap contracts $\secondOrderX[\pricePath]_{s,t}$ defines a $\Rd\odot\Rd$-valued function on $\simplex$, null and right-continuous on the diagonal,\footnote{By this we mean: $\pricep(s,s)=\lim_{t\downarrow s}\pricep(s,t) = 0$ for all $0\leq s \leq \timeHorizon$.} and such that $\pricep(s,t)$ is of finite $p/2$-variation. Let $\phi\squared$ be a continuous path of finite $q$-variation on $\Homomorphisms(\Rd\odot\Rd;\R)$, where $q$ and $p/2$ are Young complementary. Then, the integral path
\[
Y_t:=(\phi\squared.\pricep)_{0,t}
\]
exists and represents the accumulated cost in the time interval $[0,t]$ consumed by a continuously rebalanced enlarged strategy in order to adopt the positions $\phi\squared$ on the swap contracts. 

\begin{defi}\label{defi.enlargedDeltaHedging}
	Let $f(\pricePath_\timeHorizon)$ be a contingent claim, where $f$ is in $C_b(\Rd)$ and $\pricePath_\timeHorizon$ is the terminal value of a continuous $d$-dimensional price path $\pricePath$ of finite $p$-variation.  Let $\enhancedPricePair$ be an $\alpha$-diffusive model specification, $\alpha>p-2$, and let $\volatilityOperator$, $\BSsolutionForPortfolioFunction$, $\DeltaHedge$ and $\GammaHedge$ be as in Proposition \ref{prop.pathwisePortfolioIncrements}. Let $C$ be a continuous real valued function on $\timeWindow$. Then, the $C$-\emph{enlarged delta hedging} is the enlarged strategy defined as 
	\begin{equation}\label{eq.enlargedDeltaHedging}
	\begin{split}
	\phi^{0}_t =& C_t \eminusrt - \DeltaHedge_t \pricePath_t \eminusrt -Y_t\eminusrt , \\
	\phi^{1}_t = & \DeltaHedge_t , \\
	\phi\squared_t =& \GammaHedge_t,
	\end{split}
	\end{equation}
	where $Y_t:=(\phi\squared.\pricep)_{0,t}$. 
\end{defi}

A desirable property of a hedging strategy is the self-financing condition, i.e. the fact that the strategy does nor require money to readjust its positions during the hedging period. The following Proposition \ref{prop.zeroRebalancingEnlargedDeltaHedging} gives the explicit formula for $C$ in \eqref{eq.enlargedDeltaHedging} that guarantees a null rebalancing cost of the $C$-enlarged delta hedging. 

\begin{prop}\label{prop.zeroRebalancingEnlargedDeltaHedging}
	The continuous real valued function 
	\begin{equation}\label{eq.cashPositionEnlargedDeltaHedging}
	C_t = \BSsolutionForPortfolioFunction (t, \pricePath_t)
	-\interestRate \intzerot e^{\interestRate(t-u)}Y_u du,
	\end{equation}
	where $Y_t:=(\GammaHedge.\pricep)_{0,t}$, is such that the $C$-enlarged delta hedging has zero cost of continuous rebalancing. 
\end{prop}
\begin{proof}
	We adopt the notation in Definition \ref{defi.enlargedDeltaHedging}. Furthermore, we set 
	\[
	y_t := -r \intzerot e^{r(t-u)}Y_u \, du.
	\]
	We can write
	\begin{equation}\label{Eq.ODEvanishingRebal}
	y_{0,t} - r \intzerot (y_u - Y_u) du = 0.
	\end{equation}
	The cost of rebalancing along a partition $\partition$ is 
	\begin{align*}
	\rebal(u) =& \piphizero_{\uprime} \risklessAsset_u + \piDeltaHedge_{\uprime} S_u + \piGammaHedge_{\uprime}\pricep (u,\uprime) \\
	&-\big\lbrace \piphizero_u \risklessAsset_u + \piDeltaHedge_u S_u + \piGammaHedge_u \secondOrderX[\pricePath]_{u-,u} \big\rbrace \\
	= & C_{u-,u} + \GammaHedge_u \pricep(u,\uprime) - Y_{u-, u} \\
	&-\big\lbrace \phi^{0}_{u-}\risklessAsset_{u-,u} + \DeltaHedge_{u-}S_{u-,u} + \GammaHedge_{u-}\secondOrderX[\pricePath]_{u-,u} \big\rbrace. 
	\end{align*}
		Hence, summing over $u \in \pi_t$, $u>0$, we have 
	\begin{align*}
	\sum_{\substack{u \in \pi_t \\ u>0}} \rebal(u) =& V_{0,t} + y_{0,t} - Y_t + \sum_{u \in \pi_t} \GammaHedge_u \pricep(u,\uprime) - \GammaHedge_0\pricep(0,0\derivative) \\
	&- (\piphizero.\risklessAsset)_t - \big( (\piDeltaHedge, \piGammaHedge) . (S,\secondOrderX[\pricePath])\big)_t. 
	\end{align*}
	In the limit as $\abs{\pi}\rightarrow 0 $ we conclude 
	\begin{align*}
	\lim_{\abs{\pi}\rightarrow 0} \sum_{\substack{u \in \pi_t \\ u>0}} \rebal(u) =&
	V_{0,t} + y_{0,t} - r \intzerot V_u du \\
	&- r \intzerot (y_u - Y_u)du  + r\intzerot \DeltaHedge_u S_u du\\
	 & - \big( (\DeltaHedge,\GammaHedge).(\pricePath,\secondOrderX[\pricePath] )\big)_{0,t}  \\
	= & 0,
	\end{align*}
	owing to \eqref{Eq.pathwisePortfolioIncrements} and \eqref{Eq.ODEvanishingRebal}.
\end{proof}

The classical delta hedging is such that the initial endowment $\portfolioValue_0 = \BSsolutionForPortfolioFunction(0,\pricePath_0)$ is precisely what the replicating strategy requires in order to yield the amount $f(\pricePath_\timeHorizon)$ at maturity $\timeHorizon$. Therefore, the writer of an option invests $\portfolioValue_0 $ in the delta hedging strategy, and such strategy will yield exactly the amount fo money that the buyer of of the option will demand at maturity. Since  delta hedging has no additional costs of financing (i.e. rebalancing the portfolio does not consume money) the writer's profit\&loss is null. For the $C$-enlarged delta hedging in Proposition \ref{prop.zeroRebalancingEnlargedDeltaHedging}, the self-financing condition holds. Therefore, the option writer's P\&L is exclusively given by the cost of replication, namely by the difference between the due payment $f(\pricePath_\timeHorizon)$ and the final value $\phi^{0}_\timeHorizon\risklessAsset_\timeHorizon + \phi^{1}_\timeHorizon \pricePath_\timeHorizon$ of the portfolio. Notice that the latter does not comprise the payoff of the swaps, because such endowments are consumed in the rebalancing process. 
\begin{prop}
	The profit\&loss of the $C$-enlarged delta hedging with $C$ given as in \eqref{eq.cashPositionEnlargedDeltaHedging} is 
	\[
	P\&L = Y_\timeHorizon + \interestRate \int_{0}^{\timeHorizon} e^{\interestRate(\timeHorizon - t)}Y_t\,  dt , 
	\]
	where $Y_t = (\GammaHedge.\pricePath)_{0,t}$. 
\end{prop}
\begin{proof}
	The profit\&loss is given by the difference $P\&L= v(\timeHorizon,\pricePath_\timeHorizon) -  \phi^{0}_{\timeHorizon} \risklessAsset_\timeHorizon + \phi^{1}_{\timeHorizon} \pricePath _{\timeHorizon}$. Hence, the statement follows immediately from the definitions in equation \eqref{eq.enlargedDeltaHedging} with $C$ given as in equation \eqref{eq.cashPositionEnlargedDeltaHedging}. 
\end{proof}

\section{Call options in the Black-Scholes model}\label{sec.BSfromOurPathwisePerspective}
We now consider the classical Black-Scholes model, where the paths of the underlying stock price are modelled as trajectories of a geometric Brownian motion. Hence, the setting is the one presented in Section \ref{sec.preliminariesAndHeuristics}, and we take the dimension $d$ equal to 1. 

The volatility operator $\volatilityOperator$ is 
\begin{equation}\label{eq.oneDimensionalBlackScholesVolOperator}
\volatilityOperator \phi (x) =\frac{\volatilityCoefficient\squared}{2} x\squared \partial_{xx}\squared \phi (x), \qquad \phi \in C\squared (\R), 
\end{equation}
where $\volatilityCoefficient>0$ is the volatility coefficient.

Pricing a European option with payoff $f(\pricePath_\timeHorizon)$ entails to solve the partial differential equation \eqref{Eq.discountedBSpde} where the terminal constraint $\tilde{h}=e^{-r\timeHorizon} h $ appearing in this PDE stands in relation to the payoff function $f$ as expressed in equation \eqref{eq.payoffInForwardPrices}.

The volatility operator in equation \eqref{eq.oneDimensionalBlackScholesVolOperator} is not locally uniformly elliptic, i.e. it does not satisfy the requirement in equation \eqref{eq.ellipticityOfVolatilityOperator}. Therefore, we cannot rely on the theory contained in \cite[Chapter 2]{LB07ana} for the existence and uniqueness of the solution to the pricing equation.    However, such an  equation does have a solution for all continuous and bounded terminal constraint $\tilde{h}$ -- and actually for a larger class of terminal constraints. Indeed, let $w$ be the unknown in equation \eqref{Eq.discountedBSpde} and define $u(t,y):=w(\timeHorizon - t,e^{ \volatilityCoefficient y/\sqrt{2} - \volatilityCoefficient\squared (\timeHorizon - t) /2 })$. Then, $w$ solves equation \eqref{Eq.discountedBSpde} if and only if $u$ solves the equation 
\begin{equation}\label{eq.heatEquation}
\begin{cases}
\partial_t u - \partial\squared_{yy} u = 0  & \text{ in } (0,\timeHorizon] \times \R \\
u(0,y) =g(y)  &  \text{ on } \lbrace 0 \rbrace \times \R,
\end{cases}
\end{equation}
where $g(y) = \tilde{h} (e^{\volatilityCoefficient y/ \sqrt{2} - \volatilityCoefficient\squared \timeHorizon /2 })$. 
Therefore, a change of coordinates brings equation \eqref{Eq.discountedBSpde} into the  heat equation.

As pointed out in Remark \ref{remark.ellipticityAssumption}, the ellipticity was assumed in order to state a set of assumptions under which the solution to the partial differential equation in \eqref{eq.minimalPDEtechnology} exists and is unique. However, in cases where existence and uniqueness is guaranteed without relying on ellipticity, this assumption can be removed without affecting the rest of the discussion. This applies in particular to the log-normal example of this section. 

The solution to equation \eqref{eq.heatEquation} is written as 
\begin{equation*}
\begin{split}
u(t,y) =& \frac{1}{\sqrt{4\pi  t }}\int_{\R} g(\xi) \exp \left\lbrace - \frac{(y-\xi)\squared}{4t} \right\rbrace d\xi 
=  \Expectation [ g (Y_t)], 
\end{split}
\end{equation*}
where the random variable $Y_t$ is normally distributed with mean $y$ and variance $2t$. By using the definition of the initial condition $g$, we have that 
\begin{equation*}
u(t,y) 
= \Expectation[g(Y_t)] 
= \Expectation \left[
		\tilde{h}\left(
				\exp\left(
						\frac{\volatilityCoefficient}{\sqrt{2}}y+\volatilityCoefficient \brownianMotion_t - \frac{\volatilityCoefficient\squared}{2} \timeHorizon
					\right)
				\right)
			\right],
\end{equation*}
where $\brownianMotion_t$ is the time-$t$ marginal of a standard Brownian motion. Hence,
\begin{equation*}
\begin{split}
w(t,x) = & u \left (\timeHorizon - t, \frac{\sqrt{2}}{\volatilityCoefficient} \left(\ln x + \volatilityCoefficient\squared t /2\right) \right) \\
= & \Expectation \left[ 
					\tilde{h} \left(
						 \exp \left(
						 	\ln x + \volatilityCoefficient (\brownianMotion_\timeHorizon - \brownianMotion_t ) 
						 			- \frac{\volatilityCoefficient\squared}{2} (\timeHorizon - t )
						 		\right)
						 	\right)
					\right] 
\\
= & e^{-r\timeHorizon}\Expectation \left[ 
f \left(
\exp \left(
\ln x + \volatilityCoefficient (\brownianMotion_\timeHorizon - \brownianMotion_t ) 
- \frac{\volatilityCoefficient\squared}{2} (\timeHorizon - t ) + r \timeHorizon 
\right)
\right)
\right] ,
\end{split}
\end{equation*}
where on the last line we have used the relation in \eqref{eq.payoffInForwardPrices}. Finally, by recalling the discount relation $v(t,z) = e^{rt} w (t,\eminusrt z)$ we obtain that the price at time $t$ of an European option with payoff $f$ at maturity $\timeHorizon$ is given by the formula 
\begin{equation}\label{eq.BSpriceFormula}
\begin{split}
v(t,z) = &  e^{-r(\timeHorizon - t)} \Expectation \left[ 
f \left(
\exp \left(
\ln z + \volatilityCoefficient (\brownianMotion_\timeHorizon - \brownianMotion_t ) 
+(r - \frac{\volatilityCoefficient\squared}{2}) (\timeHorizon - t ) 
\right)
\right)
\right] ,
\end{split}
\end{equation}
where $z$ denotes the price at time $t$ of the underlying.

The expectations above are used in order to have compact formulas for the solutions  $u$, $v$ and $w$  to the differential equations. However, the quantities $u$, $v$ and $w$ do not descend from a probabilistic framework but only from  parabolic PDEs. 

In our framework, the classical Black-Scholes model is specified by the following enhancer
\begin{equation}\label{eq.BSmodelSpecification}
\roughBracketS_{u,v} =9 \volatilityCoefficient\squared \int_{u}^{v} \pricePath\squared_t dt,
\end{equation}
Under this specification, we now discuss the application of our pathwise framework to the case of European call options, where the payoff is 
\begin{equation}\label{eq.callOptionPayoff}
f(z) = (z-K)_{+},
\end{equation}
for some fixed strike $K>0$. 

This payoff is not bounded, so in principle it is not encompassed by the general discussion above. However, the formula in equation \eqref{eq.BSpriceFormula} extends to payoffs with linear growth and thus to the European call option. In other words, despite the fact that the semigroup associated with the PDE pricing equation was defined on the set $C_b(\R)$, this semigroup extends to a wider class than $C_b(\R)$, hence allowing to treat the European call option. However, we would like also to point out that even if the model specification did not allow for such an extension, pricing European call options could always be reduced to pricing European put options, whose payoff is in  $C_b(\R)$. This is due to the so-called \emph{put-call parity}, which is the following model-free relation between the price at time $t$ of the European call option $C_t$, the price at time $t$ of the European put option $P_t$ and the price at time $t$ of the stock $S_t$:
\begin{equation*}
P_t + S_t = C_t + Ke^{-r(\timeHorizon - t )}. 
\end{equation*}   
Because of this relation, if the price $P_t$ can be computed, then the price $C_t$ follows straightforwardly. 

With the payoff in equation \eqref{eq.callOptionPayoff}, formula \eqref{eq.BSpriceFormula} for the price of the option can be rewritten as 
\begin{equation}\label{eq.BSpriceOfCallOption}
v(t,\pricePath_t)
=\pricePath_t N(d_1 (t,\pricePath_t)) - K e^{-r(\timeHorizon - t)} N(d_2 (t,\pricePath_t)), 
\end{equation}
where $N$ is the cumulative distribution function of the normal distribution and 
\begin{equation*}
\begin{split}
d_1(t,\pricePath_t) =&  \left(\volatilityCoefficient\squared (\timeHorizon - t)\right)^{-\half} \left[\ln (\pricePath_t / K)  + \left(r + \frac{\volatilityCoefficient\squared}{2}\right) (\timeHorizon - t )\right] \\
d_1(t,\pricePath_t) =& d_1(t,\pricePath_t) - \volatilityCoefficient\sqrt{\timeHorizon - t}. 
\end{split}
\end{equation*}

In order to be able to apply Proposition \ref{prop.pathwisePortfolioIncrements}, it remains to discuss  the assumption on the $q$-moderation of the pair $(v,\pricePath)$. Unfortunately, here we see that the non-smoothness of the payoff of the call option (or equivalently of the put option) prevents us from applying directly the results established above. We will discuss this in details  now.

Recall the three conditions in Definition \ref{def.qModeration}. Let $H_t$ and $H\derivative_ t$ be the delta and the gamma sensitivities related to the price in \eqref{eq.BSpriceOfCallOption}, namely
\begin{equation}\label{eq.DeltaCallOption}
H_t = \DeltaHedge_t = \partial_{z} v (t,\pricePath_t) = N(d_1(t,\pricePath_t)) , 
\end{equation}
\begin{equation}\label{eq.GammaCallOption}
H\derivative_t =  \GammaHedge_t = \partial\squared_{zz} v(t,\pricePath) =\frac{N\derivative (d_1(t,\pricePath_t))}{\pricePath_t \volatilityCoefficient \sqrt{\timeHorizon - t}}, 
\end{equation}
where $N\derivative$ denotes the probability density function of the standard normal distribution. The fulfilment of the three conditions in Definition \ref{def.qModeration} depends on the terminal value $\pricePath_\timeHorizon$ ot the price path. Depending on this terminal value we have the following asymptotics as $t\uparrow \timeHorizon$:
\begin{equation}\label{eq.limitAtMaturityExists}
\begin{split}
d_1 (t,\pricePath_t) \sim d_2 (t,\pricePath_t) \sim  (\timeHorizon - t)^{-\half} &\qquad \text{ if } \pricePath_\timeHorizon > K ; \\
d_1 (t,\pricePath_t) \sim d_2 (t,\pricePath_t) \sim -(\timeHorizon - t)^{-\half} &\qquad  \text{ if } \pricePath_\timeHorizon < K .
\end{split}
\end{equation} 
Instead, if $\pricePath_\timeHorizon = K$, then neither $d_1$ nor $d_2$ have a limit as $t\uparrow \timeHorizon$. To see this we use the law of iterated logarithm, which gives a precise statement on the small time asymptotics of the brownian path. In the adopted geometric brownian motion case we have that the terminal value  $\pricePath_\timeHorizon$ is written as
\begin{equation*}
\pricePath_\timeHorizon 
= \pricePath_t \exp \left\lbrace
 \volatilityCoefficient \brownianMotion_\timeHorizon - \volatilityCoefficient\brownianMotion_t 
 + \left( r -\frac{\volatilityCoefficient\squared}{2}\right) (\timeHorizon - t ) \right\rbrace.
\end{equation*}
If $\pricePath_\timeHorizon = K$, then by taking logarithm on both sides of this equation we have 
\begin{equation*}
\ln K - \ln \pricePath_t = \volatilityCoefficient \brownianMotion_\timeHorizon - \volatilityCoefficient\brownianMotion_t  + \left( r -\frac{\volatilityCoefficient\squared}{2} \right) (\timeHorizon - t ). 
\end{equation*}
Hence, as $t\uparrow \timeHorizon$ we have 
\begin{equation}\label{eq.asymptoticsWithIteratedLog}
\frac{\ln K - \ln \pricePath_t}{\volatilityCoefficient \sqrt{\timeHorizon - t}} \sim 
\frac{\brownianMotion_\timeHorizon - \brownianMotion_t}{\sqrt{\timeHorizon - t}}
= \underbrace{\frac{\brownianMotion_\timeHorizon - \brownianMotion_t}{\sqrt{2(\timeHorizon - t)\ln \ln (1/(\timeHorizon - t))}} }_{\limsup = 1; \qquad \liminf = -1} 
\cdot\sqrt{2 \ln\ln \frac{1}{\timeHorizon - t}}. 
\end{equation}
The first factor on the right hand side is such that the limsup as $t\uparrow \timeHorizon$ is equal to $1$, and the liminf is equal to $-1$. Therefore, if $\pricePath_\timeHorizon = K$, then
\begin{equation}\label{eq.limitAtMaturityDoesNotExist}
\begin{split}
\limsup_{t\uparrow \timeHorizon} d_1(t,\pricePath_t) = \limsup_{t\uparrow \timeHorizon} d_2(t,\pricePath_t) = +\infty, \\
\liminf_{t\uparrow \timeHorizon} d_1(t,\pricePath_t) = \liminf_{t\uparrow \timeHorizon} d_2(t,\pricePath_t) = -\infty.
\end{split}
\end{equation}
Because of equation \eqref{eq.limitAtMaturityDoesNotExist}, conditions 1 and 2 in Definition \ref{def.qModeration} will not always be satisfied. Moreover, the singularity at $\timeHorizon$ will also impact condition 3. 

\textbf{Condition 1.} The condition on the extension of $H$ and $H\derivative$ up to the time horizon $\timeHorizon$ depends on the terminal value $\pricePath_\timeHorizon$ of the price path. Assume $\pricePath_\timeHorizon> K$. Then, the terms $d_1(t,\pricePath_t)$ and $d_2(t,\pricePath_t)$ both converge to $+\infty$ as $t\uparrow\timeHorizon$, and thus 
\begin{equation*}
\lim_{t\uparrow\timeHorizon} v(t,\pricePath_t) = \pricePath_\timeHorizon - K, \qquad 
\lim_{t\uparrow\timeHorizon} \partial_{z} v(t,\pricePath_t) = 1, \qquad
\lim_{t\uparrow\timeHorizon} \partial\squared_{zz} v(t,\pricePath_t) = 0.
\end{equation*}
Therefore, in the case $\pricePath_\timeHorizon> K$, $H$ and $H\derivative$ can be continuously extended up to $\timeHorizon$ by setting $H_\timeHorizon = 1$ and $H\derivative_\timeHorizon = 0$. Moreover, notice that $H\derivative$ is of finite $p$-variation on $[0,\timeHorizon)$, and so, since it can be continuously extended up to $\timeHorizon$, it is in fact of finite $p$-variation on $\timeWindow$.

The case $\pricePath_\timeHorizon< K$ is similar. Indeed, in this case, the terms $d_1(t,\pricePath_t)$ and $d_2(t,\pricePath_t)$ both converge to $-\infty$ as $t\uparrow\timeHorizon$, and thus 
\begin{equation*}
\lim_{t\uparrow\timeHorizon} v(t,\pricePath_t) = 0 \qquad 
\lim_{t\uparrow\timeHorizon} \partial_{z} v(t,\pricePath_t) = 0, \qquad
\lim_{t\uparrow\timeHorizon} \partial\squared_{zz} v(t,\pricePath_t) = 0.
\end{equation*}
Therefore, in the case $\pricePath_\timeHorizon< K$, $H$ and $H\derivative$ can be continuously extended up to $\timeHorizon$ by setting $H_\timeHorizon = 0$ and $H\derivative_\timeHorizon = 0$. Moreover, notice that $H\derivative$ is of finite $p$-variation on $[0,\timeHorizon)$, and so, since it can be continuously extended up to $\timeHorizon$, it is in fact of finite $p$-variation on $\timeWindow$.

The case $\pricePath_\timeHorizon = K$ instead does not allow for the extension up to time $\timeHorizon$. Indeed, because of equation \eqref{eq.limitAtMaturityDoesNotExist},  $N(d_1(t,\pricePath_t))$ will not have a limits as $t\uparrow\timeHorizon$. Moreover, $H\derivative$ will converge to $+\infty$ as $t\uparrow\timeHorizon$ because, using equation \eqref{eq.asymptoticsWithIteratedLog}, we have
\begin{equation}\label{eq.asymptoticsOfGamma_atTheMoney}
\frac{N\derivative (d_1(t,\pricePath_t))}{\sqrt{\timeHorizon - t}} 
\sim 
\frac{
	(2\pi)^{-\half}\exp\left\lbrace 
	- \left(\frac{\brownianMotion_\timeHorizon - \brownianMotion_t}{\sqrt{2(\timeHorizon - t)\ln \ln (1/(\timeHorizon - t))}}\right)\squared  
				\cdot{2 \ln\ln \frac{1}{\timeHorizon - t}} 
			\right\rbrace }
	{\sqrt{\timeHorizon - t}}.
\end{equation}

\textbf{Condition 2.} For all $x>0$ the function $t\mapsto N(d_1(t,x))$ is continuously differentiable on $[0,\timeHorizon)$. Moreover, 
\begin{equation*}
\begin{split}
\lim_{t\uparrow\timeHorizon} d_1(t,x) = +\infty & \text{ if } x>K \\
\lim_{t\uparrow\timeHorizon} d_1(t,x) = -\infty & \text{ if } x<K.
\end{split}
\end{equation*} 
Therefore the function $t\mapsto N(d_1(t,x))$ is of bounded variation on $[0,\timeHorizon)$ and can be continuously extended to the closed interval $\timeWindow$ if either $x>K$ or $x<K$. Instead, in the case $x=K$, there is not continuous extension up to $\timeHorizon$.

\textbf{Condition 3.} For all $t<\timeHorizon$ we have that the map 
\begin{equation}\label{eq.mapForCondition3}
x\mapsto \frac{N\derivative (d_1(t,x))}{x\volatilityCoefficient \sqrt{\timeHorizon - t}}
\end{equation}  
is continuously differentiable in $(0,+\infty)$. Indeed, we can compute
\begin{equation*}
\partial_{x} \left[ \frac{N\derivative (d_1(t,x))}{x\volatilityCoefficient \sqrt{\timeHorizon - t}} \right]
= - \frac{\left(d_1(t,x) + \volatilityCoefficient \sqrt{\timeHorizon - t}\right)  N\derivative  (d_1(t,x))}{x\squared \volatilityCoefficient\squared (\timeHorizon - t)}. 
\end{equation*} 
Therefore we have that for all $\epsilon>0$  
\begin{equation*}
\sup_{x\geq \epsilon}\left\lvert \partial_{x} \left[ \frac{N\derivative (d_1(t,x))}{x\volatilityCoefficient \sqrt{\timeHorizon - t}} \right] \right \rvert 
\leq \frac{C}{\epsilon\squared \volatilityCoefficient\squared (\timeHorizon - t)},
\end{equation*}
for a fixed constant $C$ that does not depend on $t$. 
The restriction $x\geq \epsilon$ does not hinder the applicability because stock prices are always strictly positive, hence the path $\pricePath$ is lower bounded by a strictly positive constant. 

Because of the latter estimate, for all $t<\timeHorizon$, the map in equation \eqref{eq.mapForCondition3} is $\alpha$-H\"older for all $0<\alpha<1$. However, the modulus of H\"older continuity goes to infinity as $t$ approaches $\timeHorizon$. This says that condition 3 in Definition \ref{def.qModeration} is not satisfied up to the option maturity, but $\GammaHedge$ controls $\DeltaHedge$ in the sense of Gubinelli only up to a time horizon strictly before the option maturity.

\vspace{0.5cm}

The three conditions assessed above  reveal that the applicability of our pathwise framework based on integration \`a la Gubinelli is hindered by the singularities of the sensitivities when $t$ approaches the option maturity. In principle, one could circumvent this issue  by a smooth approximation of the option payoff that could eliminate the point of non-differentiability; we leave this as a future exercise. Here instead, we comment on what this says about option trading in practice, and on how these singularities, exposed by our pathwise framework, could be regarded as an underpinning of the practicality of option hedging. 

The unstable behaviour of the sensitivities when time is close to maturity is known in practice, in particular in the case of options that are at-the-money (i.e. the underlying has a price equal or very close to the strike). Because of this, it is common to stop the delta hedging before the actual option maturity, and to continue with a simpler strategy as buy-and-hold. This is described by introducing a time horizon $\hat{\timeHorizon}$ smaller than the option maturity $\timeHorizon$; then the main option trading based on Black-Scholes model hinges on the undiscounted
Black-Scholes PDE
 \begin{equation*} 
\begin{cases}
\Big(\partial_t + \generatorL  \Big) (\eminusrt \hat{v}) = 0 &\text{ in } [0,\hat{\timeHorizon}) \times \Rd \\
\hat{v}(\hat{\timeHorizon},z)= v(\hat{\timeHorizon},z) & \text{ on }\lbrace \hat{\timeHorizon} \rbrace \times \Rd ,
\end{cases}
\end{equation*}
where $\generatorL$ is given by
\begin{equation}\label{eq.BSfullGenerator}
 \generatorL \varphi (z) = \half z^2  \sigma^2 \partial^{2}_{z z} \varphi (z) 
 + r z \partial_{z}\varphi(z), \quad \varphi \in C_b(\Rd)\cap C^2(\Rd).
 \end{equation}
\eqref{eq.BSfullGenerator}
and $v=v(t,z)$ was given in equation \eqref{eq.BSpriceOfCallOption}, and $\hat{v}$ denotes the unknown in the PDE. In fact, the functions $v$ and $\hat{v}$ coincide in $[0,\hat{\timeHorizon}] \times \Rd $, but we use a different notation to emphasise that the latter is thought of as the Black-Scholes solution stopped before the option maturity.  The issues encountered above do not apply to the pair $(\hat{v},\pricePath)$, which is $q$-moderate in the sense of Definition \ref{def.qModeration} up to the horizon $\hat{\timeHorizon}$. Hence, our pathwise approach exposes the mathematical features that underpin the mentioned common practice. 

After $\hat{\timeHorizon}$ and in the limit as time approaches $\timeHorizon$, the sensitivity $\GammaHedge$ in equation \eqref{eq.GammaCallOption} no longer controls $\DeltaHedge$ of equation \eqref{eq.DeltaCallOption} in the sense of Gubinelli. This is due to the failure of condition 3 in Definition \ref{def.qModeration} as discussed above. Moreover, in the case of at-the-money options, the gamma sensitivity diverges to infinity as time approaches $\timeHorizon$. This has an impact on the profit\&loss formula of Proposition \ref{prop.fundamentalTheoremDerivativeTrading}, as described in the following proposition.  

\begin{prop}\label{prop.divergentPandL}
	Assume that $\pricePath_\timeHorizon = K$. Consider  the Black-Scholes model specified by the enhancer in equation \eqref{eq.BSmodelSpecification} and consider the rough bracket $[\enhancedTruePricePath]$ of the true price signal. Let $\half <\gamma<1$.  Assume that for all $\epsilon>0$ there exists a partition $\pagebreak$ such that $\lvert \partition \rvert < \epsilon$ and 
	\begin{equation}\label{eq.conditionForDivergentPandL}
		\inf\left\lbrace [\enhancedTruePricePath]\subscriptuuprime - [\enhancedPricePath]\subscriptuuprime : \, u\in \partition \right\rbrace > \epsilon^{1-\gamma}. 
	\end{equation}
	Then, there always exists an arbitrary fine trading grid such that the profit\&loss of the delta hedging on this trading grid diverges to $-\infty$ as time approaches the option maturity. 
\end{prop}

\begin{remark}
	Proposition \ref{prop.divergentPandL} says that, in the case of at-the-money options, if the misspecification of the Black-Scholes model is such that the volatility is underestimated, then there exist trading times when following the  delta hedging  will make the trader incur in unbounded losses. Instead, in the cases of in-the-money and out-the-money options ($\pricePath_\timeHorizon>K$ and $\pricePath_\timeHorizon< K$ respectively), the gamma sensitivity has a limit as time approaches maturity and this limit is zero. Therefore, in these two cases, the Young integral describing profit\&loss can be bounded relying on the integration bounds of Section \ref{sec.pathwiseIntegrals}.  
\end{remark}

\begin{proof}[Proof of Proposition \ref{prop.divergentPandL}]
	Let $\partition$ be a trading grid up to the option maturity.  Consider the approximation of the Young integral in equation \eqref{eq.fundamentalPandLformula} on this trading grid, namely 
	\begin{equation}\label{eq.discreteApproximationOfPandL}
	\sum_{u\in \partition} \GammaHedge_u ([\enhancedPricePath]\subscriptuuprime - [\enhancedTruePricePath]\subscriptuuprime).
	\end{equation}
	The condition in equation \eqref{eq.conditionForDivergentPandL} says that for every $\epsilon>0$ there exists $\partition = \partition (\epsilon)$ such that for all $u$ in $\partition$ it holds
	\begin{equation*}
	 [\enhancedPricePath]\subscriptuuprime - [\enhancedTruePricePath]\subscriptuuprime \leq -\epsilon^{1-\gamma}.
	\end{equation*}
	Hence, if the sum in equation \eqref{eq.discreteApproximationOfPandL} is performed on this partition, then such a sum is upper bounded by
	\begin{equation*}
	-\epsilon^{1-\gamma} \Gamma_{\timeHorizon - },
	\end{equation*}
	 where $\timeHorizon - $ denotes the partition point immediately before the option maturity. By the asymptotics in equation \eqref{eq.asymptoticsOfGamma_atTheMoney}, we see that as $\epsilon \downarrow 0$ the quantity $	-\epsilon^{1-\gamma} \Gamma_{\timeHorizon - }$ goes to $-\infty$.
\end{proof}

\section{Conclusions}\label{sec.conclusions}

In this work, we proposed a technical apparatus for pricing and hedging European options that refrained from using probability. The motivation for our proposal is grounded on the fact that the change of measure in the classical paradigm of martingale pricing entails that only pathwise properties of the physical underlying securities are relevant for the valuation of derivatives. This was constructively shown in Section \ref{sec.disentanglingVolatilityFromMarginalVariances}, where an example was produced in which two  stock dynamics that are probabilistically indistinguishable on arbitrarily fine time grids actually imply arbitrarily different prices for European options written on them.

Our probability-free apparatus hinged on enhanced price paths defined in the spirit of Rough Path Theory. On the one hand, their enhancements are essential for pathwise integration, as discussed in Section \ref{sec.pathwiseIntegrals}. On the other hand, they encapsulate the specification of a model for the valuation of derivatives, carrying the information needed for the hedging (Section \ref{sec.enhancedPathOfDiffusionType}). Moreover, these enhancements allow to assess model misspecification: a P\&L formula for the hedging under `wrong' volatility was proved, generalising the so-called fundamental theorem of derivative trading (Section \ref{sec.pathwiseFormulationOfFundamentalEquationsOfHedging}).

We stated the precise assumptions that allow for the application of Gubinelli integrals in the description of hedging strategies. These assumptions are satisfied in the standard Black-Scholes case of European call and put options only up to a time $\hat{\timeHorizon}$ that strictly precedes the option maturity $\timeHorizon$. On the one hand, this opens the question about suitable approximations for the limiting case as $\hat{\timeHorizon}$ converges to $\timeHorizon$ (without using probability); on the other hand, it provides a mathematical underpinning to some hedging practises linked to unstable option sensitivities, in particular in the at-the-money case. 

Beside  this technical issue,   further possible directions of research emerged from our approach to the classical formulas of Mathematical Finance. Indeed, the present work has adopted the option writer's perspective, whereby the option price is justified as the initial endowment of a self-financing  hedging strategy that replicates the option payoff. This could be complemented with a discussion on arbitrage formulated relying on our enhanced paths, hence adopting the option buyer's perspective. The fact that our enhanced paths extend to trajectories other than semimartingales would make the no-arbitrage arguments suitable for models with transaction costs and other market imperfections. Indeed, in these cases price trajectories are usually less regular than semimartingales. Moreover, we would like to point out that the classical arguments for no-arbitrage under transaction costs   is based on consistent price systems, see  \cite{Gua06arb}, \cite{GRS08con}. This means that the absence of arbitrage is ultimately based on support theorems, hence presenting the opportunity to apply Rough Path Theory, whose application in support-type arguments has proved to be fruitful (see \cite[Chapter 19]{FV10mul}). In this direction, a recent  MSc Thesis at Imperial College London moved the first step (\cite{Pei19rou}).

%
%
%
%

	  \addcontentsline{toc}{section}{References}
	  \bibliographystyle{alpha}
	  \bibliography{bibliography_pvopm}
	  
\appendix

 \section[Disentangling historical and implied volatility]{Disentangling historical and implied volatility}
 \label{sec.disentanglingVolatilityFromMarginalVariances}
 
Brigo and Mercurio  \cite{BM00opt} produced examples of ``alternative continuous-time dynamics for discretely-observed stock prices'', in the following sense: given two distinct standard Black-Scholes processes $\geometricBrownianMotion_1$ and $\geometricBrownianMotion_2$ (i.e. geometric Brownian motions), and given a trading grid $\partition$, they constructively showed the existence of a continuous price dynamics such that the following hold simultaneously 
 \begin{enumerate}
 	\item on the grid $\partition$, all its probabilistic features are those of $\geometricBrownianMotion_1$;
 	\item it prices contingent claims as $\geometricBrownianMotion_2$ does. 
 \end{enumerate}
 However fine the grid might be, such ``alternative dynamics'' exist and they span all the range of no-arbitrage prices. In this respect, the ``alternative dynamics'' are deceptive, because statistically inferred distributional properties on the grid $\partition$ would make a trader prone to use $\geometricBrownianMotion_1$ for pricing and hedging purposes, whereas the ``correct'' volatility would be that of $\geometricBrownianMotion_2$.
This indicates that Black-Scholes pricing technology ignores discretely-observed distributional features of the underlying.
 
 This section is devoted to a reformulation of Brigo and Mercurio's construction. Our formulation emphasises the disentanglement of the concept of implied volatility from the concept of historical volatility (a more precise term could be marginal variance), hence describing Brigo and Mercurio's construction from our pathwise perspective. We present simplified direct proofs that circumvent the original discourse based on the Fokker-Planck equation and on evolutions of marginal laws in finite dimensional manifolds of densities (see for example \cite{Bri00mar}). Moreover, we tackle a limiting case that was hinted at in the original article. Indeed, in \cite{BM00opt} the ``alternative dynamics'' were constructed by patching together different processes, and this operation relied on $\epsilon$-neighbourhoods around the grid points. The authors remarked that the process that emerged in the limit as $\epsilon\downarrow 0 $ exists, but they did not treat this limiting case. By introducing the concept of weak NAP-equivalence, we instead manage to present a neat statement in Proposition \ref{prop.interpolatingHistoricalandImpliedVol} about this limiting case. 
 
 \vspace{0.5cm}
 
 We start by recalling three concepts employed in \cite{BM00opt}. They concern distributional properties of stochastic processes. 
 
 \begin{defi}[``Marginal identity'']
 	Let $X$ and $Y$ be  stochastic processes on the time window $\timeWindow$. We say that $X$ and $Y$ are \emph{marginally identical} if their marginal laws are equal at all times, namely if for all bounded measurable $f$ and  all $0\leq t\leq \timeHorizon$ it holds 
 	\begin{equation}\label{eq.conditionForMarginalIdentity}
 	\Expectation f(X_t) = \Expectation f(Y_t).
 	\end{equation} 
 \end{defi}

\begin{remark}\label{remark.differentProbabilitySpaces}
	The condition in equation \eqref{eq.conditionForMarginalIdentity} refers to the law of the two processes $X$ and $Y$. Hence, it is not actually necessary to suppose that $X$ and $Y$ are defined on the same probability space. We shall emphasise this in Definition \ref{def.NAPindifference} below, where different probability spaces represent different models for stock prices' evolutions. However, assuming that $X$ and $Y$ are defined on the same probability space does not affect generality, because a probability space that accommodates both processes can always be constructed. 
\end{remark}
 
 Let $\partition$ be a partition of  $\timeWindow$, i.e. a finite ordered collection of points in $\timeWindow$ such that the initial time $0$ and the time horizon $\timeHorizon$ are both in $\partition$. Let $t$ be in $\timeWindow$.  We adopt the following notational convention: 
 \begin{equation}\label{eq.partitionsNotationalConvention}
 \begin{split}
 t\derivative:= \inf \lbrace u \in \partition : \, u> t \rbrace ,  \qquad 
 \lfloor t \rfloor := \sup\lbrace u &\in \partition : \, u\leq t\rbrace
 \end{split}
 \end{equation}
 
 \begin{defi}[``$\partition$-Markovianity'']
 	Let $X$ be a stochastic process on $\timeWindow$, and let $(\sigmaAlgebra_t)$ be the minimal filtration generated by $X$. Let $\partition$ be a partition of  $\timeWindow$. We say that $X$ is $\partition$\emph{-Markov} if for all $s$ in $\partition$, all $t\geq s$,  and all bounded measurable $f$, it holds
 	\[
 	\Expectation \left[ f(X_t) \vert \sigmaAlgebra_s \right] 
 	= \Expectation \left[ f(X_t) \vert X_s \right] .
 	\]
 \end{defi}
 
 \begin{defi}[``$\partition$-indistinguishability'']
 	Let $X$ and $Y$ be  stochastic processes on $\timeWindow$, and let $\partition$ be a partition of  $\timeWindow$. We say that $X$ and $Y$ are $\partition$\emph{-indistinguishable} if they are $\partition$-Markov, and if
 	\begin{enumerate}
	\item for all $s$ in $\partition$, the push-forwards of the marginal laws of $X_s$ and of $Y_s$ are equivalent as probability measures on $\R$;
	\item for all $s$ in $\partition$, all $t\geq s$, and for almost every $z$ in $\R$ with respect to the push-forward of the marginal law of $X_s$ (or equivalently of $Y_s$), it holds 
 	\begin{equation}\label{eq.conditionPartitionIndistinguishability}
 	\Expectation [f(X_t) \vert X_s=z] 
 	= \Expectation [f(Y_t) \vert Y_s=z],
 	\end{equation}
 	for all all bounded measurable $f$.
 	\end{enumerate}
 \end{defi}

The condition in equation \eqref{eq.conditionPartitionIndistinguishability} says that, if the the marginals $X_s$ and $Y_s$ are equal, then the laws of the marginals $X_t$ and $Y_t$ coincide. In other words, for all $s$ in $\partition$ and  all $t\geq s$, the law of $X_t$ conditioned on $X_s$ is the same as the law of $Y_t$ conditioned on $Y_s$. 

The same observation as in Remark \ref{remark.differentProbabilitySpaces} applies to the condition in equation \eqref{eq.conditionPartitionIndistinguishability}. Moreover we observe that $\partition$-indistinguishability implies that the transition functions of the discrete-time Markov processes $^{\partition}X_t := X_{\lfloor t \rfloor}$ and $^{\partition}Y_t := Y_{\lfloor t \rfloor}$ are the same. 
 
 The last concept that we introduce  is quintessentially  financial. The acronym \emph{NAP} shall stand for \emph{no-arbitrage pricing}. We fix a deterministic interest rate $\interestRate$ so that to include in every market model the riskless asset $\risklessAsset_t=\risklessAsset_0 \exp(+\interestRate t)$. Given a price process $X$, the forward price of $X$ at time $t$ is defined to be $e^{-rt}X_t$. 
 
 \begin{defi}[``NAP-equivalence'']\label{def.NAPindifference}
 	Let $X$ and $Y$ be positive semimartingales defined respectively on $(\Omega_X,\sigmaAlgebra^X,\Prob^X)$ and $(\Omega_Y,\sigmaAlgebra^Y,\Prob^Y)$. We say that $X$ and $Y$ \emph{induce equivalent pricing kernels} / \emph{are NAP-equivalent} if there exist probability measures $\probabilityQ^X$ and $\probabilityQ^Y$, respectively defined on $(\Omega_X,\sigmaAlgebra^X)$ and $(\Omega_Y,\sigmaAlgebra^Y)$, and equivalent to $\Prob^X$ and  $\Prob^Y$, such that
 	\begin{enumerate}
 		\item the forward prices of $X$ and $Y$ are respectively  $\probabilityQ^X$ and $\probabilityQ^Y$-martingales;
 		\item for all $s$, the push-forwards of the marginal laws of $X_s$ and of $Y_s$ with respect to $\pricingMeasure^X$ and $\pricingMeasure^Y$ are equivalent as probability measures on $\R$;
 		\item  for all $s<t$, and for almost every  $z$ in $\R$ with respect to the push-forward of the marginal law of $X_s$ (or equivalently of $Y_s$),  it holds
 		\begin{equation}\label{eq.NAPindifference}
 		\Expectation_{\probabilityQ^X} \left[ f(X_t) \vert X_s=z \right] 
 		= \Expectation_{\probabilityQ^{Y}} \left[ f(Y_t) \vert Y_s=z \right] ,
 		\end{equation}
 		for all bounded measurable $f$,
 		where $\Expectation_{\probabilityQ^X}$ and $\Expectation_{\probabilityQ^{Y}}$ denote respectively expectation under $\probabilityQ^X$ and under $\probabilityQ^Y$.
 	\end{enumerate}  
 \end{defi}
 
 \begin{example}[``NAP-equivalent geometric Brownian motions and market price of risk'']\label{ex.NAPindifferenceOfGeomBM}
 	Notoriously, if $\mu_i$, $i=1,2$, are two real numbers and $X_t^i$,  $i=1,2$, are price processes following the dynamics
 	\begin{equation}\label{eq.physicalDynamicsBS}
 	dX^i = \mu_i X^i dt + \volatilityCoefficient X^i dW^i,
 	\end{equation}
 	where $\volatilityCoefficient$ is a fixed volatility coefficient and $W^1$, $W^2$ are standard one-dimensional Brownian motions, then $X^1$ and $X^2$ induce indifferent pricing kernels. The change of measure that brings the physical dynamics \eqref{eq.physicalDynamicsBS} into their respective pricing dynamics
 	\begin{equation*}\label{eq.pricingDynamicsBS}
 	dX^i = r X^i dt + \volatilityCoefficient X^i dW^i
 	\end{equation*}
 	is described by 
 	\[
 	\frac{d\probabilityQ^i}{d\Prob^i} = \mathcal{E} \left(-\frac{\mu_i - r}{\volatilityCoefficient}W^i\right),
 	\qquad i=1,2,
 	\]
 	where $\mathcal{E}$ denotes It\^o exponential. The coefficient $(\mu_i - r)/\volatilityCoefficient$ is referred to as \emph{market price of risk}, and it takes the role of the volatility in the dynamics of the Radon-Nykodim derivative of the pricing measure with respect to the physical measure. 
 \end{example}
 
 If $(X,\probabilityQ^X)$ and $(Y,\probabilityQ^Y)$ are time-homogeneous Markov processes, then equation \eqref{eq.NAPindifference} is the equivalence of their transition semigroups. Indeed, under the assumption of time-homogeneous Markovianity, we could replace equation \eqref{eq.NAPindifference} with\footnote{
 	In this formula $p^{X}_t (z,dx)$ denotes the transition function associated with the time-homogeneous Markov process $(X,\probabilityQ^X)$, and $p^{Y}_t (z,dy)$ denotes the transition function associated with the time-homogeneous Markov process $(Y,\probabilityQ^Y)$.}
 \begin{equation*}
 \int f(x)p^{X}_{t-s} (z,dx) = 
 \int f(y)p^{Y}_{t-s} (z,dy), 
 \end{equation*}
 for all bounded measurable $f$ and almost all $z$ in the support of the push-forward of $X_s$.
 This will actually be the case for the discussion to follow, where under the pricing measure the processes $X$ and $Y$ will be geometric Brownian motions. Under the assumption of Markovianity, equation \eqref{eq.NAPindifference} implies that the entire laws of the processes $X$ and $Y$ are the same; in the case of Brownian motions, this forces $X$ and $Y$ to have the same drift and the same diffusion coefficient under the pricing measure, as seen in Example \ref{ex.NAPindifferenceOfGeomBM}. 
 
 Having the geometric Brownian motions in mind, we can establish the following Proposition. It will allow us to patch together processes defined on adjacent time intervals and retain the NAP-equivalence.

 \begin{prop}\label{prop.ConcatenateNAPindifferentProcesses}
 	Let $(X^1,Y^1)$ and $(X^2,Y^2)$ be two pairs of NAP-equivalent price processes. Assume that under the pricing measure, they are time-homogeneous Markov processes, with $X^1$ independent from $X^2$, $Y^1$ independent from $Y^2$, and $X^2_0 = Y^2_0 \equiv 1$. Consider the concatenations
 	\begin{equation*}
 	X_t = \begin{cases}
 	X^1_t & 0\leq t \leq T \\
 	X^1_T X^2_{t-T} & T<t\leq 2T
 	\end{cases}
 	\end{equation*}
 	and 
 	\begin{equation*}
 	Y_t = \begin{cases}
 	Y^1_t & 0\leq t \leq T \\
 	Y^1_T Y^2_{t-T} & T<t\leq 2T.
 	\end{cases}
 	\end{equation*}
 	Then $X$ and $Y$ are NAP-equivalent.
 \end{prop}
 
 \begin{proof}
 	Let $(\Omega,\sigmaAlgebra,\probabilityQ)$ be a probability space that accomodates the processes $(X^i,\probabilityQ^{X^{i}})$ and $(Y^i,\probabilityQ^{Y^{i}})$, $i=1,2$. Firstly, we need to show that for all bounded measurable $f$ and all $0\leq s\leq t \leq 2T$ it holds 
 	\begin{equation}\label{eq.checkNAPindifference}
 	\Expectation [f(X_t)\vert X_s=z]
 	= \Expectation [f(Y_t)\vert Y_s=z],
 	\end{equation}
 	where expectations are computed with respect to $\probabilityQ$. This follows from the fact that, under $\probabilityQ$, the law of $\lbrace \log X_t, 0\leq t \leq 2\timeHorizon\rbrace$ is the same as the law of $\lbrace \log Y_t, 0\leq t \leq 2\timeHorizon\rbrace$. To see this, observe that $\log X $ and $\log Y$ are Markov and that: 1) the laws of $\lbrace \log X_t, 0\leq t \leq \timeHorizon\rbrace$ and of $\lbrace \log Y_t, 0\leq t \leq \timeHorizon\rbrace$ are the same by assumption; 2) the laws of $\lbrace \log X_t, \timeHorizon\leq t \leq 2\timeHorizon\rbrace$ and of $\lbrace \log Y_t, \timeHorizon\leq t \leq 2\timeHorizon\rbrace$ are the same, since they both coincide with the unique law of the Markov process described by the transition semigroup of $\log X^2$ and by the initial distribution $\log X^1_\timeHorizon$. 
 	
 	Secondly, we need to show that the forward prices are $\probabilityQ$-martingales. Again this is clear up to time $t=T$. If $s\geq T$, then 
 	\begin{equation*}
 	\begin{split}
 	\Expectation [e^{-rt}X_t \vert e^{-rs}X_s] =&
 	\Expectation 
 	[\Expectation \Big[e^{-rt}X_t \Big\vert e^{-rT}X^1_T, e^{-r(t-T)}X^2_{s-T}\Big] 
 	\vert e^{-rs}X_s] \\
 	=&\Expectation [ e^{-rT}X^1_T e^{-r(s-T)}X^2_{s-T}\vert e^{-rs}X_s]\\
 	=& e^{-rs}X_s.
 	\end{split}
 	\end{equation*}
 	Finally, if $s<T<t$ then 
 	\begin{equation*}
 	\begin{split}
 	\Expectation [e^{-rt}X_t \vert e^{-rs}X_s] =&
 	\Expectation[e^{-rT}X^1_T \vert e^{-rs}X^1_s]\Expectation[e^{-r(t-T)}X^2_{t-T}] \\
 	= & e^{-rs}X_s.
 	\end{split}
 	\end{equation*}
 	The martingality of $e^{-rt}Y_t$ is either proved analogously, or deduced from that of $e^{-rt}X_t$ and the equivalence in law.  
 \end{proof}
 
 A relaxed version of the concept in Definition \ref{def.NAPindifference} brings to the following
 
 \begin{defi}[``Weak NAP-equivalence'']\label{def.weakNAPindifference}
 	Let $X$ and $Y$ be positive semimartingales, interpreted as price dynamics. We say that $X$ and $Y$  are \emph{weakly NAP-equivalent} if there exist sequences of positive semimartingales  $X^n$ and $Y^n$, $n\geq 1$,  such that 
 	\begin{enumerate}
 		\item for all $t$, the log-prices $\log X^n_t$ and $\log Y^n_t$ converge respectively to $\log X_t$ and $\log Y_t$ in $L\squared(\Prob)$ as $n\uparrow \infty$;
 		\item for all $s$ and $t$, the joint law of $(X^n_s,X^n_t)$ converges to the joint law of $(X_s,X_t)$ as $n\uparrow \infty$, and the  joint law of $(Y^n_s,Y^n_t)$ converges to the joint law of $(Y_s,Y_t)$;
 		\item for every $n$, the processes $X^n$ and $Y^n$ are NAP-equivalent. 
 	\end{enumerate}
 \end{defi}
The sequences $X^n$ and $Y^n$, $n\geq 1$, in the definition above are referred to as reducing sequences for the weakly NAP-equivalent pair $(X,Y)$.

\begin{remark}\label{remark.weakNAPreq1and2_logNormal}
	Consider the log-normal case, where the processes $\log X$, $\log Y$, $\log X^n$ and  $\log Y^n$ from Definition \ref{def.weakNAPindifference} are all Gaussian processes. Then, requirement 1 in the definition actually implies requirement 2. Indeed, the joint law of $(X^n_s,X^n_t)$ converges to the joint law of $(X_s,X_t)$ if and only if the mean and the covariance matrix of the Gaussian vector $(\log X^n_s,\log X^n_t)$ converge to the mean and the covariance matrix of the Gaussian vector  $(\log X_s,\log X_t)$. The fact that $\Expectation [\log X^n_s]$,  $\Expectation [\log X^n_t]$,  $\Expectation [(\log X^n_s)\squared]$ and  $\Expectation [(\log X^n_t)\squared]$ converge respectively to $\Expectation [\log X_s]$,  $\Expectation [\log X_t]$,  $\Expectation [(\log X_s)\squared]$ and  $\Expectation [(\log X_t)\squared]$ follows immediately from the convergence $\log X_u^n \rightarrow \log X_u$ in $\Ltwo (\Prob)$ for all $u$. Moreover, from the same convergence in $\Ltwo (\Prob)$ it follows that the product $ \log X_s^n \log X^n_t $ converges to the product $ \log X_s \log X_t $ in $L^{1} (\Prob)$, whence the covariance $\text{cov}(\log X^n_s, \log X^n_t)$ convergences to the covariance  $\text{cov}(\log X_s, \log X_t)$.
	
	The case for the convergence of $Y^n$ to $Y$ is analogous. 
\end{remark}

The possibility to concatenate NAP-equivalent processes extends immediately to weakly NAP-equivalent processes. 

\begin{corol}\label{corol.concatenateWeaklyNAPindifferentProcesses}
	Let $(X^1,Y^1)$ and $(X^2,Y^2)$ be two pairs of weakly NAP-equivalent processes. Let $(X^{1,n},Y^{1,n})_n$ and $(X^{2,n},Y^{2,n})_n$, $n\geq 1$, be their reducing sequences and assume that for every $n$ the NAP-equivalent processes $(X^{1,n},Y^{1,n})$ and $(X^{2,n},Y^{2,n})$ satisfy the assumptions of Proposition \ref{prop.ConcatenateNAPindifferentProcesses}. Then the concatenations 
	\begin{equation*}
	X_t = \begin{cases}
	X^1_t & 0\leq t \leq T \\
	X^1_T X^2_{t-T} & T<t\leq 2T
	\end{cases}
	\end{equation*}
	and 
	\begin{equation*}
	Y_t = \begin{cases}
	Y^1_t & 0\leq t \leq T \\
	Y^1_T Y^2_{t-T} & T<t\leq 2T.
	\end{cases}
	\end{equation*}
	are weakly NAP-equivalent.
\end{corol}
 
 \vspace{0.5cm}
 
 Having introduced the concepts above, we prepare the construction of  ``alternative  dynamics''.

 Let $\probabilitySpace$ be a probability space and let $\brownianMotion$ be a Brownian motion on it. We consider the probability measure $\Prob$ as fixed and we refer to it as \emph{physical measure}. Let $(\sigmaAlgebra_t)$ be the minimal $\Prob$-completed right-continuous filtration generated by $\brownianMotion$.
 We consider processes defined in the time window $\timeWindow$. Given $\tzero$ in $[0,\timeHorizon[$, the space $L\squared (\Prob \otimes \frac{dt}{\timeHorizon - \tzero})$ $=\Ltwo (\Omega \times [\tzero,\timeHorizon],$ $ \sigmaAlgebra \otimes \mathcal{B}[\tzero,\timeHorizon],$ $\Prob \otimes \frac{dt}{\timeHorizon - \tzero})$ is the space of square integrable random variables on $\Omega \times [\tzero,\timeHorizon]$ with respect to the product measure $\Prob \otimes \frac{dt}{\timeHorizon - \tzero}$, where $dt/(\timeHorizon-\tzero)$ is the normalised Lebesgue measure on $[\tzero,\timeHorizon]$. We use the symbol $\fint dt $ for the integral with respect to such normalised Lebesgue measure. For $\xi$ in $\Ltwo (\Prob\otimes dt/(\timeHorizon - \tzero))$ we set 
 \begin{equation}\label{eq.defTripleNorm}
 \tripleNorm{\xi} := 
 \fint_{\tzero}^{\timeHorizon} \norm[\xi(t)]_{\Ltwo(\Prob)}dt
 \end{equation}
 and we observe 
 \[
 \tripleNorm{\xi} \leq \norm[\xi]_{\Ltwo (\Prob\otimes dt/(\timeHorizon - \tzero))}.
 \]
 Let $\tripleNormClosure$ be the closure of $\Ltwo (\Prob\otimes dt/(\timeHorizon - \tzero))$ with respect to $\tripleNorm{\cdot}$. 
 
 Let $\hurstExponent$ be a strictly positive real number. For $s,t > \tzero$ we introduce the functions
 \begin{equation}\label{eq.defVolterraKernel}
 \volterraKernel (\tzero, s, t) := \left(\frac{s-\tzero}{t-\tzero}\right)^{\hurstExponent-\half}
 \end{equation}
 and
 \begin{equation}\label{eq.defCovarianceKernel}
 \covarianceKernel (\tzero, s, t) := \frac{(t-\tzero)^{\hurstExponent+\half}}{(s-\tzero)^{\hurstExponent-\half}}.
 \end{equation}
 
 We collect  few facts about $\volterraKernel$ and $\covarianceKernel$ in the following two lemmas. The proof are straightforward and omitted. 
 
 \begin{lemma}\label{lemma.propertiesOfVolterraKernel}
 	Consider the function $\volterraKernel$ in equation \eqref{eq.defVolterraKernel}. Then,
 	\begin{enumerate}
 		\item the real-valued function $s\mapsto \volterraKernel (\tzero, s, t)$ is square integrable over the interval $]\tzero,t]$ with 
 		\begin{equation*}
 		\int_{\tzero}^{s}\volterraKernel(\tzero,u,s)\volterraKernel(\tzero,u,t)du = \covarianceKernel (\tzero,t,s)/2\hurstExponent,
 		\end{equation*}
 		for any $\tzero<s\leq t$;
 		\item the real-valued function $(s,t)\mapsto \volterraKernel(\tzero, s, t)$ is square integrable over the simplex $\lbrace \tzero\leq s\leq t\leq \timeHorizon\rbrace$ with 
 		\begin{equation*}
 		\int_{\tzero}^{\timeHorizon} dt \int_{\tzero}^{t}ds \volterraKernel\squared (\tzero,s,t) 
 		= \covarianceKernel\squared (\tzero,\timeHorizon,\timeHorizon)/4\hurstExponent;
 		\end{equation*}
 		\item for all $s\leq t$ it holds $\covarianceKernel(\tzero,t,s)\leq \covarianceKernel (\tzero,t,t)$. 
 	\end{enumerate}
 \end{lemma}
 
 \begin{lemma}\label{lemma.propertiesOfCovarianceKernel}
 	Consider the reciprocal $\covarianceKernel\inverse$ of $\covarianceKernel$, defined as 
 	\[
 	\covarianceKernel\inverse (\tzero,s,t) = \frac{(s-\tzero)^{\hurstExponent-\half}}{(t-\tzero)^{\hurstExponent+\half}}.
 	\]
 	Then,
 	\begin{enumerate}
 		\item the real-valued function $s\mapsto \covarianceKernel\inverse (\tzero,s,t)$ is square integrable over the interval $]\tzero,t]$ with 
 		\begin{equation*}
 		\int_{\tzero}^{s}\covarianceKernel\inverse (\tzero,u,s)\covarianceKernel\inverse (\tzero,u,t)du
 		=\covarianceKernel\inverse(\tzero,s,t)/2\hurstExponent,
 		\end{equation*}
 		for all $\tzero<s\leq t$;
 		\item the real-valued function $(s,t)\mapsto \covarianceKernel\inverse (\tzero,s,t)$ is in $L^{q}(t<s\leq t\leq T)$ for all $1\leq q <2$, but not square integrable over the simplex $\lbrace \tzero\leq s\leq t\leq \timeHorizon\rbrace$, with 
 		\begin{equation*}
 		\int_{\tzero}^{\timeHorizon} dt \int_{\tzero}^{t}ds \covarianceKernel^{-q} (\tzero,s,t)
 		=\covarianceKernel^{2-q} (\tzero,T,T)/(2-q)(q\hurstExponent +1-q/2);
 		\end{equation*}
 		\item for all $0<\epsilon_1 \leq \epsilon_2$ it holds
 		\begin{equation*}
 		\covarianceKernel\inverse (\tzero, t +\epsilon_1, t+\epsilon_2)
 		\leq \covarianceKernel\inverse (\tzero, t , t).
 		\end{equation*}
 	\end{enumerate}
 \end{lemma}
 
 The functions $\volterraKernel$ and $\covarianceKernel$ are used to describe the Gaussian processes $\volterraProcess$ and $\psi$ introduced in the following two lemmata. 
 
 \begin{lemma}\label{lemma.DefOfPsi}
 	The Volterra-type formula
 	\begin{equation*}
 	\psi(t,\tzero,\hurstExponent) = \int_{t_0}^{t}\covarianceKernel\inverse (\tzero,s,t)dW_s
 	\end{equation*}
 	defines a centred Gaussian $(\sigmaAlgebra_t)$-adapted process on $]\tzero,t]$ with covariance structure 
 	\begin{equation*}
 	\Expectation \left[ \psi (s,\tzero,\hurstExponent)\psi(t,\tzero,\hurstExponent) \right]
 	= \covarianceKernel\inverse (\tzero,s,t)/2\hurstExponent, \qquad s\leq t. 
 	\end{equation*}
 	Setting $\psi(\tzero,\tzero,\hurstExponent) :=0$, the adapted process $\lbrace \psi (t,\tzero,\hurstExponent): \, \tzero\leq t \leq \timeHorizon\rbrace$ is a well defined element of $\tripleNormClosure$ and it is approximated with respect to $\tripleNorm{\cdot}$ by the sequence
 	\begin{equation}\label{eq.adaptedApproximationOfRoughVolterraProcess}
 	\psi_{\epsilon}(t):=\int_{\tzero}^{t}\covarianceKernel\inverse (\tzero,s+\epsilon,t+\epsilon) dW_s
 	\end{equation}
 	of elements of $\Ltwo(\Prob\otimes dt/(\timeHorizon-\tzero))$.
 \end{lemma}
 \begin{remark}
 	For every $\epsilon>0$ the process $\psi_{\epsilon}$ of equation \eqref{eq.adaptedApproximationOfRoughVolterraProcess} is a semimartingale adapted to the filtration $(\sigmaAlgebra_t)$ of the Brownian motion $\brownianMotion$.
 \end{remark}
\begin{proof}[Proof of Lemma \ref{lemma.DefOfPsi}]
	Consider the function $g$ in $C^{1,2} (]\tzero,\timeHorizon]\times \R)$ defined as 
	\[
	g(t,x):= \left(t-\tzero\right)^{-\half - \hurstExponent}  x.
	\]
	Let $\epsilon>0$. Consider the centred Gaussian martingale
	\begin{equation}\label{eq.auxiliaryCentredGaussianMArtingale}
	\xi_\epsilon (t) := \int_{\tzero}^{t} (u+\epsilon - \tzero)^{\hurstExponent-\half}dW_u,
	\end{equation}
	and the process
	\begin{equation}\label{eq.nonAdaptedApproximationOfRoughVolterraProcess}
	\begin{split}
	\tilde{\psi}_{\epsilon}(t):=& g(t+\epsilon,\xi_{\epsilon}(t+\epsilon)) \\
	=& \int_{\tzero}^{t+\epsilon}\covarianceKernel\inverse(\tzero,s+\epsilon,t+\epsilon)d\brownianMotion_s.
	\end{split}
	\end{equation}
	Let $0<\epsilon_1\leq \epsilon_2$. We can estimate
	\begin{equation*}
	\begin{split}
	\int_{\tzero}^{\timeHorizon} \lVert\tilde{\psi}_{\epsilon_1}(t) -& \tilde{\psi}_{\epsilon_2}(t)\rVert_{\Ltwo(\Prob)} dt \\
	= & \frac{1}{2\hurstExponent} \int_{\tzero}^{\timeHorizon} \vert \covarianceKernel\inverse (\tzero,t+\epsilon_1,t+\epsilon_1)
	+ \covarianceKernel\inverse (\tzero,t+\epsilon_2,t+\epsilon_2)\\
	&-2  \covarianceKernel\inverse (\tzero,t+\epsilon_1,t+\epsilon_2) \vert ^{1/2} dt \\
	\leq & \frac{1}{2\hurstExponent} \int_{\tzero}^{\timeHorizon} \vert\covarianceKernel\inverse (\tzero,t+\epsilon_1,t+\epsilon_1)
	-\covarianceKernel\inverse (\tzero,t+\epsilon_1,t+\epsilon_2) \vert ^{1/2} dt\\
	& + \frac{1}{2\hurstExponent} \int_{\tzero}^{\timeHorizon} \vert\covarianceKernel\inverse (\tzero,t+\epsilon_2,t+\epsilon_2)
	-\covarianceKernel\inverse (\tzero,t+\epsilon_1,t+\epsilon_2) \vert ^{1/2} dt\\
	\longrightarrow&  0, \qquad \text{ as } \epsilon_1,\epsilon_2 \downarrow 0.
	\end{split}
	\end{equation*}
	We have used dominated convergence with domination 
	\begin{equation*}
	\begin{split}
	\vert\covarianceKernel\inverse (\tzero,t+\epsilon_1,t+\epsilon_1)
	-\covarianceKernel\inverse   (\tzero,t+\epsilon_1,t+\epsilon_2)& \vert ^{1/2} \\
	\leq& 2 \covarianceKernel^{-\half} (\tzero,t,t).
	\end{split}
	\end{equation*}
	Therefore, $\lbrace\psi(t,\tzero,\hurstExponent): \, \tzero\leq t \leq \timeHorizon\rbrace$ exists as limit in $\tripleNormClosure$ and defines a  Gaussian process on $]\tzero,\timeHorizon]$ with the claimed covariance structure. Finally,
	\begin{equation*}
	\tilde{\psi}_{\epsilon}(t) - \psi_{\epsilon}(t)
	= \int_{t}^{t+\epsilon}\covarianceKernel\inverse (\tzero,s+\epsilon,t+\epsilon)d\brownianMotion_s
	\end{equation*}
	and 
	\begin{equation*}
	\Expectation \left(\tilde{\psi}_{\epsilon}(t) - \psi_{\epsilon}(t)\right)\squared =
	(t+\epsilon-\tzero)^{-2\hurstExponent-1}\int_{t}^{t+\epsilon}(s+\epsilon-\tzero)^{2\hurstExponent-1}ds.
	\end{equation*}
	In both cases $0<\hurstExponent<1/2$ and $\hurstExponent\geq 1/2$, we have 
	\begin{equation*}
	\Expectation \left(\tilde{\psi}_{\epsilon}(t) - \psi_{\epsilon}(t)\right)\squared 
	\lesssim \epsilon(t+\epsilon-\tzero)^{-2},
	\end{equation*}
	so that 
	\begin{equation*}
	\begin{split}
	\int_{\tzero}^{\timeHorizon}
	\lVert
	\tilde{\psi}_{\epsilon}(t) - \psi_{\epsilon}(t)
	\rVert_{\Ltwo(\Prob)}
	dt
	\lesssim& \quad
	\epsilon^{1/2}\int_{\tzero}^{\timeHorizon}
	(t+\epsilon-\tzero)\inverse dt \\
	=&\quad \epsilon^{1/2} (\log(\timeHorizon+\epsilon-\tzero)
	-\log\epsilon).
	\end{split}
	\end{equation*}
	The right hand side goes to zero as $\epsilon\downarrow 0$. 
\end{proof}

 \begin{lemma}\label{lemma.DefOfVolterraProcess}
 	The Volterra-type formula
 	\begin{equation*}
 	\volterraProcess(t,\tzero,\hurstExponent) = \int_{\tzero}^{t} \volterraKernel(\tzero,s,t)dW_s
 	\end{equation*}
 	defines a centred Gaussian $(\sigmaAlgebra_t)$-adapted process on $[\tzero,\timeHorizon]$ with covariance structure 
 	\begin{equation*}
 	\Expectation \left[ \volterraProcess (s,\tzero,\hurstExponent)\volterraProcess(t,\tzero,\hurstExponent) \right]
 	= \covarianceKernel(\tzero,t,s)/2\hurstExponent, \qquad s\leq t.
 	\end{equation*}
 	Moreover, $\volterraProcess$ is a semimartingale  and for all $\tzero\leq t \leq \timeHorizon$
 	\begin{equation}\label{eq.semimartDecompositionOfVolterraProcess}
 	\volterraProcess(t,\tzero,\hurstExponent)
 	= W_t - W_{\tzero} + (\half - \hurstExponent) \int_{\tzero}^{t} \psi(s,\tzero,\hurstExponent)ds, 
 	\end{equation}
 	where equality is meant in $\Ltwo(\Prob)$ and $\psi$ was defined in Lemma \ref{lemma.DefOfPsi}.
 \end{lemma}
 
 \begin{proof}
 	Point 2 of Lemma \ref{lemma.propertiesOfVolterraKernel} yields the first claim. We establish the second claim. 
 	Consider the function $f$ in $C^{1,2} (]\tzero,\timeHorizon]\times \R)$ defined as 
 	\[
 	f(t,x):= \left(t-\tzero\right)^{\half - \hurstExponent}  x.
 	\]
 	Let $\epsilon>0$. Consider the processes
 	\begin{equation*}
 	\begin{split}
 	\volterraProcess_\epsilon (t) := &
 	f \Big(t+\epsilon, \xi_\epsilon (t)\Big),
 	\end{split}
 	\end{equation*}
 	where $\xi_\epsilon$ was defined in equation \eqref{eq.auxiliaryCentredGaussianMArtingale}. Since $f$ is twice continuously differentiable on $[\tzero+\epsilon,\timeHorizon]\times\R$, by It\^o's lemma $\volterraProcess_{\epsilon}$ is a semimartingale in $\Ltwo(\Prob\otimes dt/(\timeHorizon-\tzero))$ and 
 	\begin{equation*}
 	\begin{split}
 	\volterraProcess_{\epsilon}(t) = &
 	f(\tzero+\epsilon,\xi_{\epsilon}(\tzero)) 
 	+ \int_{\tzero}^{t}\partial_{x}f(s+\epsilon,\xi_{\epsilon}(s))d\xi_{\epsilon}(s)\\
 	& \qquad+\int_{\tzero}^{t}\partial_t f (s+\epsilon,\xi_{\epsilon}(s)) ds \\
 	=& \brownianMotion_t + \brownianMotion_{\tzero} 
 	+\big(\half - \hurstExponent\big)
 	\int_{\tzero}^{t} \psi_{\epsilon}(s)ds +o_\epsilon (1),
 	\end{split}
 	\end{equation*}
 	where $\psi_{\epsilon}$ was defined in equation \eqref{eq.adaptedApproximationOfRoughVolterraProcess} and $o_\epsilon(1)$ is going to $0$ in $\Ltwo(\Prob)$ as $\epsilon\downarrow 0$. By Minkowski integral inequality, we have that 
 	\begin{equation*}
 	\Big\lVert
 	\int_{\tzero}^{t}
 	\psi_{\epsilon}(s)-\psi(s) \quad ds
 	\Big\rVert_{\Ltwo(\Prob)}
 	\leq 
 	\int_{\tzero}^{t}
 	\lVert
 	\psi_{\epsilon}(s)-\psi(s) \rVert_{\Ltwo(\Prob)}
 	ds.
 	\end{equation*}
 	Therefore, letting $\epsilon\downarrow 0$ yields equation \eqref{eq.semimartDecompositionOfVolterraProcess}.
 \end{proof}
 Consider, for $\epsilon>0$, the process
 \begin{equation}\label{eq.DefVolterraProcessEpsilon}
 \volterraProcess_\epsilon (t,\tzero,\hurstExponent) :=
 \brownianMotion_t - \brownianMotion_{\tzero}
  + (\half - \hurstExponent) \int_{\tzero}^t \psi_\epsilon (s) ds, \qquad \tzero\leq t \leq T,
 \end{equation}
 where $\psi_\epsilon$ was defined in equation \eqref{eq.adaptedApproximationOfRoughVolterraProcess}. The proof above shows that $\volterraProcess_\epsilon(t,\tzero,\hurstExponent)$ converges to $\volterraProcess(t,\tzero,\hurstExponent)$ in $\Ltwo(\Prob)$. Since $\Variance \psi_\epsilon (t) \leq  \covarianceKernel\inverse (\tzero, t+\epsilon, t+\epsilon)/2\hurstExponent$, we have that 
 \[
 \sup_{\tzero\leq t \leq \timeHorizon} \Variance \psi_\epsilon (t) \leq  \epsilon\inverse / 2\hurstExponent,
 \]
 and for $0<\eta<2\hurstExponent \epsilon$ 
 \begin{equation*}
 \sup_{\tzero\leq t \leq \timeHorizon} \Expectation \exp \Big(\eta \psi_{\epsilon}\squared (t)\Big) < \infty.
 \end{equation*}
 This is a Novikov-type condition, see  \cite[Chapter VIII, (1.40) Exercise]{RY99con}. Therefore, for all $\epsilon>0$ there exists a probability $\Prob^{\epsilon}$, equivalent to the physical measure $\Prob$, such that $\volterraProcess_\epsilon(\cdot,\tzero,\hurstExponent)$ is a Brownian motion under $\Prob^{\epsilon}$. More precisely, $\Prob^{\epsilon}$ is given by the formula 
 \begin{multline*}
 \frac{d\Prob^{\epsilon}}{d\Prob}\restrictedto{\sigmaAlgebra_t} = 
 \exp \left( (\hurstExponent - \half) \int_{\tzero}^{t} \psi_{\epsilon} (s) dW_s
 - \half  (\hurstExponent - \half)\squared \int_{\tzero}^{t} \psi_{\epsilon}\squared (s) ds
 \right), \\
 \tzero \leq t \leq \timeHorizon.
 \end{multline*}
 
 \begin{remark}
 	Informally passing to the limit as $\epsilon\downarrow 0 $ in the change of measure above yields the Wick exponential of 
 	\[
 	(\hurstExponent - \half)\int_{\tzero}^t \psi(s,\tzero,\hurstExponent)dW_s.
 	\]
 	Borrowing the terminology introduced in Example \ref{ex.NAPindifferenceOfGeomBM}, we can then refer to $(\hurstExponent -1/2)\psi(t,\tzero,\hurstExponent)$ as the (time-dependent) market price of risk. However, the limiting change of measure is delicate because it entails that some mass is lost; indeed, $\Prob(\int_{\tzero}^{t} \psi\squared (s,\tzero,\hurstExponent) ds = \infty)>0$. The concept of weak NAP equivalence was introduced to circumvent this technical issue and to acheive a neat statement without having to mention the approximating sequence explicitly. Such a neat statement will be contained in Proposition \ref{prop.interpolatingHistoricalandImpliedVol}. Notice that the article \cite{BM00opt} does not make this choice and states everything in terms of approximating sequences. 
 \end{remark}

 Let $\mu$ be a real number, which we fix. With $\volatilityCoefficient$ in $\R_{+}$, we define the line $\linearDrag (t, \tzero,\mu,\sigma)$ as 
 \begin{equation}\label{eq.definitionLinearDrag}
 \linearDrag (t,t_0,\mu,\volatilityCoefficient):=(\mu -\volatilityCoefficient\squared/2)(t-t_0), \qquad t_0, t \in \timeWindow.
 \end{equation}
 The letter $X$ will refer to geometric Brownian motion, defined for $\tzero \leq t \leq \timeHorizon$ as 
 \begin{equation}\label{eq.definitionGeometricBrownianMotion}
 \geometricBrownianMotion(t,t_0,\mu,\volatilityCoefficient) 
 := x_0 \exp 
 \Big(\linearDrag(t,t_0,\mu,\volatilityCoefficient) + \volatilityCoefficient \brownianMotion_{t} - \volatilityCoefficient \brownianMotion_{t_0}\Big).
 \end{equation}
 
 Example \ref{ex.NAPindifferenceOfGeomBM} has shown that $\lbrace X(\cdot,\tzero,\mu,\volatilityCoefficient): \, \mu \in\R \rbrace$ is a family of NAP-equivalent processes, whose pricing dynamics is the one of $X(\cdot,\tzero,r,\volatilityCoefficient)$, with $r$ denoting the fixed interest rate in the market model.
 
 \begin{prop}\label{prop.deceptiveNAPindifference}
 	Let $\mu$ be a real coefficient an let  $\volatilityCoefficient_1$ and $\volatilityCoefficient_2$ be two positive real numbers. 
 	Then, the process\footnote{The drift $\mu$ is suppressed from the notation for $Y$.} 
 	\begin{equation*}
 	\begin{split}
 	Y(t,\tzero,\volatilityCoefficient_1,\volatilityCoefficient_2)
 	= x_0 \exp \Big(
 	\volatilityCoefficient_2 \volterraProcess(t,\tzero,\frac{\volatilityCoefficient_2\squared}{2\volatilityCoefficient_1\squared}) 
 	+ \linearDrag (t,\tzero,\mu,\volatilityCoefficient_1)
 	\Big)&, \\
 	&\qquad \tzero \leq t \leq \timeHorizon, 
 	\end{split}
 	\end{equation*}
 	is simultaneously weakly NAP-equivalent to $X(\cdot,\tzero,\mu, \volatilityCoefficient_2)$ and marginally identical to $X(\cdot,\tzero,\mu, \volatilityCoefficient_1)$.
 \end{prop}
 \begin{remark}\label{remark.blindnessToMarginalVariances}
 	The quadratic variation of $Y$ is 
 	\[
 	[Y]_{\tzero,t} = \volatilityCoefficient_{2}\squared \int_{\tzero}^{t} Y\squared_{s} ds.
 	\]
 	This is the same as the one of $X(\cdot,\tzero,\mu, \volatilityCoefficient_2)$, since 
 	\[
 	[X(\cdot,\tzero,\mu, \volatilityCoefficient_2)]_{\tzero,t} = \volatilityCoefficient_{2}\squared \int_{\tzero}^{t} X\squared_s ds,
 	\]
 	but different from that of $X(\cdot,\tzero,\mu, \volatilityCoefficient_1)$. In this respect, no-arbitrage pricing is sensitive to quadratic variations but blind to marginal variances/historical volatilities. 
 \end{remark}
 \begin{proof}[Proof of Proposition \ref{prop.deceptiveNAPindifference}]
 	For simplicity we take $x_0 = 1$. Consider the process
 	\begin{equation}\label{eq.reducingProcess}
 	Y_\epsilon (t,\tzero,\volatilityCoefficient_1, \volatilityCoefficient_2) 
 	= \exp \Big(
 	\volatilityCoefficient_2 \volterraProcess_{\epsilon}(t,\tzero,\frac{\volatilityCoefficient_2\squared}{2\volatilityCoefficient_1\squared}) 
 	+ \linearDrag (t,\tzero,\mu,\volatilityCoefficient_1)
 	\Big),
 	\end{equation}
 	where $\volterraProcess_{\epsilon}$ was defined in equation \eqref{eq.DefVolterraProcessEpsilon}. We know already that $\log Y_\epsilon(t) \rightarrow \log Y(t)$ in $\Ltwo(\Prob )$. By remark \ref{remark.weakNAPreq1and2_logNormal}, we also have that for all $s<t$ the joint law of $(Y_\epsilon(s), Y_\epsilon(t))$ converges to the joint law of  $(Y(s), Y(t))$. Moreover, there exists a probability measure $\Prob^{\epsilon}$, equivalent to $\Prob$, such that $(\volterraProcess_{\epsilon},\Prob^\epsilon)$ is a Brownian motion, and thus there exists an equivalent $\probabilityQ^\epsilon$ such that $(\log Y_\epsilon, \probabilityQ^\epsilon)$ has the law of $\log X (\cdot, \tzero, r,\volatilityCoefficient_2)$. This shows the asserted weak NAP-equivalence.
 	
 	As for the marginal identity, it suffices to notice that $\log Y(t,\tzero,\volatilityCoefficient_1,\volatilityCoefficient_2)$ is a Gaussian random variable with mean $\linearDrag (t,\tzero,\mu,\volatilityCoefficient_1)$ and variance 
 	\begin{equation*}
 	\begin{split}
 	\volatilityCoefficient_2\squared \Variance
 	\volterraProcess (t,\tzero,\frac{\volatilityCoefficient_1\squared}{2\volatilityCoefficient_2\squared})
 	=& 	\volatilityCoefficient_2\squared 
 	\Big[\covarianceKernel (\tzero,t,t)/2H\Big]_{H=\volatilityCoefficient_2\squared / 2\volatilityCoefficient_1\squared}\\
 	=& \volatilityCoefficient_1\squared (t-\tzero).
 	\end{split}
 	\end{equation*}
 	These are the mean and the variance of the Gaussian random variable $\log X (t,\tzero,\mu,\volatilityCoefficient_1)$.
 \end{proof}
 
 \begin{remark}
 	For $n$ in $\N$ let $Y^n$ be the process $\lbrace Y_{1/n} (t,t_0,\volatilityCoefficient_1,\volatilityCoefficient_2)\rbrace$, where for $\epsilon>0$ the process $Y_{\epsilon} (\cdot, t_0,\volatilityCoefficient_1,\volatilityCoefficient_2)$ was defined in equation \eqref{eq.reducingProcess}. Let $X^n = X(\cdot, t_0, \mu, \volatilityCoefficient_2)$ be a geometric Brownian motion with drift $\mu$ and volatility $\volatilityCoefficient_2$ for all $n$ in $\N$. Then, the proof of Proposition \ref{prop.deceptiveNAPindifference} shows that the pair $(X^n, Y^n)$ is a reducing sequence for the weak NAP-equivalence between the process $Y$ defined in the statement of Proposition \ref{prop.deceptiveNAPindifference} and the geometric Brownian motion $X(\cdot, t_0, \mu, \volatilityCoefficient_2)$. We remark that $X^n$ is the constant sequence equal to $X(\cdot, t_0, \mu, \volatilityCoefficient_2)$, hence for all bounded continuous function $f$ we have that 
 	\begin{equation*}
 	\Expectation_{\pricingMeasure^n} \left[ f(Y^n_T) \vert Y_{t_0} = x_0\right]
 	= \Expectation_{\pricingMeasure} \left[ f(X(T, t_0, \mu, \volatilityCoefficient_2)) \vert X(t_0, t_0, \mu, \volatilityCoefficient_2) = x_0\right],
 	\end{equation*}
 	where $\pricingMeasure^n$ is the pricing measure associated with $Y^n$ and $\pricingMeasure$ is the measure under which $X(\cdot, t_0, \mu, \volatilityCoefficient_2)$ has the law of $X(\cdot, t_0, r, \volatilityCoefficient_2)$. This says that the price of the European option with payoff $f$ at maturity $T$ is the same for all stochastic models $Y^n$, and such a price is equal to the price of that option in the classical Black-Scholes model. 
 \end{remark}

\begin{remark}
	The reducing sequence for the weak NAP-equivalence between $Y$ and the geometric Brownian motion $X(\cdot, t_0, \mu, \volatilityCoefficient_2)$ is constructed from the process $Y_\epsilon$ defined in equation \eqref{eq.reducingProcess}. This process is closely related to the process constructed in \cite{BM00opt}, but it is not the same. Indeed, the process constructed in \cite{BM00opt} is given by $\bar{Y}_\epsilon(t,t_0) = x_0 \exp (\bar{Z}_\epsilon (t,t_0))$, where $\bar{Z}_\epsilon$ is null at $t=t_0$ and follows the dynamics
	\begin{equation*}
	\begin{split}
	\bar{Z}_\epsilon (t,t_0) & \\
	=\quad & \ell (t,t_0,\mu,\volatilityCoefficient_1) \\
	&+\begin{cases}
	\volatilityCoefficient_1 (\brownianMotion_t - \brownianMotion_{t_0}) &  \text{ if } t_0 \leq t < t_0 + \epsilon \\
	\left(\epsilon^{-1}\left(t-t_0\right)\right)^{(\volatilityCoefficient_1\squared - \volatilityCoefficient_2\squared)/2\volatilityCoefficient_1\squared}
 \\ \qquad \cdot\Big[\volatilityCoefficient_1 (\brownianMotion_{t_0 + \epsilon} -  \brownianMotion_{t_0}) 
	+\volatilityCoefficient_2 \int_{t_0 + \epsilon}^{t} \left( \epsilon^{-1} (s-t_0)\right)^{(\volatilityCoefficient_2\squared - \volatilityCoefficient_1\squared)/2\volatilityCoefficient_1\squared} & d\brownianMotion_s \Big] \\
	& \text{ if }  t\geq t_0. 
	\end{cases}
	\end{split}
	\end{equation*}
	This is equation \cite[Equation (3)]{BM00opt}. Both $\log Y_\epsilon$ and $\log \bar{Y}_\epsilon$ are Gaussian processes such that, for all $t$, the marginals $\log Y_\epsilon(t)$ and $\log \bar{Y}_\epsilon(t,t_0)$ converge to $\log Y(t,t_0,\volatilityCoefficient_1,\volatilityCoefficient_2)$ in $\Ltwo (\Prob)$. Moreover, $ \bar{Y}_\epsilon$ is marginally identical to $X(\cdot, t_0, \mu, \volatilityCoefficient_2)$ for all $\epsilon>0$. 
\end{remark}
 
 Let $\partition$ be a partition of the time window $\timeWindow$, and recall the notational convention in equation \eqref{eq.partitionsNotationalConvention}. For $u$ in $\partition$ consider the process
 \begin{equation}\label{eq.defDeceptiveArithmeticBM} 
 \deceptiveArithmeticBM (t,u,\volatilityCoefficient_1,\volatilityCoefficient_2) 
 =\begin{cases}
 0 & 0\leq t \leq u\\
 \linearDrag (t,u,\mu,\volatilityCoefficient_1)
 +\volatilityCoefficient_2 \volterraProcess(t,u,\frac{\volatilityCoefficient_2\squared}{2\volatilityCoefficient_1\squared})
 & u<t\leq \uprime \\
 \linearDrag (\uprime,u,\mu,\volatilityCoefficient_1)
 +\volatilityCoefficient_2 \volterraProcess(\uprime,u,\frac{\volatilityCoefficient_2\squared}{2\volatilityCoefficient_1\squared})
 & t> \uprime .
 \end{cases}
 \end{equation}

 \begin{prop}[{\cite[Propositions 2.1 and 2.2]{BM00opt}}]\label{prop.interpolatingHistoricalandImpliedVol}
 	Let $\volatilityCoefficient_1$ and $\volatilityCoefficient_2$ be two positive real numbers, and correspondingly consider the geometric Brownian motions $X(\cdot,\volatilityCoefficient_i)=X(\cdot,0,\mu,\volatilityCoefficient_i)$, $i=1,2$, defined in equation \eqref{eq.definitionGeometricBrownianMotion}, for some $\mu$ in $\R$. Let $\partition$ be a time grid in the time window $\timeWindow$, and correspondingly define the processes $\deceptiveArithmeticBM$ as in equation \eqref{eq.defDeceptiveArithmeticBM}. Let $Y=Y(\cdot,\volatilityCoefficient_1,\volatilityCoefficient_2)$ be the process 
 	\begin{equation*}
 	\begin{split}
 	Y(t,\volatilityCoefficient_1,\volatilityCoefficient_2) =
 	x_0 \exp
 	\sum_{u\in\partition} \deceptiveArithmeticBM(t,u,\volatilityCoefficient_1,\volatilityCoefficient_2), \\
 	&\qquad 0\leq t \leq \timeHorizon.
 	\end{split}
 	\end{equation*}
 	Then, it simultaneously holds
 	\begin{enumerate}
 		\item $Y(\cdot,\volatilityCoefficient_1,\volatilityCoefficient_2)$ and $X(\cdot,\volatilityCoefficient_1)$ are $\partition$-indistinguishable;
 		\item $Y(\cdot,\volatilityCoefficient_1,\volatilityCoefficient_2)$ and $X(\cdot,\volatilityCoefficient_2)$ are weakly NAP-equivalent. 
 	\end{enumerate}
 \end{prop}
 \begin{proof}
 	We split the proof in two parts, which correspond to the statements. 
 	\begin{enumerate}
 		\item Let $u$ be a partition point and observe that for $t>u$ the variable $\log Y(t,\volatilityCoefficient_1,\volatilityCoefficient_2) - \log Y(u,\volatilityCoefficient_1,\volatilityCoefficient_2)$ is independent from $\sigmaAlgebra_u$. Moreover, 
 		\begin{equation*}
 		\begin{split}
 		\log Y(t,\volatilityCoefficient_1,\volatilityCoefficient_2) - \log& Y(u,\volatilityCoefficient_1,\volatilityCoefficient_2)\\
 		=& \underbrace{\sum_{\substack{v\in\partition \\ u\leq v < \lfloor t \rfloor }}\log \frac{Y(v\derivative,\volatilityCoefficient_1,\volatilityCoefficient_2)}{Y(v,\volatilityCoefficient_1,\volatilityCoefficient_2)}}_{\mathtt{L1}}
 		+ \underbrace{\log \frac{Y(t,\volatilityCoefficient_1,\volatilityCoefficient_2)}{Y(\lfloor t \rfloor ,\volatilityCoefficient_1,\volatilityCoefficient_2)}}_{\mathtt{L2}}.
 		\end{split}
 		\end{equation*}
 		The two summands $\mathtt{L1}$ and $\mathtt{L2}$ are independent. The second summand, $\mathtt{L2}$, is normally distributed with mean $\linearDrag(t,\lfloor t \rfloor , \mu , \volatilityCoefficient_1 )$ and variance $\volatilityCoefficient_1\squared (t-\lfloor t \rfloor)$. As for the first summand $\mathtt{L1}$, we further notice the independence of the variables $\log [{Y(v\derivative,\volatilityCoefficient_1,\volatilityCoefficient_2)}/{Y(v,\volatilityCoefficient_1,\volatilityCoefficient_2)}]$, $v \in \partition$, which are normally distributed with mean $\linearDrag(v\derivative, v, \mu , \volatilityCoefficient_1)$ and variance $\volatilityCoefficient_1\squared (v\derivative - v)$. Therefore, $\mathtt{L1}$ is normally distributed with mean $\linearDrag(\lfloor t \rfloor, u, \mu, \volatilityCoefficient_1 )$ and variance $\volatilityCoefficient_1\squared (\lfloor t \rfloor - u)$. Hence, $\log[ Y(t,\volatilityCoefficient_1,\volatilityCoefficient_2) /Y(u,\volatilityCoefficient_1,\volatilityCoefficient_2)]$ is distributed as $\log[ X(t,\volatilityCoefficient_1)/X(u,\volatilityCoefficient_1)]$.
 		\item On each subinterval $[u,u\derivative]$ of $\partition$, the processes 
 		\[
 		\left\lbrace 
 		\frac{Y(t,\volatilityCoefficient_1,\volatilityCoefficient_2)}{Y(u,\volatilityCoefficient_1,\volatilityCoefficient_2)}: \, u\leq t \leq u\derivative 
 		\right \rbrace
 		\]
 		and
 		\[
 		\left\lbrace 
 		\frac{X(t,\volatilityCoefficient_2)}{X(u,\volatilityCoefficient_2)} = X(t,u,\mu,\volatilityCoefficient_2) : \, u\leq t \leq u\derivative 
 		\right \rbrace
 		\] 
 		are weakly NAP-equivalent, as argued in Proposition \ref{prop.deceptiveNAPindifference}. Therefore we conclude by recalling Corollary \ref{corol.concatenateWeaklyNAPindifferentProcesses}.
 	\end{enumerate}
 \end{proof}

Proposition \ref{prop.interpolatingHistoricalandImpliedVol} achieves the disentanglement between the concept of implied volatility and the concept of historical volatility. Indeed, if we refer to implied volatility as the feature of the price path that is relevant for option pricing, we  see from the statement of Proposition \ref{prop.interpolatingHistoricalandImpliedVol} that such a feature is at a remove from the distributional properties of the physical evolution of the price path. More precisely, with reference to the discussion in Section \ref{sec.introduction}, we see that the implied volatility of the process $Y(\cdot, \volatilityCoefficient_1,\volatilityCoefficient_2)$ is $\volatilityCoefficient_{2}$, but for all $t$ the physical variance $\Variance_{\Prob}(Y(t, \volatilityCoefficient_1,\volatilityCoefficient_2))$ of the $t$-marginal of $Y$ is $\volatilityCoefficient_1\squared t$.

This motivates our pathwise perspective on the mathematical models of option pricing. Such a perspective demands in particular to reconsider the employed  integration theory, as we will explain in the next section.   
\section{Proofs for Section \ref{sec.pathwiseIntegrals}}
\label{sec.proofsPathwiseIntegrals}

\begin{proof}
 	[Proof of Proposition \ref{Prop.[RY99,ChapterIV,(2.21)Exercise]}]
 	The one-dimensional case $d=1$ suffices. 
 	Let $(\onepi_n)_n$ and $(\twopi_n)_n$ be two sequences of partitions with vanishing mesh-size. Define $\tilde{\pi}_{2k+1} := \onepi_k$ and $\tilde{\pi}_{2k} : = \twopi_k$ for $k$ in $\N$. The assumption guarantees that $(\piH[\tilde{\pi}_n].X)_T$, $n\geq 1$,  is a Cauchy sequence for every $H$ in $C([0,T],\R)$. Therefore, the triangulation
 	\begin{align*}
 	\lvert \lim_n (\piH[\onepi_n].X)_T -& \lim_n (\piH[\twopi_n].X)_T \rvert \\ \leq &
 	\abs{\lim_n (\piH[\onepi_n].X)_T -(\piH[\onepi_N].X)_T } \\
 	& + \abs{(\piH[\tilde{\pi}_{2N+1}].X)_T - (\piH[\tilde{\pi}_{2N}].X)_T} \\
 	& + \abs{(\piH[\twopi_N].X)_T - \lim_n (\piH[\twopi_n].X)_T  }
 	\end{align*}
 	yields the first claim. As a consequence, for every $H$ in $C([0,T],\R)$, we have
 	\[
 	\sup\left\lbrace \lvert(\piH.X)_T\rvert : \, \pi \text{ partition of }[0,T] \right\rbrace < \infty.
 	\]
 	But the map $H\mapsto (\piH.X)_T$ is a bounded linear operator on $C([0,T],\R)$ with 
 	\[
 	\abs{(\piH.X)_T} \leq \norm[H]_{\infty} \sum_{u \in \pi} \abs{X\subscriptuuprime}.
 	\]
 	Furthermore, given $\pi$ the integrand
 	\begin{align*}
 	S_t = \left\lbrace 1- \frac{t-\lfloor t \rfloor}{\lfloor t \rfloor\derivative - \lfloor t \rfloor}\right\rbrace 
 	\mathrm{sign} \Big(X_{\lfloor t \rfloor, \lfloor t \rfloor\derivative}\Big)
 	+ \left\lbrace  \frac{t-\lfloor t \rfloor}{\lfloor t \rfloor\derivative - \lfloor t \rfloor}\right\rbrace
 	\mathrm{sign}  \Big(&X_{\lfloor t \rfloor\derivative, \lfloor t \rfloor^{\prime \prime}}\Big), \\
 	& \qquad \qquad 0\leq t \leq T,	
 	\end{align*}
 	is such that 
 	\[
 	({^{\pi}}\! S.X)_T = \sum_{u\in \pi} \abs{X\subscriptuuprime}.
 	\]
 	Therefore, an application of the uniform boundedness principle concludes the proof. 
 \end{proof}

\begin{proof}[Proof of Proposition \ref{prop.SewingLemma}]
	Given a partition $\partition$ of $[s,t]\subset \timeWindow$, let us set 
	\[
	\int_{\partition} \Xi := \sum_{u\in\partition} \Xi_{u,\uprime}.
	\]
	
	We start by showing that, for any pair $\partition$, $\tilde{\partition}$ of partitions of $[s,t]$, it holds
	\begin{equation}\label{Eq.differentPartitions}
	\begin{split}
	\Big\lvert\int_{\partition} \Xi - &\int_{\tilde{\partition}}\Xi\Big\rvert \\
	\leq&
	 2^\gamma \zeta(\gamma) \control(s,t) \norm[\delta\Xi]_{\control, \gamma} \Big(\mathrm{osc}(\control,\abs{\partition})^{\gamma-1} - \mathrm{osc}(\control,\abs{\tilde{\partition}})^{\gamma-1}\Big),
	 \end{split}
	\end{equation}
	where $\zeta(\gamma) := \sum_{n\geq 1} n^{-\gamma}$ is the zeta function, and $\mathrm{osc}(\control,\abs{\pi}):= \sup\lbrace \control(s,t): \, \abs{t-s}\leq \abs{\pi}\rbrace$ is the modulus of continuity of $\control$ on a scale smaller or equal than the mesh-size $\abs{\pi}$.
	
	Let $\pi$ be a partition of $[s,t]\subset[0,T]$ with at least two subintervals and let\[
	m:=\# \lbrace [u,\uprime] \in \pi\rbrace \geq 2
	\]
	denote the number of subintervals of $\pi$. It is easily seen by contradiction that there must exists some internal point $u$ of $\pi$ such that $[u-,u], [u,\uprime] \in \pi$ and
	\[
	\control(u-, \uprime) \leq \frac{2}{m-1}\control(s,t).
	\]
	We estimate
	\begin{align*}
	\abs{\int_{\pi\setminus \lbrace u \rbrace}  \!\!\!\!\! \Xi \quad  - \int_\pi \Xi } = & 
	\abs{\Xi_{u-, \uprime} - \Xi_{u-, u} - \Xi_{u,\uprime}} \\
	\leq & \norm[\delta \Xi]_\gamma \control^{\gamma}(u-, \uprime) \\
	\leq &  \norm[\delta \Xi]_\gamma \frac{2^\gamma}{(m-1)^\gamma}\control^{\gamma}(s,t) .
	\end{align*}
	By iteration we see
	\begin{equation}\label{Eq.maximalInequality}
	\abs{\Xist - \int_\pi \Xi} \leq \norm[\delta \Xi]_\gamma \big[2\control(s,t)\big]^\gamma \sum_{n\geq 1} \frac{1}{n^\gamma}.
	\end{equation}
	Now, if $\tilde{\pi}$ is a partition that refines $\pi$ we have
	\[
	\int_\pi \Xi - \int_{\tilde{\pi}}\Xi = \sum_{u\in \pi}\left\lbrace \Xi\subscriptuuprime - \int_{\tilde{\pi} \cap [u,\uprime]} \Xi \right\rbrace
	\]
	and equation \eqref{Eq.maximalInequality} yields
	\begin{align*}
	\abs{\int_\pi \Xi - \int_{\tilde{\pi}}\Xi} \leq & 
	\sum_{u \in \pi} \norm[\delta \Xi]_\gamma\big(2 \control(u,\uprime)\big)^\gamma \zeta(\gamma) \\
	\leq & 2^\gamma \zeta(\gamma) \norm[\delta \Xi]_\gamma \mathrm{osc}(\control,\abs{\pi})^{\gamma - 1}   \control(s,t).
	\end{align*}
	 The general case for $\pi$, $\tilde{\pi}$ can be reduced to the case where $\tilde{\pi}$ refines $\pi$. This proves \eqref{Eq.differentPartitions} and says that $\int \Xi$ is well-defined and consistent as pointwise limit of $t\mapsto \int _{\pi^n_t } \Xi$ along any sequence $(\pi^n)_n$ of partitions of $[0,T]$ with mesh-sizes $\abs{\pi^n}$ shrinking to zero. Here we have used the notation $\pi^n_t:= (\pi^n \cup \lbrace t \rbrace)\cap [0,t]$. 
	
	In order to get the bound in \eqref{Eq.errorOfIncrements}, we consider the dyadic sequence of partitions of $[s,t]$, i.e.
	$\pi_0 = \lbrace [s,t]\rbrace$ and 
	\[
	\pi_{n+1} = \bigcup_{u\in \pi_n} \Big\lbrace \left[u,\hat{u}\right], \left[\hat{u}, \uprime\right]\Big\rbrace, 
	\qquad n\geq 0,
	\]
	where $\hat{u}:=\inf \lbrace v >u: \, \control(u,v)\geq 2^{-(n+1)}\control(s,t)\rbrace $. We are assuming, without loss of generality, that $(\pi_n)_n$ has vanishing mesh-size, i.e. that $\control$ is strictly increasing, in the sense that $\control(s,t)>0$ if $s<t$.  Notice that by continuity of $\control$ it holds $\control(u,\hat{u})=2^{-(n+1)}\control(s,t)$ and by subadditivity $\control(\hat{u},\uprime)\leq 2^{-(n+1)}\control(s,t)$. Thus, 
	\[
	\int_{\pi_{n+1}} \Xi = \int_{\pi_n} \Xi - \sum_{u\in\pi_n} \delta \Xi_{u,\hat{u}, \uprime} 
	\]
	and 
	\begin{align*}
	\abs{\int_{\pi_{n+1} }\Xi - \int_{\pi_n}\Xi} \leq & 
	\sum_{u\in\pi_n} \norm[\delta \Xi]_\gamma \control^\gamma (u,\uprime) \\
	\leq & \norm[\delta \Xi]_\gamma  \control^\gamma(s,t)\sum_{u\in\pi_n}  2^{-n\gamma} \\
	=& \norm[\delta \Xi]_\gamma  \control^\gamma(s,t) \,  2^{n(1-\gamma)}.
	\end{align*}
	The right hand side is summable in $n$. Hence, 
	\begin{eqnarray*}
		\abs{\int_s^t \Xi - \Xist }  & & \\ 
		& \leq \sum_{n\geq 0} \Big\lvert \left( \int_{\pi_{n+1}} - \int_{\pi_n} \right) \big(\Xi) \Big\rvert  &  \\
		&&\leq \norm[\delta \Xi]_\gamma \frac{\control^\gamma(s,t)}{1-2^{1-\gamma}}.
	\end{eqnarray*}
	This established equation \eqref{Eq.errorOfIncrements}.
	
	Having obtained equation \eqref{Eq.errorOfIncrements}, the continuity of the path $\int\Xi$ follows from the assumption $\lim_{t\downarrow s} \Xist $ $ = 0$.
	
\end{proof}

\begin{proof}[Proof of Proposition \ref{prop.YoungIntegral}]
	Let $\control_X$ and $\control_H$ be respectively the $p$-variation and the $q$-variation controls of $X$ and of $H$. Using additivity of $X$ we see that 
	\[
	H_s\Xst - H_s\Xsu - H_u\Xut = -H_{s,u}\Xut 
	\]
	and 
	\[
	H_t\Xst - H_u\Xsu - H_t\Xut = H_{u,t}\Xut. 
	\]
	Therefore in both cases
	\[
	\abs{\delta \Xi_{s,u,t}} \leq \control_H ^{1/q}\control_X^{1/p} (s,t).
	\]
	This shows the claimed approximate additivity. Moreover with $1/p\derivative = 1- 1/p$ we can estimate
	\begin{equation*}
	\begin{split}
	\sum_{u\in\partition}\abs{H_u X\subscriptuuprime - H_{\uprime}X\subscriptuuprime}
		\leq &
		\left(\sum\abs{H\subscriptuuprime}^{p\derivative} \right)^{1/p\derivative}
		\left(\sum\abs{X\subscriptuuprime}^{p} \right)^{1/p} \\
		\leq & 
		\text{osc}^{\frac{p\derivative - q}{p\derivative}}(H,\abs{\partition}) \, 
		\control_H ^{1/p\derivative}\control_X^{1/p} (0,\timeHorizon) \\
		\longrightarrow & 0 \qquad \text{ as } \abs{\partition} \downarrow 0.
	\end{split}
	\end{equation*}
\end{proof}

\begin{proof}[Proof of Lemma \ref{lemma.approxAdditivityOfCompensatedSummand}]
	Let $\gamma:= 1/q+2/p$ and notice that by the subadditivity of the $1/\gamma$-th power, for every $s\leq u\leq t$, it holds
	\begin{align*}
	\abs{\delta \Xi_{s,u,t}}^{1/\gamma} \leq &
	\Big[ \control_{R^H}^{1/q+1/p}(s,u)\control_{X}^{1/p}(u,t) + \control_{H\derivative}^{1/q}(s,u)\control_{\secondOrderX}^{2/p}(u,t)\Big]^{1/\gamma} \\
	\leq & \big[\control_{R^H}^{1-1/\gamma p }\control_{X}^{1/\gamma p} + \control_{H\derivative}^{1/\gamma q}\control_{\secondOrderX}^{2/\gamma p}\big](s,t),
	\end{align*}
	where $\control_{R^H}$, $\control_{X}$, $\control_{H\derivative}$ and $\control_{\secondOrderX}$ are the variation controls of $R^H$, $X$, $H\derivative$ and $\secondOrderX$ with the appropriate exponents. Since $(1-\frac{1}{\gamma p}) + \frac{1}{\gamma p } = \frac{1}{\gamma q} + \frac{2}{\gamma p } = 1$, the term in the squared brackets is a control.
\end{proof}

\begin{proof}[Proof of Lemma \ref{lemma.qModerationAndControlledPaths} ]
	Let $p*$ and $\control$ be as in the definition of $q$-moderation. Then,
	\begin{equation*}
	\begin{split}
	\lvert 
	\gradx w (t,\Xt) -& \gradx w (s,\Xs) - \Hessianx w (s,\Xs)\Xst
	\rvert \\
	\leq &
	\lvert 
	\gradx w (t,\Xt) - \gradx w (s,\Xt)
	\rvert \\
	& + 
	\lvert
	\gradx w (s,\Xt) - \gradx w (s,\Xs)-\Hessianx w (s,\Xs)\Xst
	\rvert \\
	\leq & 
	\control^{1/p*}(s,t) \\
	& + \Big\lvert
	\int_{0}^{1} \Big[
	\Hessianx w\left(s,(1+y)\Xs + y\Xt \right) - \Hessianx w \left(s,\Xs\right)
	\Big]
	\Xst \, dy
	\Big\rvert \\
	\leq &
	\control^{1/p*}(s,t) \\
	&+\norm[\Hessianx w (s,\cdot)]_{\alpha\text{-H\"ol}, \convexHull X\timeWindow}
	\frac{\control_{X}^{\frac{1+\alpha}{p}}(s,t)}{1+\alpha}.
	\end{split}
	\end{equation*}
	The symbol $\control_{X}$ denotes the $p$-variation control of the path $X$. By assumption $1/q < \alpha /p$ and so $\frac{1+\alpha}{p}p*>1$. This says that the $p*$-power of the right hand side is of bounded variation. 
\end{proof}

\end{document}